\def \VersionLong {}
\def \VersionFinal {}

\ifdefined \VersionLong
	\newcommand{\LongVersion}[1]{\ifdefined\VersionWithComments{\color{red!40!black}#1}\else#1\fi}
	\newcommand{\ShortVersion}[1]{\ifdefined\VersionWithComments{\color{black!40}#1}\fi}
\else
	\newcommand{\LongVersion}[1]{\ifdefined\VersionWithComments{\color{black!40}#1}\fi}
	\newcommand{\ShortVersion}[1]{\ifdefined\VersionWithComments{\color{red!40!black}#1}\else#1\fi}
\fi

\documentclass[a4paper,10pt,envcountsame]{llncs}

\usepackage[utf8]{inputenc}
\usepackage[english]{babel}

 \usepackage{algorithm}
 \usepackage[noend]{algpseudocode} 
 \usepackage{etoolbox}

 \makeatletter
 \newcommand*{\algrule}[1][\algorithmicindent]{%
   \makebox[#1][l]{%
       \hspace*{.2em}
       \vrule height .75\baselineskip depth .25\baselineskip
   }
 }

 \newcount\ALG@printindent@tempcnta
 \def\ALG@printindent{%
    \ifnum \theALG@nested>0
    \ifx\ALG@text\ALG@x@notext
    \else
    \unskip
    \ALG@printindent@tempcnta=1
    \loop
    \algrule[\csname ALG@ind@\the\ALG@printindent@tempcnta\endcsname]%
    \advance \ALG@printindent@tempcnta 1
    \ifnum \ALG@printindent@tempcnta<\numexpr\theALG@nested+1\relax
    \repeat
    \fi
    \fi
 }
 \patchcmd{\ALG@doentity}{\noindent\hskip\ALG@tlm}{\ALG@printindent}{}{\errmessage{failed to patch}}
 \patchcmd{\ALG@doentity}{\item[]\nointerlineskip}{}{}{} 
 \makeatother

\usepackage{subcaption}
\usepackage{paralist} 
\usepackage{xspace}

\newenvironment{ienumeration}
	{\ifdefined\VersionLong\begin{enumerate}\else\begin{inparaenum}[\itshape i\upshape)]\fi}
	{\ifdefined\VersionLong\end{enumerate}\else\end{inparaenum}\fi}

\usepackage{amsmath} 
\usepackage{amssymb} 
\usepackage{multirow}

\usepackage[misc,geometry]{ifsym} 

\ifdefined\VersionWithComments
	\usepackage{draftwatermark}
	\SetWatermarkText{draft}
	\SetWatermarkScale{15}
	\SetWatermarkColor[gray]{0.9}
\fi

\usepackage[svgnames,table]{xcolor}
\definecolor{darkblue}{rgb}{0.0,0.0,0.6}
\definecolor{darkgreen}{rgb}{0, 0.5, 0}
\definecolor{darkpurple}{rgb}{0.7, 0, 0.7}
\definecolor{darkblue}{rgb}{0, 0, 0.7}

\usepackage[
		colorlinks=true,
	\ifdefined \VersionWithComments
		pagebackref=true,
	\fi
		citecolor=darkgreen,
		linkcolor=darkblue,
		urlcolor=darkpurple,
	]{hyperref}
	
\usepackage[capitalise,english,nameinlink]{cleveref} 
\crefname{line}{\text{line}}{\text{lines}} 
\crefname{assumption}{\text{Assumption}}{\text{Assumptions}} 

\usepackage{tikz}
\usetikzlibrary{arrows,automata,positioning}
\tikzstyle{every node}=[initial text=]
\tikzstyle{location}=[rectangle, rounded corners, minimum size=12pt, draw=black, fill=blue!10, inner sep=2pt]
\tikzstyle{final}=[double]
\tikzstyle{accepting}=[final]
\tikzstyle{PTPMOPT}=[,dashed,color=red,semithick]

\definecolor{coloract}{rgb}{0.50, 0.70, 0.30}
\definecolor{colorclock}{rgb}{0.4, 0.4, 1}
\definecolor{colorconst}{rgb}{0.50, 0.20, 0.00}
\definecolor{colordisc}{rgb}{1, 0, 1}
\definecolor{colorloc}{rgb}{0.4, 0.4, 0.65}
\definecolor{colorparam}{rgb}{1, 0.6, 0.0}

\newif\iftikzgnuplot
\tikzgnuplottrue 
\usepackage{gnuplot-lua-tikz}
\usepackage{pgfplotstable}
\pgfplotsset{compat=1.12}
\usepackage{booktabs}

\newcommand{\regionstate}[4][]{
 \node[state,region] (#2) [#1] {
 \begin{tabular}{c}
  #3\\
  #4\\
 \end{tabular}
 }}
 \newcommand{\regioninitialstate}[4][]{
 \node[state,initial,region] (#2) [#1] {
 \begin{tabular}{c}
  #3\\
  #4\\
 \end{tabular}
 }}
 
\newcommand{\regionstatenew}[5][]{
 \node[state,region,#2] (#3) [#1] {
 \begin{tabular}{c}
  #4\\
  #5\\
 \end{tabular}
 }}
\usetikzlibrary{automata,positioning,matrix,shapes.callouts}
\tikzset{
region/.style={
rectangle,
rounded corners,
draw=black,very thick
},
accepting/.style={double distance=2pt}
}

\newcommand{\Z}{{\mathbb{Z}}}

\newcommand{\R}{{\mathbb{R}}}
\newcommand{\Rp}{{\mathbb{R}_{>0}}}

\newcommand{\ttrue}{\mathrm{t{\kern-1.5pt}t}}
\newcommand{\ffalse}{\mathrm{f{\kern-1.5pt}f}}

\newcommand{\Rnn}{\R_{\ge 0}}

\newcommand{\Znn}{\Z_{\ge 0}}

\newcommand{\powerset}[1]{\mathcal{P}({#1})}
\newcommand{\oast}{\circledast}

\makeatletter
\newcommand{\figcaption}[1]{\def\@captype{figure}\caption{#1}}
\newcommand{\tblcaption}[1]{\def\@captype{table}\caption{#1}}
\makeatother

\newif\iftikzgnuplot
\tikzgnuplottrue 
\usepackage{arydshln}
\usepackage{enumitem}
\setlistdepth{20}
\renewlist{itemize}{itemize}{20}
\setlist[itemize]{label=\textbullet}

\newcommand{\ARun}{\mathit{ARuns}}

\newcommand{\Path}{\mathit{Paths}}

\newcommand{\timedomain}{\Rnn}

\newcommand{\zerovalue}[1][C]{\mathbf{0}_{#1}}

\newcommand{\Zones}[1][C,\Rnn]{\mathcal{Z}(#1)}

\newtheorem{mytheorem}{Theorem}

\newtheorem{myremark}[mytheorem]{Remark}

\spnewtheorem*{myproof}{Proof}{\itshape}{\rmfamily}

\ifdefined\VersionWithComments
	\usepackage[colorinlistoftodos,textsize=footnotesize]{todonotes}
\else
	\usepackage[disable]{todonotes}
\fi
\newcommand{\gennote}[3]{\todo[linecolor=#2,backgroundcolor=#2!25,bordercolor=#2]{#3: #1}}
\newcommand{\ea}[1]{\gennote{#1}{blue}{ÉA}}
\newcommand{\ih}[1]{{\gennote{#1}{purple}{IH}}}
\newcommand{\mw}[1]{\gennote{#1}{orange}{MW}}
\newcommand{\instructions}[1]{{\gennote{\bfseries #1}{red}{Instructions}}}

\ifdefined\VersionFinal
\else
	\usepackage[pagewise]{lineno} 
	\linenumbers
	
\fi

\newcommand{\absoluteTime}{t}
\newcommand{\relativeTime}{\tau}

\newcommand{\signalState}{a}

\newcommand{\semiring}{\mathbb{S}}
\newcommand{\semiringBase}{S}
\newcommand{\semiringElem}{s}
\newcommand{\semiringPlus}{\oplus}
\newcommand{\semiringPlusUnit}{e_\oplus}
\newcommand{\semiringTimes}{\otimes}
\newcommand{\semiringTimesUnit}{e_\otimes}
\newcommand{\semiringInside}{(\semiringBase, \semiringPlus, \semiringTimes, \semiringPlusUnit, \semiringTimesUnit)}
\newcommand{\semiringWithInside}{\semiring = \semiringInside}
\newcommand{\dist}{\mathrm{Dist}}
\newcommand{\Vfrom}{V_{\mathrm{from}}}
\newcommand{\Vto}{V_{\mathrm{to}}}
\newcommand{\supInfSemiring}{(\R \disjUnion \{\pm\infty\},\sup,\inf,-\infty,+\infty)}
\newcommand{\tropicalSemiring}{(\R\amalg\{+\infty\},\inf,+,+\infty,0)}

\newcommand{\signal}{\sigma}
\newcommand{\signalInside}[1][n]{\signalState_1^{\relativeTime_1}\signalState_2^{\relativeTime_2}\cdots \signalState_{#1}^{\relativeTime_{#1}}}
\newcommand{\signalWithInside}[1][n]{\signal=\signalInside[#1]}
\newcommand{\duration}[1]{|#1|}
\newcommand{\values}{\mathit{Values}}
\newcommand{\absConcat}{\circ}
\newcommand{\datavariables}{X}
\newcommand{\DVar}{\datavariables}
\newcommand{\dvar}{x}
\newcommand{\DDom}{\mathbb{D}}
\newcommand{\datadomain}{\DDom}
\newcommand{\dConstant}{d}
\newcommand{\datavaluations}[1][\datavariables]{\datadomain^{#1}}
\newcommand{\clockvaluations}[1][\Clock]{(\timedomain)^{#1}}
\newcommand{\signals}[1][\datavaluations]{\mathcal{T}({#1})}

\newcommand{\loc}{l}
\newcommand{\Loc}{L}
\newcommand{\initLoc}{\loc_0}
\newcommand{\InitLoc}{\Loc_0}
\newcommand{\accLoc}{\loc_F}
\newcommand{\AccLoc}{\Loc_F}
\newcommand{\cval}{\nu}
\newcommand{\CVal}{\clockvaluations}
\newcommand{\TSWA}{\mathcal{W}}
\newcommand{\TSWAWithInside}{\TSWA=(\TSA,\costFunc)}
\newcommand{\TSA}{\mathcal{A}}
\newcommand{\TSAWithInside}{\mathcal{A}=(\datavariables,\Loc,\InitLoc,\AccLoc,\Clock,\Transition,\Label)}
\newcommand{\Transition}{\Delta}

\newcommand{\LabelDomain}{\Constraint{\datavariables}{\datadomain}}
\newcommand{\Label}{\Lambda}
\newcommand{\Guard}{\Constraint{\Clock}{\Znn}}
\newcommand{\guard}{g}
\newcommand{\Resets}{\powerset{\Clock}}
\newcommand{\resets}{\rho}
\newcommand{\Constraint}[2]{\Phi(#1, #2)}

\newcommand{\reset}[2]{#1[#2:=0]}
\newcommand{\costFunc}{\kappa}
\newcommand{\dguard}{u}
\newcommand{\DGuard}{\LabelDomain}
\newcommand{\Clock}{C}
\newcommand{\clock}{c}
\newcommand{\disjUnion}{\amalg}
\newcommand{\ClockWithAbs}{\Clock \disjUnion \{\absClock\}}
\newcommand{\ZonesWithAbs}{\Zones[\ClockWithAbs]}

\newcommand{\project}[2]{{#1{\downarrow_{#2}}}}

\newcommand{\WTTSstate}{q}
\newcommand{\WTTSState}{Q}

\newcommand{\InitWTTSState}{\WTTSState_{0}}

\newcommand{\AccWTTSState}{\WTTSState_{F}}
\newcommand{\WTTSTransition}{{\to}}
\newcommand{\WTTSTransitionRel}{\to}
\newcommand{\WTTSWeight}{W}
\newcommand{\WTTS}{\mathcal{S}}
\newcommand{\dvalSeq}[1][]{\overline{a#1}}
\newcommand{\DValSeq}{(\datavaluations)^\oast}
\newcommand{\dvalSeqi}[1]{\overline{a_{#1}}}
\newcommand{\WTTSWithInside}{\WTTS = (\WTTSState,\InitWTTSState,\AccWTTSState,\WTTSTransition,\WTTSWeight)}
\newcommand{\WTTSPath}{\pi}
\newcommand{\pathValue}{\mu}
\newcommand{\traceValue}{\alpha}

\newcommand{\matchSet}{\mathcal{M}}
\newcommand{\trimBegin}{\absoluteTime}
\newcommand{\trimEnd}{\absoluteTime'}

\newcommand{\trimSignal}{\signal([\trimBegin,\trimEnd))}
\newcommand{\trimBeginVar}{\clock_{\mathrm{begin}}}
\newcommand{\trimEndVar}{\clock_{\mathrm{end}}}
\newcommand{\zoneMatch}{\Zones[\{\trimBeginVar,\trimEndVar\}]}

\newcommand{\WSTTS}{\mathcal{S}^{\mathrm{sym}}}
\newcommand{\WSTTSWithInside}{\WSTTS = (\WSTTSState, \WSTTSInitState,\WSTTSAccState,\WSTTSTransition, \WSTTSWeight)}
\newcommand{\WSTTSstate}{q^{\mathrm{sym}}}
\newcommand{\WSTTSState}{Q^{\mathrm{sym}}}
\newcommand{\WSTTSInitState}[1][]{\WSTTSState_{#1 0}}
\newcommand{\WSTTSAccState}[1][]{\WSTTSState_{#1 F}}
\newcommand{\WSTTSTransition}{{\to}^{\mathrm{sym}}}

\newcommand{\zone}{Z}
\newcommand{\WSTTSWeight}{W^{\mathrm{sym}}}
\newcommand{\WSTTSPath}{{\pi^{\mathrm{sym}}}}
\newcommand{\absClock}{T}
\newcommand{\symbolicPathValue}{\mu^{\mathrm{sym}}}
\newcommand{\symbolicTraceValue}{\alpha^{\mathrm{sym}}}

\newcommand{\reach}[1][\WSTTS]{\mathrm{Reach}(#1)}

\newcommand{\weights}{w}
\newcommand{\incrementFunc}{\mathit{incr}}
\newcommand{\partIncrementFunc}{\mathit{incr}_{<}}
\newcommand{\incrementalWeight}{\mathit{weight}}

\newcommand{\matching}[1]{{#1}_{\mathrm{match}}}
\newcommand{\init}[1]{{#1}_{\mathit{init}}}

\ifdefined \VersionWithComments
 	\definecolor{colorok}{RGB}{80,80,150}
\else
	\definecolor{colorok}{RGB}{0,0,0}
\fi

\newcommand{\eg}{\textcolor{colorok}{e.\,g.,}\xspace}
\newcommand{\ie}{\textcolor{colorok}{i.\,e.,}\xspace}

\makeatletter
\AtBeginDocument{%
  \@ifpackageloaded{hyperref}
  {\def\@doi#1{\href{https://doi.org/#1}
      {\ttfamily https://doi.org/#1}\egroup}}
  {\def\@doi#1{\ttfamily https://doi.org/#1\egroup}}
  \def\doi{\bgroup\catcode`\_=12\relax\@doi}}
\makeatother

\title{Online Quantitative Timed Pattern Matching with Semiring-Valued Weighted Automata\thanks{%
	\LongVersion{%
		This is the author (and extended) version of the manuscript of the same name published in the proceedings of the 17th International Conference on Formal Modeling and Analysis of Timed Systems (FORMATS 2019).
	The final version is available at \url{www.springer.com}.
	This version contains additional proofs.
	}%
        Thanks are due to Ichiro Hasuo for a lot of useful comments and Sasinee Pruekprasert for a feedback.
	This work is partially supported
	by
	JST ERATO HASUO Metamathematics for Systems Design Project (No.\ JPMJER1603) and
        by JSPS Grants-in-Aid No.\ 15KT0012 \& 18J22498.
}}
\author{Masaki Waga\inst{1,2,3}\orcidID{0000-0001-9360-7490}}
\date{\today{}}
\institute{%
National Institute of Informatics, Tokyo, Japan
\and
SOKENDAI (The Graduate University for Advanced Studies), Tokyo, Japan
\and
JSPS Research Fellow, Tokyo, Japan
}

\usepackage{graphicx}	
\usepackage{version}	
\begin{document}

\pagestyle{plain}

\maketitle

\setcounter{footnote}{0}

\thispagestyle{plain}

\ifdefined \VersionWithComments
	\textcolor{red}{\textbf{This is the version with comments. To disable comments, comment out line~3 in the \LaTeX{} source.}}
\fi

\begin{abstract}
\emph{Monitoring} of a signal plays an essential role in the runtime verification of cyber-physical systems. \emph{Qualitative timed pattern matching} is one of the mathematical formulations of monitoring, which gives a Boolean verdict for each sub-signal according to the satisfaction of the given specification. There are two orthogonal directions of extension of the qualitative timed pattern matching. One direction on the result is \emph{quantitative}: what engineers want is often not a qualitative verdict but the \emph{quantitative measurement} of the satisfaction of the specification.  The other direction on the algorithm is \emph{online} checking: the monitor returns some verdicts before obtaining the entire signal, which enables to monitor a running system. It is desired from application viewpoints. In this paper, we conduct these two extensions, taking an automata-based approach. This is the first \emph{quantitative} and \emph{online} timed pattern matching algorithm to the best of our knowledge. More specifically, we employ what we call \emph{timed symbolic weighted automata} to specify quantitative specifications to be monitored, and we obtain an online algorithm using the \emph{shortest distance} of a weighted variant of the zone graph and \emph{dynamic programming}. Moreover, our problem setting is \emph{semiring-based} and therefore, general. Our experimental results confirm the scalability of our algorithm for specifications with a time-bound.

\keywords{quantitative monitoring,
timed automata,
weighted automata,
signals,
zones,
dynamic programming,
semirings\LongVersion{,
timed pattern matching,
runtime verification}}
\end{abstract}



\instructions{Submissions should not exceed 15 pages in length (not including the bibliography, which is thus not restricted), but may be supplemented with a clearly marked appendix, which will be reviewed at the discretion of the program committee.}

\ea{hello}
\ih{hello}
\mw{hello}

\section{Introduction}\label{section:introduction}

\paragraph{Background}
\emph{Monitoring} a system behavior plays an essential role in the runtime verification or falsification of \emph{cyber-physical systems (CPSs)}, where various formalisms such as temporal logic formulas or automata are used for \emph{specification}. Usually, a CPS is a \emph{real-time} system, and real-time constraints must be included in the specification.
An example of such a specification is that the velocity of a self-driving car should be more than 70 km/h within 3 s after the car enters an empty highway.
\emph{Timed automata}~\cite{Alur1994} is a formalism that captures real-time constraints. They are equipped with clock variables and timing constraints on the transitions.
Applications of monitoring of real-time properties include data classification~\cite{DBLP:conf/hybrid/BombaraVPYB16} and Web services~\cite{DBLP:conf/sigsoft/RaimondiSE08} as well as CPSs (e.g., automotive systems~\cite{DBLP:conf/rv/KaneCDK15} and
 medical systems~\cite{DBLP:journals/jcse/ChenSWL16}).

The behavior of a CPS is usually described as a \emph{real-valued signal} that is mathematically a function $\sigma$ mapping a time $t$ to the condition $\sigma(t)\in\R^n$ of the system at time $t$.
Usual automata notions (e.g., NFA and timed automata) handle only finite alphabets, and in order to monitor signals over $\R^n$, automata must be extended to handle infinite alphabets. \emph{Symbolic automata}~\cite{DBLP:conf/popl/VeanesHLMB12} handle large or even infinite alphabets, including real vectors.
In a symbolic automaton over a real vector space $\mathbb{R}^n$, each location (or transition) is labeled with a \emph{constraint} over $\mathbb{R}^n$ instead of one vector $v \in\mathbb{R}^n$; therefore, one location (or transition) corresponds to infinitely many vectors.


\begin{table}[tbp]
 \centering
 \caption{Comparison of the problem settings with related studies}
 \label{table:related_works}
\scalebox{0.77}{ \begin{tabular}{c||c|c|c|c}
& Quantitative? & Online? & Dense time? & Result of which part?\\
\hline\hline
\cite{DBLP:conf/formats/BakhirkinFNMA18} & No & Yes & Yes & 
\begin{tabular}{c}
 All sub-signals (pattern matching)\\
\end{tabular}\\\hline
\cite{DBLP:conf/formats/BakhirkinFMU17} & Yes & No & Yes & 
\begin{tabular}{c}
 All sub-signals (pattern matching)\\
\end{tabular}\\\hline
\cite{DBLP:journals/tcad/JaksicBGN18} & Yes & Yes & No & The whole signal\\\hline
\cite{DBLP:conf/rv/DeshmukhDGJJS15} & Yes & Yes & Yes & The whole signal\\\hline
\textbf{This Paper} & \textbf{Yes}& \textbf{Yes}& \textbf{Yes}& 
\begin{tabular}{c}
 \textbf{All sub-signals}
 \textbf{(pattern matching)}\\
\end{tabular}\\
 \end{tabular}}
\end{table}

Monitoring can be formulated in various ways. 
They are classified according to the following criteria. 
\cref{table:related_works} shows a comparison of various formulations\LongVersion{ of monitoring problems}.
\begin{description}
 \item[Qualitative vs. quantitative semantics] 
When an alphabet admits subtraction and comparisons, in addition to the qualitative semantics (\ie{} true or false), one can define a \emph{quantitative} semantics (\eg{} robustness) of a signal with respect to the specification~\cite{DBLP:journals/tcs/FainekosP09,DBLP:conf/formats/DonzeM10,DBLP:conf/cav/AkazakiH15,DBLP:conf/formats/BakhirkinFMU17}. \emph{Robust semantics} shows how robustly a signal satisfies (or violates) the given specification. For instance, the specification $v > 70$ is satisfied more robustly by $v = 170$ than by $v = 70.0001$. In the context of CPSs, \emph{robust semantics} for signal temporal logic is used in robustness-guided falsification~\cite{DBLP:conf/cav/Donze10,DBLP:conf/tacas/AnnpureddyLFS11}. \emph{Weighted automata}~\cite{DBLP:journals/iandc/Schutzenberger61b,Droste:2009:HWA:1667106} are employed for expressing such a quantitative semantics~\cite{DBLP:journals/fmsd/JaksicBGNN18,DBLP:journals/tcad/JaksicBGN18}.
 \item[Offline vs. online]
 Consider monitoring of a signal $\signal = \signal_1 \cdot \signal_2$ over a specification $\TSWA$.
 In offline monitoring, the monitor returns the result $\matchSet(\signal,\TSWA)$ after obtaining the entire signal $\signal$.
In contrast, in online monitoring, the monitor starts returning the result before obtaining the entire signal $\signal$.
For example, the monitor may return a partial result $\matchSet(\signal_1,\TSWA)$ for the first part $\signal_1$ before obtaining the second part $\signal_2$.
 \item[Discrete vs. dense time] 
In a discrete time setting, timestamps are natural numbers while, in a dense time setting, timestamps are positive (or non-negative) real numbers.
 \item[Result of which part?] 
Given a signal $\signal$, we may be interested in the properties of different sets of sub-signals of $\signal$. 
The simplest setting is where we are interested only in the whole signal $\signal$ (\eg{}\cite{DBLP:conf/rv/DeshmukhDGJJS15,DBLP:journals/tcad/JaksicBGN18}). 
Another more comprehensive setting is where we are interested in the property of \emph{each} sub-signal of $\signal$; problems in this setting are called \emph{timed pattern matching}~\cite{DBLP:conf/formats/UlusFAM14,DBLP:conf/formats/WagaAH16,DBLP:conf/formats/BakhirkinFMU17}.
\end{description}

\paragraph{Our problem}

\begin{figure}[tb]
 \begin{minipage}{0.49\linewidth}
  \centering
  \includegraphics[scale=0.37]{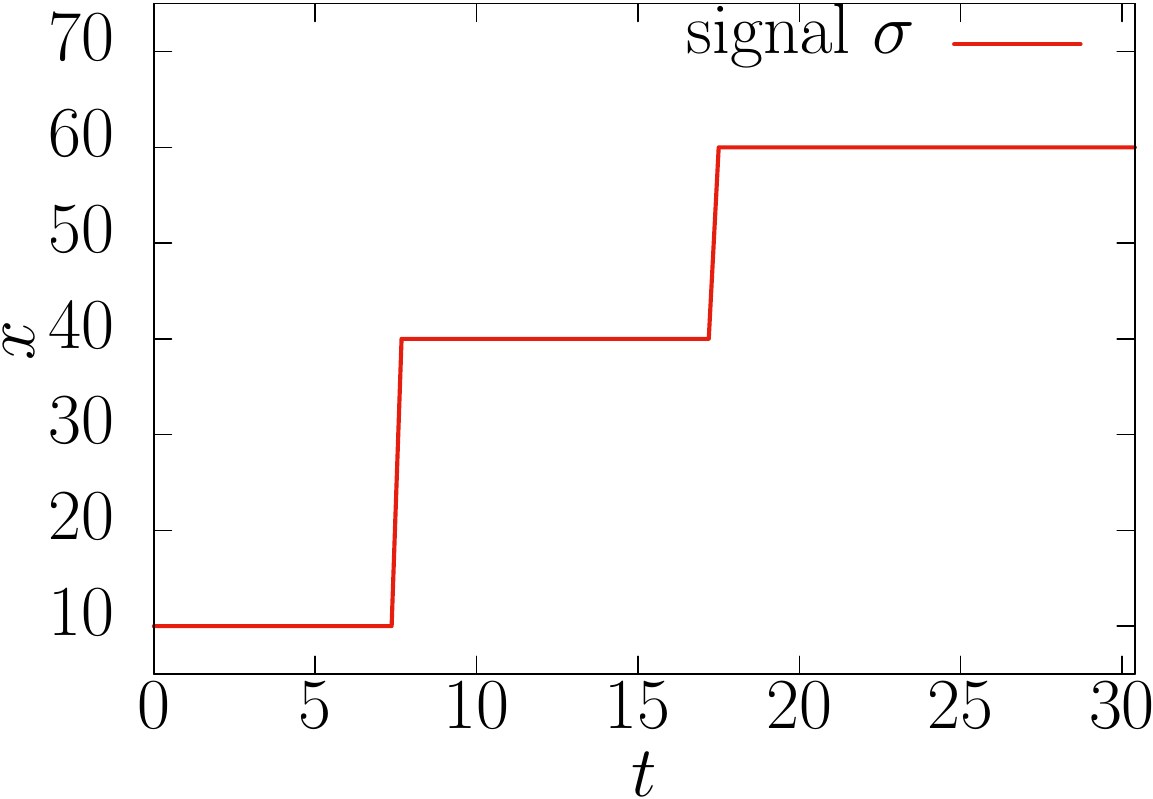}
 \end{minipage}
 \begin{minipage}{0.49\linewidth}
 \centering
 \includegraphics[scale=0.55]{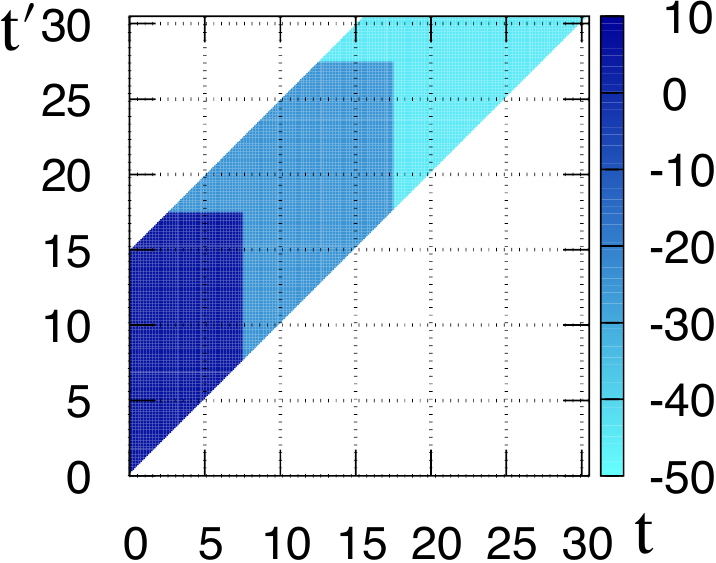}
 \end{minipage}
 \caption{Piecewise-constant signal $\signal$ (left) and an illustration of the quantitative matching function
 $(\matchSet(\signal,\TSWA)) (t,t')$ for $[t,t') \subseteq [0,30.5)$ (right). 
 In the right figure, the score in the white areas is $-\infty$.
 The specification $\TSWA$ is outlined in \cref{ex:intro}. 
 In the right figure, the value at $(3,15)$ is $5$. 
 It shows that the score $\bigl(\matchSet(\signal,\TSWA)\bigr)(3,15)$, for the restriction $\signal\bigl([3,15)\bigr)$ of $\signal$ to the interval $[3,15)$, is $5$.}
 \label{fig:qtpm_example}
\end{figure}
Among the various problem settings of monitoring, we focus on an \emph{online} algorithm for \emph{quantitative timed pattern matching}~\cite{DBLP:conf/formats/BakhirkinFMU17} in a \emph{dense} time setting. 
See \cref{table:related_works}. 
Given a piecewise-constant signal $\signal$ and a specification $\TSWA$ expressed by what we call a \emph{timed symbolic weighted automaton}, our algorithm returns the \emph{quantitative matching function} $\matchSet(\signal,\TSWA)$ that maps each interval $[t,t') \subseteq [0,|\sigma|)$ to the (quantitative) semantics $\bigl(\matchSet(\signal,\TSWA)\bigr)(t,t')$, with respect to $\TSWA$, for the restriction $\signal([t,t'))$ of $\signal$ to the interval $[t,t')$, where $\duration{\signal}$ is the duration of the signal.
An illustration of $\matchSet(\signal,\TSWA)$ is in \cref{fig:qtpm_example}.
In~\cite{DBLP:conf/formats/BakhirkinFMU17}, quantitative timed pattern matching was solved by an offline algorithm using a syntax tree of \emph{signal regular expressions}.
In this paper, we propose an \emph{online} algorithm for quantitative timed pattern matching with automata. To the best of our knowledge, this is the first online algorithm for quantitative timed pattern matching. 
Moreover, our (quantitative) semantics is parameterized by a \emph{semiring} and what we call a \emph{cost function}.
This algebraic formulation makes our problem setting general.

\begin{example}
 \label{ex:intro}
 Let $\signal$ be the piecewise-constant signal in the left of \cref{fig:qtpm_example} and $\TSWA$ be the specification meaning the following.
\begin{itemize}
 \item At first, the value of $x$ stays less than 15, and then the value of $x$ becomes and remains greater than 5 within 5 s.
 \item We are only interested in the behavior within 10 s after the value of $x$ becomes greater than 5.
 \item We want the score showing how robustly the above conditions are satisfied.
\end{itemize}
The right of \cref{fig:qtpm_example} illustrates  the result of quantitative timed pattern matching.
 Quantitative timed pattern matching computes the semantics $\bigl(\matchSet(\signal,\TSWA)\bigr) (t,t')$, with respect to $\TSWA$, for each sub-signal $\signal([t,t'))$ of $\signal$.
 The current semantics shows how robustly the conditions are satisfied.
 The semantics $\bigl(\matchSet(\signal,\TSWA)\bigr)(3,15)$ for the sub-signal $\signal\bigl([3,15)\bigr)$ is $5$, which is the value at $(3,15)$ in the right of \cref{fig:qtpm_example}.
 \LongVersion{ This is because the distance between the first constraint $x < 15$ and the first valuation $x=10$ of the sub-signal $\signal\bigl([3,15)\bigr)$ is $5$, and 
 the distance between the second constraint $x > 5$ and the valuations $x=10$, $x=40$, and  $x=60$ of the sub-signal $\signal\bigl([3,15)\bigr)$ is not smaller than $5$.}
 The semantics $\bigl(\matchSet(\signal,\TSWA)\bigr)(10,15)$ for the sub-signal $\signal\bigl([10,15)\bigr)$ is $-25$, which is the value at $(10,15)$ in the right of \cref{fig:qtpm_example}.
 Thus, the sub-signal $\signal\bigl([3,15)\bigr)$ satisfies the condition specified in $\TSWA$ more robustly than the sub-signal $\signal\bigl([10,15)\bigr)$.

 Our algorithm is \emph{online} and it starts returning the result before obtaining the entire signal $\signal$.
 For example, after obtaining the sub-signal $\signal\bigl([0,7.5)\bigr)$ of the initial 7.5 s, our algorithm returns that for any $[t,t') \subseteq [0,7.5)$, the score
$\bigl(\matchSet(\signal,\TSWA)\bigr)(t,t')$ is $5$.
\end{example}

\paragraph{Our solution}
We formulate quantitative timed pattern matching using the \emph{shortest distance}~\cite{Mohri2009} of semiring-valued (potentially \emph{infinite}) weighted graphs.
We reduce it to the shortest distance of \emph{finite} weighed graphs.
This is in contrast with the qualitative setting: the semantics is defined by the \emph{reachability} in a (potentially \emph{infinite}) graph and it is reduced to the reachability in a \emph{finite} graph.{ The following is an overview.}
 \begin{figure}[tb]
  \centering
 \scalebox{0.7}{ 
 \begin{tikzpicture}[shorten >=1pt,node distance=3.0cm,on grid,auto] 
 \node[state,initial] (q_0) {$\loc_0,x < 15$};
 \node[state,node distance=4.0cm] (q_1)[right=of q_0] {$\loc_1,x > 5$};
 \node[state,accepting,node distance=2.7cm] (q_2)[right=of q_1] {$\loc_2,\top$};

 \path[->] 
 (q_0) edge [above] node {$c < 5$ $/c:=0$} (q_1)
 (q_1) edge [above] node {$c < 10$} (q_2);
 \end{tikzpicture}}
 \small
 \begin{align*}
 \kappa_r\bigl(u,(a_1 a_2\dots a_m)\bigr) &= \inf_{i\in\{1,2,\dots,n\}}\kappa_{r}(u,(a_i))\\
 \kappa_{r}\bigl(\bigwedge_{i = 1}^n (x_i \bowtie_{i} d_i),(a)\bigr) &= \inf_{i\in\{1,2,\dots,n\}}
 \kappa_{r}(x_i \bowtie_{i} d_i,(a)) \ \text{where $\bowtie_{i}\in\{>,\geq,\leq,<\}$}\\
 \kappa_{r}(x \succ d,(a)) &= a(x) - d \quad \text{where $\succ\in\{\geq,>\}$}\\
 \kappa_{r}(x \prec d,(a)) &= d - a(x)\quad \text{where $\prec\in\{\leq,<\}$}
 \end{align*}
 \normalsize
  \caption{Example of a TSWA $\mathcal{W}=(\mathcal{A},\kappa_{r})$ which is the pair of the TSA $\mathcal{A}$ (upper) and the cost function $\kappa_r$ (lower). 
  See \cref{def:tsa} for the precise definition.}
 \label{fig:timed_symbolic_automaton}
 \label{fig:running_example_tsa}
 \end{figure}
\begin{description}
 \item[Problem formulation] 
We introduce \emph{timed symbolic weighted automata (TSWAs)} and define the (quantitative) semantics $\traceValue(\signal,\TSWA)$ of a signal $\signal$ with respect to a TSWA $\TSWA$. 
 Moreover, we define \emph{quantitative timed pattern matching} for a signal and a TSWA. 
A TSWA $\TSWA$ is a pair $(\TSA,\costFunc)$ of a \emph{timed symbolic automaton (TSA)} $\TSA$ --- that we also introduce in this paper --- and a cost function $\costFunc$. 
The cost function $\costFunc$ returns a semiring value at each transition of $\TSA$, and the semiring operations specify how to accumulate such values over time.
This algebraic definition makes our problem general.
\cref{fig:timed_symbolic_automaton} shows an example of a TSWA.
 \item[Algorithm by zones] We give an algorithm for computing our semantics $\traceValue(\signal,\TSWA)$ of a signal $\signal$ by the shortest distance of a \emph{finite} weighted graph. The constructed weighted graph is much like the \emph{zone graph}~\cite{DBLP:conf/ac/BengtssonY03} for reachability analysis of timed automata.
Our algorithm is general and works for any semantics defined on an idempotent and complete semiring. 
(See \cref{example:semirings} later for examples of such semirings.)
\item[Incremental and online algorithms] We present an incremental algorithm for computing the semantics $\traceValue(\signal,\TSWA)$ of a signal $\signal$ with respect to the TSWA $\TSWA$. 
       Based on this incremental algorithm{ for computing $\traceValue(\signal,\TSWA)$},
       we present an online algorithm for quantitative timed pattern matching. 
       To the best of our knowledge, this is the first online algorithm for quantitative timed pattern matching. Our online algorithm for quantitative timed pattern matching works incrementally, much like in \emph{dynamic programming}. \cref{fig:incremental_algorithm} shows an illustration.
\end{description}

\paragraph{Contribution}
We summarize our contributions as follows.
\begin{itemize}
 \item We formulate the semantics of a signal with respect to a TSWA by a shortest distance of a potentially \emph{infinite} weighted graph.
 \item We reduce the above graph to a \emph{finite} weighted graph.
 \item We give an online algorithm for quantitative timed pattern matching.
\end{itemize}

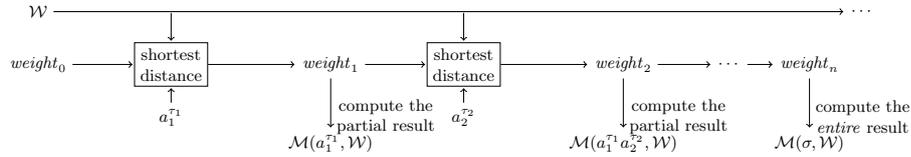
\begin{figure}[tb]
 \centering
 \scalebox{0.7}{
 \begin{tikzpicture}[shorten >=1pt,node distance=3cm,on grid,auto] 
 \node (W) {$\TSWA$};
 \node[node distance=1cm] (weight0)[below=of W] {$\incrementalWeight_0$};
 \node[rectangle,draw,node distance=2.5cm,align=center] (wzc1)[right=of weight0] {
 shortest\\
 distance};
 \node[inner sep=0pt,node distance=2.5cm] (wzc1b)[right=of W] {};
 \node[node distance=1cm] (a1)[below=of wzc1] {$a_1^{\tau_1}$};
 \node[node distance=3cm] (weight1)[right=of wzc1] {$\incrementalWeight_1$};
 \node[rectangle,draw,node distance=2.5cm,align=center] (wzc2)[right=of weight1] {
 shortest\\distance};
 \node[node distance=1cm] (a2)[below=of wzc2] {$a_2^{\tau_2}$};
 \node[inner sep=0pt,node distance=1cm] (wzc2b)[above=of wzc2] {};

 \node[node distance=3cm] (weight2)[right=of wzc2] {$\incrementalWeight_2$};
 \node[node distance=2.0cm] (space1)[right=of weight2] {$\cdots$};
 \node[node distance=7.5cm] (space2)[right=of wzc2b] {$\cdots$};

 \node[node distance=1.5cm] (weightn)[right=of space1] {$\incrementalWeight_n$};
 \node[node distance=1.5cm] (trace_value1)[below=of weight1] {$\matchSet(a_1^{\tau_1},\TSWA)$};
 \node[node distance=1.5cm] (trace_value2)[below=of weight2] {$\matchSet(a_1^{\tau_1}a_2^{\tau_2},\TSWA)$};
 \node[node distance=1.5cm] (trace_valuen)[below=of weightn] {$\matchSet(\signal,\TSWA)$};

 \path[->] 
 (weight0) edge node {} (wzc1)
 (a1) edge node {} (wzc1)
 (wzc1b) edge node {} (wzc1)
 (wzc1) edge node {} (weight1)
 (weight1) edge node {} (wzc2)
 (a2) edge node {} (wzc2)
 (wzc2b) edge node {} (wzc2)
 (wzc2) edge node {} (weight2)
 (weight2) edge node {} (space1)
 (W) edge node {} (space2)
 (space1) edge node {} (weightn)
 (weight1) edge[pos=0.75] node[align=center] {compute the\\partial result} (trace_value1)
(weight2) edge[pos=0.75] node[align=center] {compute the\\partial result} (trace_value2)
 (weightn) edge[pos=0.75] node[align=center] {compute the\\ \emph{entire} result} (trace_valuen)
;
 \end{tikzpicture}}
\caption{Illustration of our online algorithm for quantitative timed pattern matching of a signal $\signalWithInside$ meaning ``the signal value is $a_1$ for $\tau_1$, the signal value is $a_2$ for the next $\tau_2$, $\ldots$'' and a TSWA $\TSWA$. The intermediate data $\incrementalWeight_i$ for the weight computation is represented by zones. The precise definition of the $\incrementalWeight_i$ is introduced later in \cref{def:incremental_weight}.}
\label{fig:incremental_algorithm}
\end{figure}

\paragraph{Related work}

\cref{table:related_works} shows a comparison of the present study with some related studies.
Since the formulation of \emph{qualitative} timed pattern matching~\cite{DBLP:conf/formats/UlusFAM14}, many algorithms have been presented~\cite{DBLP:conf/formats/UlusFAM14,DBLP:conf/tacas/UlusFAM16,DBLP:conf/formats/WagaAH16,DBLP:conf/formats/WagaHS17,DBLP:conf/formats/BakhirkinFNMA18}, including the online algorithms~\cite{DBLP:conf/formats/WagaHS17,DBLP:conf/formats/BakhirkinFNMA18} using timed automata.\LongVersion{ In the consequence, two tools have been presented~\cite{DBLP:conf/cav/Ulus17,DBLP:conf/cpsweek/WagaHS18}.}
\emph{Quantitative} timed pattern matching was formulated and solved by an offline algorithm in~\cite{DBLP:conf/formats/BakhirkinFMU17}.
This offline algorithm is based on the syntax trees of signal regular expressions, and it is difficult to extend for online monitoring.
Weighted automata are used for quantitative monitoring in~\cite{DBLP:conf/sas/ChatterjeeHO16,DBLP:journals/fmsd/JaksicBGNN18,DBLP:journals/tcad/JaksicBGN18}, but the time model was discrete.

The online quantitative monitoring for signal temporal logic~\cite{DBLP:conf/rv/DeshmukhDGJJS15} is one of the closest work. Since we use the clock variables of TSAs to represent the intervals of timed pattern matching, it seems hard to use the algorithm in~\cite{DBLP:conf/rv/DeshmukhDGJJS15} for quantitative timed pattern matching.

\emph{Parametric} timed pattern matching~\cite{DBLP:conf/iceccs/AndreHW18,WA19} is another orthogonal extension of timed pattern matching, where timing constraints are parameterized. Symbolic monitoring~\cite{WAH19} is a further generalization to handle infinite domain data\LongVersion{ \ie{} real values and string labels}. These problems answer feasible parameter valuations and different from our problem.

\paragraph{Organization of the paper}

\cref{sec:preliminary} introduces preliminaries on signals and semirings.
\cref{sec:TSA} defines timed symbolic weighted automata (TSWAs), and our quantitative semantics of signals over a TSWA.
\cref{sec:qtpm} defines the quantitative timed pattern matching problem.
\cref{section:trace_value_algorithm} and \cref{section:qtpm_algorithm} describe our algorithms for computing the quantitative semantics and the quantitative timed pattern matching problem, respectively.
\cref{sec:experiments} presents our experimental results for the sup-inf and tropical semirings, which confirm the scalability of our algorithm under some reasonable assumptions.
\cref{sec:conclusions_and_future_work} presents conclusions and some future perspectives.

\section{Preliminary}
\label{sec:preliminary}

For a set $X$, its powerset is denoted by $\powerset{X}$. 
We use $\varepsilon$ to represent the empty sequence.
All the signals in this paper are piecewise-constant\LongVersion{, which is one of the most common interpolation methods of sampled signals}.

\begin{definition}
 [signal]
 \label{def:signals}
 Let $X$ be a finite set of variables defined over a data domain $\datadomain$.
 A (piecewise-constant) \emph{signal} $\signal$ is a sequence $\signalWithInside$,
 where for each $i \in\{1,2,\dots,n\}$, $\signalState_i \in \datavaluations$ and
 $\tau_i \in \Rp$.
 The set of signals over $\datavaluations$ is denoted by $\signals$.
 The \emph{duration} $\sum_{i=1}^{n}\relativeTime_i$ of a signal $\signal$ is denoted by
 $\duration{\signal}$.
 The sequence $\signalState_1 \circ \signalState_2\circ \dots\circ \signalState_n$ of the values of a signal $\signal$ is denoted by $\values(\signal)$,
 where $\signalState \absConcat \signalState'$ is 
\ShortVersion{$\signalState \absConcat \signalState' = \signalState \signalState'$
 if $\signalState \ne \signalState'$ and
$\signalState \absConcat \signalState' = \signalState$ if $\signalState = \signalState'$.}
\LongVersion{the absorbing concatenation
\begin{displaymath}
 \signalState \absConcat \signalState' =
\begin{cases}
 \signalState \signalState' & \text{if $\signalState \ne \signalState'$}\\
 \signalState    & \text{if $\signalState = \signalState'$}
\end{cases}
\enspace. 
\end{displaymath}}
 We denote the set
 $\{\signalState_1 \absConcat \signalState_2 \absConcat \dots \absConcat \signalState_n \mid 
 n \in\Znn,
 \signalState_1, \signalState_2,\dots, \signalState_n\in \datavaluations\}$ 
 by $(\datavaluations)^\oast$.
 For $\absoluteTime \in [0,\duration{\signal})$, we define $\signal(\absoluteTime) = \signalState_k$, where
 $k$ is such that
 $\sum_{i=1}^{k-1}\relativeTime_i \leq \absoluteTime <\sum_{i=1}^{k}\relativeTime_i$.
 For an interval $[\absoluteTime,\absoluteTime') \subseteq [0,\duration{\signal})$, 
 we define $\signal([\absoluteTime,\absoluteTime')) = \signalState_k^{\sum_{i=1}^{k}\relativeTime_i-\absoluteTime} \signalState_{k+1}^{\relativeTime_{k+1}}\dots \signalState_{l-1}^{\relativeTime_{l-1}}\dots \signalState_l^{\absoluteTime'-\sum_{i=1}^{l-1}\relativeTime_i}$, where
 $k$ and $l$ are such that
 $\sum_{i=1}^{k-1}\relativeTime_i \leq \absoluteTime <\sum_{i=1}^{k}\relativeTime_i$
 and
 $\sum_{i=1}^{l-1}\relativeTime_i < \absoluteTime' \leq\sum_{i=1}^{l}\relativeTime_i$.

\end{definition}

\begin{definition}
 [semiring]
 A system $\semiringWithInside$ is a \emph{semiring} if we have the following.
\begin{itemize}
 \item $(\semiringBase,\semiringPlus,\semiringPlusUnit)$ is a commutative monoid with identity element $\semiringPlusUnit$.
 \item $(\semiringBase,\semiringTimes,\semiringTimesUnit)$ is a monoid with identity element $\semiringTimesUnit$.
 \item For any $\semiringElem,\semiringElem',\semiringElem'' \in \semiringBase$, we have
       $(\semiringElem \semiringPlus \semiringElem') \semiringTimes \semiringElem'' = (\semiringElem \semiringTimes \semiringElem'') \semiringPlus (\semiringElem' \semiringTimes \semiringElem'')$ and
       $\semiringElem \semiringTimes (\semiringElem' \semiringPlus \semiringElem'') = (\semiringElem \semiringTimes \semiringElem') \semiringPlus (\semiringElem \semiringTimes \semiringElem'')$.
 \item For any $\semiringElem \in \semiringBase$, we have
       $\semiringPlusUnit \semiringTimes \semiringElem = \semiringElem \semiringTimes \semiringPlusUnit = \semiringPlusUnit$.
\end{itemize}
\end{definition}

 A semiring $\semiringInside$ is \emph{complete} if 
 for any $\semiringBase' \subseteq \semiringBase$,
 $\bigoplus_{\semiringElem \in \semiringBase'} \semiringElem$ 
 is an element of $\semiringBase$\ShortVersion{ such that: 
 if $\semiringBase' = \emptyset$, $\bigoplus_{\semiringElem \in \semiringBase'} \semiringElem = \semiringPlusUnit$;
 if $\semiringBase' = \{\semiringElem\}$, $\bigoplus_{\semiringElem \in \semiringBase'} \semiringElem = \semiringElem$;
 for any partition $\semiringBase' = \coprod_{i\in I} \semiringBase'_i$, 
 we have $\bigoplus_{\semiringElem \in \semiringBase'} \semiringElem = \bigoplus_{i \in I} \bigl(\bigoplus_{\semiringElem \in \semiringBase'_i} s\bigr)$;
 for any $\semiringElem \in \semiringBase$, 
 we have $\semiringElem \otimes \bigl(\bigoplus_{\semiringElem' \in \semiringBase'} \semiringElem'\bigr) = \bigoplus_{\semiringElem' \in \semiringBase'} (\semiringElem \otimes \semiringElem')$; and 
 for any $\semiringElem \in \semiringBase$, 
 we have $\bigl(\bigoplus_{\semiringElem \in \semiringBase'} \semiringElem\bigr) \semiringTimes \semiringElem' = \bigoplus_{\semiringElem \in \semiringBase'} (\semiringElem \semiringTimes \semiringElem')$.}\LongVersion{ satisfying the following.
 \begin{align*}
  \bigoplus_{\semiringElem \in \semiringBase'} \semiringElem &= \semiringPlusUnit \quad \text{if $\semiringBase' = \emptyset$}
  \qquad\qquad\quad
  \bigoplus_{\semiringElem \in \semiringBase'} \semiringElem = \semiringElem \quad \text{if $\semiringBase' = \{\semiringElem\}$}\\
  \bigoplus_{\semiringElem \in \semiringBase'} \semiringElem &= \bigoplus_{i \in I} 
  \bigl(\bigoplus_{\semiringElem \in \semiringBase'_i} s\bigr) \qquad \text{for any partition $\semiringBase' = \coprod_{i\in I} \semiringBase'_i$}\\
  \semiringElem \otimes \bigl(\bigoplus_{\semiringElem' \in \semiringBase'} \semiringElem'\bigr) &= \bigoplus_{\semiringElem' \in \semiringBase'} (\semiringElem \otimes \semiringElem') \quad \text{and}\quad
  \bigl(\bigoplus_{\semiringElem \in \semiringBase'} \semiringElem\bigr) \semiringTimes \semiringElem' = \bigoplus_{\semiringElem \in \semiringBase'} (\semiringElem \semiringTimes \semiringElem') \qquad \text{for any $\semiringElem \in \semiringBase$}
 \end{align*}}
 A semiring $\semiringWithInside$ is \emph{idempotent} if for any $\semiringElem \in \semiringBase$, 
 $\semiringElem \semiringPlus \semiringElem = \semiringElem$ holds.
For a semiring $\semiringInside$ and
$s_1,s_2,\dots,s_n\in S$, we denote 
$\bigoplus_{i=1}^n s_i = s_1\oplus s_2\oplus\dots\oplus s_n$ and
$\bigotimes_{i=1}^n s_i = s_1\otimes s_2\otimes\dots\otimes s_n$.

\begin{example}
 \label{example:semirings}
 The \emph{Boolean} semiring $(\{\top,\bot\},\lor,\land,\bot,\top)$, the
 \emph{sup-inf} semiring $\supInfSemiring$, and the
 \emph{tropical} semiring $\tropicalSemiring$ are 
 complete and idempotent.
\end{example}

Let $\semiringWithInside$ be a semiring and $G=(V,E,W)$ be a weighted graph over $\semiring$, \ie{} 
$V$ is the finite set of vertices,
$E \subseteq V \times V$ is the finite set of edges, and
$W\colon V \times V \to \semiring$ is the weight function.
For $\Vfrom, \Vto \subseteq V$,
the \emph{shortest distance} from $\Vfrom$ to $\Vto$ is 
\begin{math}
 \dist(\Vfrom,\Vto,V,E,W) =
 \bigoplus_{v\in\Vfrom,v'\in\Vto}\,
 \bigoplus_{v=v_1 v_2 \dots v_n=v' \in \Path(G)} \bigotimes_{i=1}^{n-1} W(v_i,v_{i+1})
\end{math}, where $\Path(G)$ is the set of the paths in\LongVersion{ the directed graph} $G$, \ie{}
$\Path(G) = \{v_1 v_2 \dots v_n \mid \forall i \in \{1,2,\dots,n-1\}.\,, (v_i,v_{i+1}) \in E\}$.
For any complete semiring, the shortest distance problem can be solved by a generalization of the Floyd-Warshall algorithm~\cite{Mohri2009}. 
\LongVersion{Under some conditions, the shortest distance problem can be solved more efficiently by a generalization of the Bellman-Ford algorithm~\cite{Mohri2009}.}

\section{Timed symbolic weighted automata}
\label{sec:TSA}

We propose timed symbolic automata (TSAs), timed symbolic weighted automata (TSWAs), and the (quantitative) semantics of TSWAs.
TSAs are an adaptation of timed automata~\cite{Alur1994} for handling signals over $\datadomain$ rather than signals over a finite alphabet.
In the remainder of this paper, we assume that the data domain $\datadomain$ is equipped with a partial order $\leq$.
A typical example of\LongVersion{ the data domain} $\datadomain$ is the reals $\R$ with the usual order.
We note that TSAs are much like the state-based
variant of timed automata~\cite{EugeneAsarin,DBLP:conf/formats/BakhirkinFNMA18} rather than the
original{, event-based} definition~\cite{Alur1994}.

 For a finite set $\datavariables$ of variables and a poset ($\datadomain,\leq$),
 we denote by $\LabelDomain$ the set of constraints defined by a finite conjunction of inequalities $\dvar \bowtie \dConstant$, where $\dvar \in \DVar$,
 $\dConstant \in \DDom$, and ${\bowtie} \in \{>,\geq,<,\leq\}$.
 We denote $\bigwedge \emptyset \in \LabelDomain$ by $\top$.
 For a finite set $\Clock$ of clock variables,
 a \emph{clock valuation} is a function $\cval \in \CVal$.
For a clock valuation $\cval \in \CVal$ over $\Clock$ and
$\Clock'\subseteq\Clock$, 
we let $\project{\cval}{\Clock'} \in\clockvaluations[\Clock']$ be the clock valuation over $\Clock'$
satisfying $\project{\cval}{\Clock'}(\clock) = \cval(\clock)$
for any $\clock\in\Clock'$.
 For a finite set $\Clock$ of clock variables,
 let $\zerovalue[\Clock]$ be the clock valuation
 $\zerovalue[\Clock] \in \CVal$
 satisfying  $\zerovalue[\Clock](\clock) = 0$ for any $\clock \in \Clock$.
 For a clock valuation $\cval$ over $\Clock$ and $\relativeTime\in\timedomain$, we denote by 
 $\cval + \relativeTime$ the valuation satisfying $(\cval+\relativeTime)(\clock)=\cval(\clock)+\relativeTime$ for any $\clock \in \Clock$.
 For \LongVersion{a clock valuation} $\cval\in\CVal$ and $\resets \subseteq \Clock$, we denote by 
 $\reset{\cval}{\resets}$ the valuation such that
 $(\reset{\cval}{\resets})(x)=0$ for
 $\clock \in \resets$ and
 $(\reset{\cval}{\resets}(\clock)=\cval(\clock)$ for
 $\clock \not\in \resets$.

The definitions of TSAs and TSWAs are as follows.
As shown in \cref{fig:timed_symbolic_automaton}, TSAs are similar to the timed automata in~~\cite{EugeneAsarin,DBLP:conf/formats/BakhirkinFNMA18},
but the locations are labeled with a constraint on the signal values $\datavaluations$ instead of a character in a finite alphabet.

\begin{definition}
 [timed symbolic, timed symbolic weighted automata]
 \label{def:tsa}
 For a poset $(\datadomain,\leq)$, a
 \emph{timed symbolic automaton} (TSA) over $\DDom$ is a
 7-tuple $\TSAWithInside$, where:
\begin{itemize}
 \item $\datavariables$ is a finite set of variables over $\datadomain$;
 \item $\Loc$ is the finite set of locations;
 \item $\InitLoc\subseteq\Loc$ is the set of initial locations;
 \item $\AccLoc\subseteq\Loc$ is the set of accepting locations;
 \item $\Clock$ is the finite set of clock variables;
 \item $\Transition \subseteq \Loc\times\Guard \times\Resets\times\Loc$ is the set of transitions; and
 \item $\Label$ is the labeling function $\Label: \Loc \to \LabelDomain$.
\end{itemize}

 For a poset $(\datadomain,\leq)$ and a complete semiring $\semiringWithInside$,
 a \emph{timed symbolic weighted automaton} (TSWA) over $\DDom$ and $\semiring$ is a pair
 $\TSWAWithInside$ of a TSA $\TSA$ over $\DDom$ and a cost function
 $\costFunc\colon \LabelDomain \times \DValSeq \to \semiringBase$ over $\semiring$.
\end{definition}

The semantics of a TSWA $\TSWAWithInside$ on a signal $\signal$ is defined by the \emph{trace value} $\traceValue(\WTTS)$ of the \emph{weighted timed transition systems (WTTS)} $\WTTS$ of $\signal$ and $\TSWA$.
The trace value $\traceValue(\WTTS)$ depends on the cost function $\costFunc$ and implicitly on its range semiring $\semiring$ as well as the signal $\signal$ and the TSA $\TSA$.
As shown below, the state space of a WTTS $\WTTS$ is $\WTTSState = \Loc\times\clockvaluations\times[0,\duration{\signal}]\times\DValSeq$.
Intuitively, a state $(\loc,\cval,\absoluteTime,\dvalSeq) \in \WTTSState$ of $\WTTS$ consists of: the current location $\loc$; the current clock valuation $\cval$; the current absolute time $\absoluteTime$; and the observed signal value $\dvalSeq$ after the latest transition.
The transition $\WTTSTransition$ of $\WTTS$ is for a transition of $\TSA$ or time elapse.
\mw{We can remove this intuition if we need more space.}

\begin{definition}
 [weighted timed transition systems]
 \label{def:wtts}
 For a signal $\signal\in\signals$ and a TSWA $\TSWAWithInside$ over the data domain $\DDom$ and semiring $\semiring$,
 the \emph{weighted timed transition system (WTTS)} 
 $\WTTS = (\WTTSState,\InitWTTSState,\AccWTTSState,\WTTSTransition,\WTTSWeight)$
 is as follows, where $\TSAWithInside$ is a TSA over $\DDom$ and $\costFunc$ is a cost function over $\semiring$.
 \begin{itemize}
  \item $\WTTSState = \Loc\times\clockvaluations\times[0,\duration{\signal}]\times\DValSeq$
  \item $\InitWTTSState = \{(\initLoc, \zerovalue, 0, \varepsilon) \mid \initLoc \in \InitLoc\}$
  \item $\AccWTTSState = \{ (\accLoc,\cval,\duration{\signal},\varepsilon) \mid \accLoc\in\AccLoc, \cval\in\CVal\}$
  \item $\WTTSTransition \subseteq \WTTSState\times\WTTSState$ is the relation such that $\bigl( (\loc,\cval,\absoluteTime,\dvalSeq), (\loc',\cval',\absoluteTime',\overline{a'}) \bigr) \in \WTTSTransition$ if and only if either of the following holds.
        \begin{description}
         \item[(transition of $\TSA$)] $\exists (\loc,\guard,\resets,\loc')\in\Transition$ satisfying $\cval\models\guard$, $\cval'=\cval[\resets:=0]$, $\absoluteTime'=\absoluteTime$, $\overline{a'} = \varepsilon$, and $\dvalSeq\neq\varepsilon$
         \item[(time elapse)]$\exists \relativeTime\in\Rp$ satisfying $\loc = \loc'$, $\cval'=\cval + \relativeTime$, $\absoluteTime' = \absoluteTime + \relativeTime$, and $\overline{a'} = \dvalSeq \absConcat \values(\signal( [\absoluteTime,\absoluteTime + \relativeTime) ))$
        \end{description}
  \item $\WTTSWeight \bigl( (\loc,\cval,\absoluteTime,\dvalSeq), (\loc',\cval',\absoluteTime',\overline{a'}) \bigr)$ is 
        $\costFunc(\Label(\loc),\dvalSeq)$ if $\overline{a'}=\varepsilon$; and 
        $e_{\otimes}$ if $\overline{a'}\neq\varepsilon$
 \end{itemize}
\end{definition}


\begin{definition}
 [trace value]
 For a WTTS $\WTTSWithInside$, the \emph{trace value} $\traceValue(\WTTS$) is the shortest distance
 $\dist(\InitWTTSState,\AccWTTSState,\WTTSState,\WTTSTransition,\WTTSWeight)$ from $\InitWTTSState$ to $\AccWTTSState$.
\end{definition}

For a signal $\signal$ and a TSWA $\TSWA$, by $\traceValue(\signal,\TSWA)$, we denote the trace value $\traceValue(\WTTS)$ of the WTTS $\WTTS$ of $\signal$ and $\TSWA$.

\begin{example}
 \label{example:cost_functions}
 \label{example:semantics}
 By changing the semiring $\semiring$ and the cost function $\costFunc$, various semantics can be defined by the trace value. Let $\datadomain = \R$.
 For the Boolean semiring $(\{\top,\bot\},\lor,\land,\bot,\top)$ in \cref{example:semirings}, the following function $\costFunc_{b}$ is a prototypical example of a cost function, where $\dguard \in \DGuard$ and 
 $(a_1 a_2\dots a_m) \in \DValSeq$.
 \begin{align*}
 \kappa_{b}\bigl(\dguard,(a_1 a_2\dots a_m)\bigr) &= \bigwedge_{i=1}^{m}\kappa_{b}(\dguard,(a_i))\\
 \kappa_{b}\bigl(\bigwedge_{i = 1}^n (\dvar_i \bowtie_{i} \dConstant_i),(a)\bigr) &= \bigwedge_{i = 1}^n
 \kappa_{b}\bigl(\dvar_i \bowtie_{i} \dConstant_i,(a)\bigr) \quad\text{where $\bowtie_{i}\in\{>,\geq,\leq,<\}$}\\
 \kappa_{b}(\dvar \bowtie \dConstant,(a)) &= 
  \begin{cases}
   \top & \text{if $a \models \dvar \bowtie \dConstant$}\\
   \bot & \text{if $a \not\models \dvar \bowtie \dConstant$}
  \end{cases}
\end{align*}
 For the sup-inf semiring $\supInfSemiring$ in \cref{example:semirings}, 
 the trace value defined by the cost function $\costFunc_{r}$ in \cref{fig:timed_symbolic_automaton} captures the essence of the so-called space robustness~\cite{DBLP:journals/tcs/FainekosP09,DBLP:conf/formats/BakhirkinFMU17}.
 For the tropical semiring $\tropicalSemiring$ in \cref{example:semirings}, an example cost function $\costFunc_{t}$ is as follows.
 \begin{align*}
  \costFunc_{t}\bigl(\dguard,(a_1 a_2\dots a_m)\bigr) &= \sum_{i = 1}^n\costFunc_{r}(\dguard,(a_i))\\
  \costFunc_{t}\bigl(\bigwedge_{i = 1}^n (\dvar_i \bowtie_{i} \dConstant_i),(a)\bigr) &= \sum_{i=1}^n
  \costFunc_{t}(\dvar_i \bowtie_{i} \dConstant_i,(a)) \quad\text{where $\bowtie_{i}\in\{>,\geq,\leq,<\}$}\\
  \costFunc_{t}(\dvar \succ \dConstant,(a)) &= a(\dvar) - \dConstant \quad \text{where $\succ\in\{\geq,>\}$}\\
  \costFunc_{t}(\dvar \prec \dConstant,(a)) &= \dConstant - a(\dvar)\quad \text{where $\prec\in\{\leq,<\}$}
 \end{align*}
\end{example}

\begin{example}
 Let $\TSWAWithInside$ be a TSWA over $\R$ and $\semiring$, where $\TSA$ is the TSA over $\R$ in \cref{fig:timed_symbolic_automaton},
 $\signal$ be the signal
 $\signal=\{x=10\}^{2.5} \{x=40\}^{1.0} \{x=60\}^{3.0}$.
 When $\semiring=\supInfSemiring$ and $\costFunc$ is the cost function $\costFunc_r$ in \cref{example:semantics}, 
 we have $\alpha(\signal,\TSWA) = 5$.
 When $\semiring=\tropicalSemiring$ and $\costFunc$ is the cost function $\costFunc_t$ in \cref{example:semantics}, 
 we have $\alpha(\signal,\TSWA) = 35$.
 \todo{perhaps I should add some intuition}
\end{example}

\section{Quantitative timed pattern matching}
\label{sec:qtpm}

Using TSWAs, we formulate quantitative timed pattern matching as follows.

\begin{definition}
 [quantitative timed pattern matching]
 For a TSWA
 $\TSWA$ over the data domain $\DDom$ and complete semiring $\semiring$,
 and a signal $\signal\in\signals$, the \emph{quantitative matching function}
 $\mathcal{M}(\signal,\TSWA)\colon \mathit{dom}(\signal)\to \semiringBase$ is 
 $(\mathcal{M}(\signal,\TSWA))(\absoluteTime,\absoluteTime') = \traceValue\bigl(\signal\bigl([t,t)\bigr),\TSWA\bigr)$,
 where $\mathit{dom}(\signal)={\{(t,t')\mid 0\le t < t'\leq \duration{\signal}\}}$ and
 $\semiringBase$ is the underlying set of $\semiring$.
 Given a signal $\signal\in\signals$ and 
 a TSWA $\TSWA${ over {the data domain} $\DDom$ and  {complete semiring} $\semiring$},
 the \emph{quantitative timed pattern matching} problem asks for\LongVersion{ the quantitative matching function}
 $\mathcal{M}(\signal,\TSWA)$.
\end{definition}

\begin{example}
  Let
 $\TSWA$ be the TSWA shown in \cref{fig:timed_symbolic_automaton}, 
 which is defined over the reals $\R$ and the sup-inf semiring $\supInfSemiring$, and
 $\signal$ be the signal $\signal=\{x=10\}^{7.5} \{x=40\}^{10.0} \{x=60\}^{13.0}$.
 The quantitative matching function $\mathcal{M}(\signal,\TSWA)$ is as follows. 
 \cref{fig:qtpm_example} shows an illustration.
 \begin{displaymath}
  \bigl(\matchSet(\signal,\TSWA)\bigr)(t,t') =
  \begin{cases}
   5 &\text{when 
   \begin{tabular}{l}
    $t \in [0,7.5), t'\in(0,17.5], t'-t < 10$ or \\
    $t \in [0,7.5), t'\in(10,17.5], t'-t \in[10,15)$
   \end{tabular}}\\
   -25&\text{when 
   \begin{tabular}{l}
    $t \in [7.5,17.5), t'\in(7.5,27.5], t'-t < 10$ or \\
    $t \in [2.5,17.5), t'\in(17.5,27.5], t'-t \in[10,15)$
   \end{tabular}}\\
   -45&\text{when 
   \begin{tabular}{l}
    $t \in [17.5,30.5), t'\in(17.5,30.5], t'-t < 10$ or \\
    $t \in [12.5,30.5), t'\in(27.5,30.5], t'-t \in[10,15)$
   \end{tabular}}\\
  \end{cases}
 \end{displaymath}
\end{example}

Although the domain $\{(t,t')\mid0\leq t< t'\leq |\signal|)\}$ of the quantitative
matching function $\matchSet(\signal,\TSWA)$ is an infinite set, 
$\matchSet(\signal,\TSWA)$ is a piecewise-constant function with finitely many pieces.
Moreover, each piece of $\matchSet(\signal,\TSWA)$ can be represented by a special form of convex polyhedra called \emph{zones}~\cite{DBLP:conf/avmfss/Dill89}.

\begin{definition}[zone]
 For a finite set of clock variables $\Clock$, a \emph{zone} is a $|\Clock|$-dimensional convex polyhedron
 defined by a finite conjunction of the constraints of the form $c \bowtie d$ or 
 $c - c' \bowtie d$, where $c, c' \in C$, ${\bowtie} \in\{>,\geq,\leq,<\}$, and $d \in \R$.
 The set of zones over $\Clock$ is denoted by $\Zones[\Clock]$.
 By a zone $\zone\in\Zones[\Clock]$, we also represent the set 
 $\{\cval \mid \cval\models \zone\} \subseteq \CVal$ of clock valuations.
\end{definition}

\begin{theorem}
 For any TSWA $\TSWA$ over $\DDom$ and $\semiring$ and for any signal $\signal\in\signals$, 
 there is a finite set
 $\{(\zone_1,s_1),(\zone_2,s_2),\dots,(\zone_n,s_n)\} \subseteq \zoneMatch \times \semiringBase$
 such that $\zone_1,\zone_2,\dots,\zone_n$ is a partition of 
 the domain $\{(t,t') \mid 0 \leq t < t'\leq \duration{\signal}\}$, and
 for any $[t,t') \subseteq \timedomain$ satisfying $0 \leq t<t'\leq \duration{\signal}$, 
 there exists 
 $i \in\{1,2,\dots,n\}$ and $\cval \in \zone_i$ satisfying
 $\cval(\trimBeginVar) = \trimBegin$,
 $\cval(\trimEndVar) = \trimEnd$, and
 $(\matchSet(\signal,\TSWA))(\trimBegin,\trimEnd) = s_i$. \qed
\end{theorem}

\section{Trace value computation by shortest distance}
\label{section:trace_value_algorithm}

We present an algorithm to compute the trace values $\traceValue(\WTTS)$.
Since a WTTS possibly has infinitely many states and transitions (see \cref{def:wtts}),
we need a finite abstraction of it. 
We use zone-based abstraction for what we call \emph{weighted symbolic timed transition systems (WSTTSs)}.
In addition to the clock variables in the TSA,
we introduce a fresh clock variable $\absClock$ to represent the absolute time.

\begin{definition}
 [weighted symbolic timed transition system]
 \label{def:weighted_zone_graph}
 For a TSWA $\TSWAWithInside$ over\LongVersion{ the data domain} $\DDom$ and\LongVersion{ complete semiring} $\semiring$, and a signal $\signalWithInside\in\signals$, where $\TSAWithInside$,
 the \emph{weighted symbolic timed transition system} (WSTTS) is \LongVersion{the 5-tuple} $\WSTTSWithInside$ defined as follows.
 \begin{itemize}
  \item $\WSTTSState = \{(\loc,\zone,\dvalSeq) \in \Loc\times\ZonesWithAbs\times\DValSeq\mid \zone\neq\emptyset,\forall \cval\in \zone.\, \cval(\absClock) \leq \duration{\signal}, \dvalSeq = \varepsilon \text{ or }\dvalSeq \absConcat \signal(\cval(\absClock)) = \dvalSeq \}$
  \item $\WSTTSInitState = \{ (\initLoc, \{\zerovalue[\ClockWithAbs]\}, \varepsilon) \mid \initLoc\in\InitLoc\}$
  \item $\WSTTSAccState = \{ (\accLoc, \zone, \varepsilon)\mid \accLoc\in\AccLoc, \exists \cval\in\zone.\, \cval(\absClock) = \duration{\signal}\}$
  \item $\WSTTSTransition \subseteq \WSTTSState\times\WSTTSState$ is the relation such that $\bigl( (\loc,\zone,\dvalSeq), (\loc',\zone', \dvalSeq[']) \bigr) \in \WSTTSTransition$ if and only if one of the following holds.
        \begin{description}
         \item[(transition of $\TSA$)] there exists $(\loc,\guard,\resets,\loc')\in\Transition$, satisfying $\zone'= \{ \cval[\resets:=0] \mid \cval\in\zone, \cval\models\guard \}$, $\dvalSeq\neq\varepsilon$, and $\dvalSeq['] = \varepsilon$.
         \item[(punctual time elapse)] $\loc=\loc'$, $\dvalSeq['] = \dvalSeq\absConcat\values(\signal([\tilde{\cval}(\absClock),\tilde{\cval}'(\absClock))))$, and there is $i\in\{1,2,\dots,n\}$ satisfying
               $\zone' = \{\cval + \relativeTime \mid \cval \in \zone, \relativeTime\in\Rp\} \cap M_{i,=}$,
               where $\tilde{\cval}\in\zone,\tilde{\cval}'\in\zone'$\footnote{The choice of $\tilde{\cval}$ and $\tilde{\cval}'$ does not change $\signal(\tilde{\cval}(T))$ and $\signal(\tilde{\cval}'(T))$ due to the definition of $\WSTTSState$.}, 
               $M_{i,=} = \{ \cval \mid \cval(\absClock) = \sum_{j=0}^{i} \tau_j\}$.
         \item[(non-punctual time elapse)] $\loc=\loc'$, $\dvalSeq['] = \dvalSeq\absConcat\values(\signal([\tilde{\cval}(\absClock),\tilde{\cval}'(\absClock))))$, and there is $i\in\{1,2,\dots,n\}$ satisfying
               $\zone' = \{\cval + \relativeTime \mid \cval \in \zone, \relativeTime\in\Rp\} \cap M_{i}$,
               where $\tilde{\cval}\in\zone,\tilde{\cval}'\in\zone'$, and
               $M_{i} = \{ \cval \mid \sum_{j=0}^{i-1} \tau_j < \cval(\absClock) < \sum_{j=0}^{i} \tau_j\}$.
        \end{description}
  \item $\WSTTSWeight \bigl( (\loc,\zone,\dvalSeq), (\loc',\zone',\dvalSeq[']) \bigr)$ is 
        $\costFunc(\Label(\loc),\dvalSeq)$ if $\dvalSeq[']=\varepsilon$; and 
        $e_{\otimes}$ if $\dvalSeq[']\neq\varepsilon$
 \end{itemize}
\end{definition}

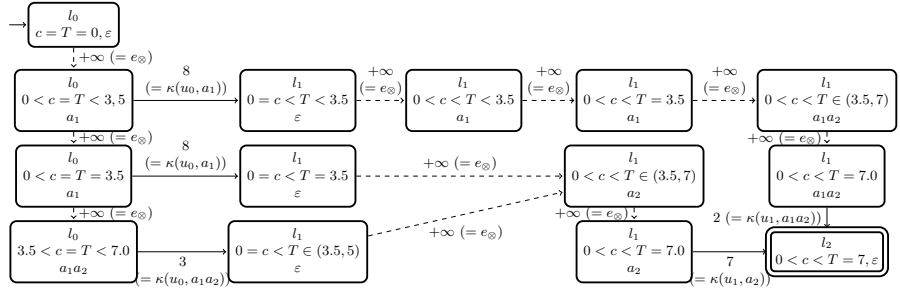
\begin{figure}[tbp]
\scalebox{0.59}{
 \begin{tikzpicture}[shorten >=1pt,node distance=5.0cm,on grid,auto] 
  \regioninitialstate {begin} {$\loc_0$}{$c=T=0, \varepsilon$};
  \regionstatenew[below=of begin] {node distance=1.7cm}{q_00} {$\loc_0$}{$0 < c=T <3,5$\\$a_1$};
  \regionstate[right=of q_00] {q_1f0} {$\loc_1$}{$0=c<T<3.5$\\$\varepsilon$};
  \regionstatenew[below=of q_1f0] {node distance=1.7cm}{q_1f1} {$\loc_1$}{$0=c<T=3.5$\\$\varepsilon$};
  \regionstatenew[right=of q_1f0] {node distance=7.5cm}{q_11} {$\loc_1$}{$0<c<T=3.5$\\$a_1$};
  \regionstatenew[left=of q_11] {node distance=3.8cm}{q_10} {$\loc_1$}{$0<c<T <3.5$\\$a_1$};
  \regionstatenew[below=of q_1f1] {node distance=1.7cm}{q_1f2} {$\loc_1$}{$0=c<T\in(3.5,5)$\\$\varepsilon$};
  \regionstatenew[below=of q_11] {node distance=1.7cm}{q_12} {$\loc_1$}{$0<c<T\in(3.5,7)$\\$a_2$};

  \regionstatenew[below=of q_12] {node distance=1.7cm}{q_13} {$\loc_1$}{$0<c<T=7.0$\\$a_2$};
  \regionstatenew[left=of q_1f1] {}{q_01} {$\loc_0$}{$0<c=T=3.5$\\$a_1$};
  \regionstatenew[left=of q_1f2] {}{q_02} {$\loc_0$}{$3.5<c=T<7.0$\\$a_1a_2$};

  \regionstatenew[right=of q_11] {node distance=4.3cm}{q_1t2} {$\loc_1$}{$0<c<T\in(3.5,7)$\\$a_1a_2$};
  \regionstatenew[right=of q_12] {node distance=4.3cm}{q_1t3} {$\loc_1$}{$0<c<T=7.0$\\$a_1a_2$};

  \regionstatenew[below=of q_1t3] {accepting,node distance=1.7cm}{q_23} {$\loc_2$}{$0<c<T=7,\varepsilon$};

  \path[->]
  (begin)[bend right=0,dashed] edge [right,pos=0.45] node {$+\infty\ (= e_{\otimes})$} (q_00)
  (q_00)[bend right=0,dashed] edge [right,pos=0.45] node {$+\infty\ (= e_{\otimes})$} (q_01)
  (q_01)[bend right=0,dashed] edge [right,pos=0.45] node {$+\infty\ (= e_{\otimes})$} (q_02)
  (q_1f0)[bend right=0,dashed] edge [above] node[align=center] {$+\infty$\\$(= e_{\otimes})$} (q_10)
  (q_1f1)[bend right=0,dashed] edge [above] node {$+\infty\ (= e_{\otimes})$} (q_12)
  (q_1f2)[bend right=0,dashed] edge [below] node[below=0.15cm] {$+\infty\ (= e_{\otimes})$} (q_12)
  (q_10)[bend right=0,dashed] edge [above] node[align=center] {$+\infty$\\$(= e_{\otimes})$} (q_11)
  (q_12)[bend right=0,dashed] edge [left] node {$+\infty\ (= e_{\otimes})$} (q_13)
  (q_11)[bend right=0,dashed] edge [above] node[align=center] {$+\infty$\\ $(= e_{\otimes})$} (q_1t2)
  (q_1t2)[bend right=0,dashed] edge [left] node {$+\infty\ (= e_{\otimes})$} (q_1t3);

  \path[->] 
  (q_00) edge [above] node[align=center] {$8$\\ $(=\costFunc(u_0,a_1))$} (q_1f0)
  (q_01) edge [above] node[align=center] {$8$\\ $(=\costFunc(u_0,a_1))$} (q_1f1)
  (q_02) edge [below] node[align=center] {$3$\\ $(=\costFunc(u_0,a_1a_2))$} (q_1f2)
  (q_1t3)[bend right=0] edge [left] node {$2\ (=\kappa(u_1,a_1a_2))$} (q_23)
  (q_13)[bend right=0] edge [below] node[align=center] {$7$\\ $(=\costFunc(u_1,a_2))$} (q_23);
 \end{tikzpicture}}
 \caption{WSTTS $\WSTTS$ of the TSWA $\TSWA$ in \cref{fig:timed_symbolic_automaton} and the signal $\signal = a_1^{3.5} a_2^{3.5}$, where $u_0=x<15$, $u_1=x>5$, $a_1=\{x=7\}$, and $a_2=\{x=12\}$. The states unreachable from the initial state or unreachable to the accepting state are omitted. The transition for time elapse which can be represented by the composition of other transitions are also omitted. A dashed transition is for the time elapse and a solid transition is for a transition of $\TSA$.}
 \label{fig:step_zone_graph}
\end{figure}
Although the state space $\WSTTSState$ of the WSTTS $\WSTTS$ may be infinite, 
there are only finitely many states reachable from $\WSTTSInitState$ and therefore,
we can construct the reachable part of $\WSTTS$.
See \cref{subsec:finiteness_proof} for the proof.
An example of a WSTTS is shown in \cref{fig:step_zone_graph}.
For a WSTTS $\WSTTS$, we define the \emph{symbolic trace value} $\symbolicTraceValue(\WSTTS)$ as 
the shortest distance
$\dist(\WSTTSInitState,\WSTTSAccState,\WSTTSState,\WSTTSTransition,\WSTTSWeight)$ from $\WSTTSInitState$ to $\WSTTSAccState$.

\begin{theorem}
 \label{corollary:trace_value_correctness}
 Let $\TSWA$ be a TSWA over $\DDom$ and $\semiring$, and
 $\signal\in\signals$ be a signal.
 Let $\WTTS$ and $\WSTTS$ be the WTTS (in \cref{def:wtts}) and WSTTS of $\TSWA$ and $\signal$, respectively.
 If $\semiring$ is idempotent,
 we have $\traceValue(\WTTS) = \symbolicTraceValue(\WSTTS)$.
 \qed
\end{theorem}

Because of \cref{corollary:trace_value_correctness}, we can compute \LongVersion{the trace value} $\traceValue(\WTTS)$ by
\begin{ienumeration}
 \item constructing the reachable part of $\WSTTS$; and
 \item computing the symbolic trace value $\symbolicTraceValue(\WSTTS)$ using an algorithm for the shortest distance problem.
\end{ienumeration}
For example, the symbolic trace value of the WSTTS in \cref{fig:step_zone_graph} is
$\symbolicTraceValue(\WSTTS) = \max\{\min\{8,2\},\min\{8,7\},\min\{3,7\}\} = 7$.
However, this method requires the whole signal to compute the trace value, and it does not suit for the use in online quantitative timed pattern matching.
Instead, we define the \emph{intermediate weight} $\incrementalWeight_i$ and give an incremental algorithm to compute \LongVersion{the trace value} $\traceValue(\WTTS)$.
\LongVersion{Intuitively, for each state $(\loc,\zone,\dvalSeq) \in \WSTTSState$ of the WSTTS $\WSTTS$,
the \emph{intermediate weight} $\incrementalWeight_i$ assign the shortest distance to reach $(\loc,\zone,\dvalSeq)$
by reading the sub-signal $\signalInside[i]$ of $\signalWithInside$.}

\begin{definition}
 [$\incrementFunc$, $\incrementalWeight_i$]
 \label{def:incremental_weight}
 For
 a TSWA $\TSWAWithInside$ over the data domain $\DDom$ and complete semiring $\semiring$, 
 $\signalState \in \datavaluations$, and
 $\absoluteTime \in \Rp$, 
 the \emph{increment function} 
 \[
 \incrementFunc(\signalState,\absoluteTime)\colon \powerset{\Loc\times\ZonesWithAbs\times\DValSeq\times\semiringBase} \to \powerset{\Loc\times\ZonesWithAbs\times\DValSeq\times\semiringBase}  
 \]
 is as follows, 
 where
 $\TSAWithInside$ and
$(\WSTTSState_{\signalState,\absoluteTime}, \WSTTSInitState[\signalState,\absoluteTime,],\WSTTSAccState[\signalState,\absoluteTime,],\WSTTSTransition_{\signalState,\absoluteTime}, \WSTTSWeight_{\signalState,\absoluteTime})$ is the WSTTS of $\signalState^{\absoluteTime}$ and $\TSWA$.
 \begin{align*}
  \incrementFunc(a,\absoluteTime)(\weights) = 
  \{(&\loc',\zone',\dvalSeq['],\semiringElem') \in \Loc\times\ZonesWithAbs\times\DValSeq\times\semiringBase \mid
  \forall \nu'\in\zone'.\, \nu'(\absClock)= \absoluteTime,\\
 \semiringElem' &=
  \bigoplus_{(\loc,\zone,\dvalSeq,\semiringElem)\in w}
  \semiringElem \semiringTimes \dist(\{(\loc,\zone,\dvalSeq)\},\{(\loc',\zone',\dvalSeq['])\},\WSTTSState_{\signalState,\absoluteTime},\WSTTSTransition_{\signalState,\absoluteTime},\WSTTSWeight_{\signalState,\absoluteTime})
  \}
 \end{align*}
 For a TSWA $\TSWA$ over $\DDom$ and $\semiring$, a signal $\signalWithInside$, and $i\in\{1,2,\dots,n\}$, 
 the \emph{intermediate weight} $\incrementalWeight_i$ is defined as follows, where 
 $T_j = \sum_{k=1}^{j}\tau_k$.
 \begin{displaymath}
  \incrementalWeight_i = \bigl(\incrementFunc\bigl(a_i, T_i \bigr) \circ \dots \circ \incrementFunc\bigl(a_{1}, T_{1} \bigr)\bigr) (\{(\initLoc,\{\zerovalue[\ClockWithAbs]\},\varepsilon,\semiringTimesUnit)\mid \initLoc \in\InitLoc\})
 \end{displaymath}
\end{definition}

\begin{algorithm}[t]
 \caption{Incremental algorithm for trace value computation}
 \label{alg:incremental_trace_value}
 \scalebox{0.9}{
 \parbox{1.1\linewidth}{
 \begin{algorithmic}[1]
  \Require A WSTTS $\WSTTSWithInside$ of $\signalWithInside$ and $\TSWA$
  \Ensure $R$ is the symbolic trace value $\symbolicTraceValue(\WSTTS)$
  \State $\incrementalWeight \gets \{(\initLoc,\{\zerovalue[\ClockWithAbs]\},\varepsilon,\semiringTimesUnit)\mid \initLoc \in\InitLoc\}$;\,$R \gets \semiringPlusUnit$
  \Comment{initialize}
  \For {$i \in \{1,2,\dots,n\}$}
  \State $\incrementalWeight \gets \incrementFunc(a_i,T_i)$, where $T_i = \sum_{k=1}^{i}\tau_k$
  \Comment{We have $\incrementalWeight = \incrementalWeight_i$.}
  \EndFor
  \For {$(\loc,\zone,\dvalSeq,\semiringElem)\in\incrementalWeight$}
  \If {$(\loc,\zone,\dvalSeq) \in \WSTTSAccState$}
  \State $R \gets R \semiringPlus \semiringElem$
  \EndIf
  \EndFor
 \end{algorithmic}
 }}
\end{algorithm}
Because of the following\LongVersion{ theorem}, we can incrementally compute the symbolic trace value $\symbolicTraceValue(\WSTTS)$, which is equal to the trace value $\traceValue(\signal,\TSWA)$, by \cref{alg:incremental_trace_value}.

\begin{theorem}
 \label{theorem:incremental_correctness}
 For any WSTTS $\WSTTS$ of a signal $\signalWithInside$ and a TSWA $\TSWA$\LongVersion{ over $\DDom$ and $\semiring$},
\ShortVersion{we have 
\begin{math}
 \symbolicTraceValue(\WSTTS) = 
  \bigoplus_{(\loc,\zone,\dvalSeq) \in \WSTTSAccState}
  \bigoplus_{(\loc,\zone,\dvalSeq,\semiringElem)\in\incrementalWeight_n}
\semiringElem
\end{math}, where $\WSTTSAccState$ is the accepting states of $\WSTTS$.}
\LongVersion{we have the following, where $\WSTTSAccState$ is the accepting states of $\WSTTS$.
 \begin{displaymath}
  \symbolicTraceValue(\WSTTS) = 
  \bigoplus_{(\loc,\zone,\dvalSeq) \in \WSTTSAccState}
  \bigoplus_{(\loc,\zone,\dvalSeq,\semiringElem)\in\incrementalWeight_n}
  \semiringElem
 \end{displaymath}}
\qed
\end{theorem}
\todo{put an example here?}

\section{Online algorithm for quantitative timed pattern matching}
\label{section:qtpm_algorithm}

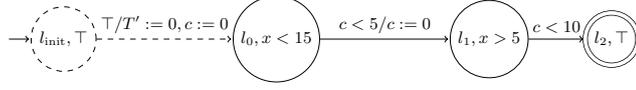
\begin{figure}[tbp]
 \centering
\scalebox{0.70}{
 \begin{tikzpicture}[shorten >=1pt,node distance=4.0cm,on grid,auto] 
 \node[state,initial,dashed] (q_init) {$\loc_{\mathrm{init}},\top$};
 \node[state] (q_0)[right=of q_init] {$\loc_0,x < 15$};
 \node[state] (q_1)[right=of q_0] {$\loc_1,x > 5$};
 \node[state,accepting,node distance=2.3cm] (q_2)[right=of q_1] {$\loc_2,\top$};

 \path[->] 
 (q_init) edge [above,dashed] node {$\top/T':=0,c:=0$} (q_0)
 (q_0) edge [above] node {$c < 5 /c:=0$} (q_1)
 (q_1) edge [above] node {$c < 10$} (q_2);
 \end{tikzpicture}}
 \caption{Matching automaton $\matching{\TSA}$ for the TSA $\TSA$ shown in \cref{fig:running_example_tsa}. The fresh initial location $\loc_{\mathrm{init}}$ and the transition to the original initial location $\loc_0$ are added.}
 \label{fig:example_augumented_tsa}
\end{figure}

In quantitative timed pattern matching, we compute the trace value $\traceValue(\trimSignal,\TSWA)$ for each sub-signal $\trimSignal$.
In order to try matching for each sub-signal $\trimSignal$, we construct the \emph{matching automaton}~\cite{DBLP:conf/formats/BakhirkinFNMA18} $\matching{\TSA}$ from the TSA $\TSA$.
\LongVersion{The matching automaton $\matching{\TSA}$ is constructed by adding a new clock variable $T'$ and a new initial state $\init{\loc}$ to the TSA $\TSA$.
The new clock variable $T'$ represents the duration from the beginning $\trimBegin$ of the sub-signal $\trimSignal$.
The new state $\init{\loc}$ is used to start the sub-signal in the middle of the signal.
We add transitions from $\init{\loc}$ to each initial state $\initLoc$ of $\TSA$, resetting all of the clock variables.}
\cref{fig:example_augumented_tsa} shows an example of $\matching{\TSA}$.
We also define the auxiliary $\partIncrementFunc$ for our online algorithm for quantitative timed pattern matching.

\begin{definition}
 [matching automaton~\cite{DBLP:conf/formats/BakhirkinFNMA18} $\matching{\TSA}$]
 For a TSA $\TSAWithInside$ over $\DDom$, the \emph{matching automaton} is the TSA
 $\matching{\TSA} = (\datavariables, \Loc \disjUnion \{\init{\loc}\},\{\init{\loc}\},\AccLoc,\Clock \disjUnion\{T'\},\Transition',\Label')$ over $\DDom$, 
 where \LongVersion{the transition is}
 $\Transition'=\Transition \disjUnion \{(\init{\loc},\top,\Clock\amalg\{T'\},\initLoc)\mid \initLoc \in \InitLoc\}$, 
 \LongVersion{the labeling function is}
 $\Label'(\init{\loc}) = \top$, and
 $\Label'(\loc) = \Label(\loc)$ for $\loc\in \Loc$.
\end{definition}

\begin{algorithm}[t]
 \caption{Online algorithm for quantitative timed pattern matching}
 \label{alg:qtpm}
 \scalebox{0.9}{
 \parbox{1.1\linewidth}{
 \begin{algorithmic}[1]
  \Require A signal $\signalWithInside$ and a TSWA $\TSWAWithInside$
  \Ensure $M$ is the quantitative matching function $\mathcal{M}(\signal,\TSWA)$.
  \State $\matching{\TSA} \gets$ the matching automaton of $\TSA$ \label{algline:constr_matching_autom}
  \State $\incrementalWeight \gets \{(\initLoc,\{\zerovalue[\Clock \disjUnion \{T,T'\}]\},\varepsilon,\semiringTimesUnit)\mid \initLoc \in\InitLoc\}$;\, for each $[t,t') \subseteq [0,|\signal|)$, $M(t,t') \gets \semiringPlusUnit$
  \For{$i \in \{1,2,\dots,n\}$}
  \State $\incrementalWeight \gets (\partIncrementFunc(a_i, T_i))(\incrementalWeight)$, where $T_i = \sum_{k=1}^{i}\tau_k$ \label{algline:compute_part_weight}
  \For {$(\loc,\zone,\varepsilon,\semiringElem) \in \incrementalWeight, \cval \in \zone$}
  \If {$\loc\in \AccLoc$}
  \State $M(\cval(T')-\cval(T),\cval(T')) \gets M(\cval(T')-\cval(T),\cval(T')) \oplus \semiringElem$. \label{algline:add_weigh}
  \EndIf
  \EndFor
  \State $\incrementalWeight \gets (\incrementFunc(a_i, T_i))(\incrementalWeight)$, where $T_i = \sum_{k=1}^{i}\tau_k$ \label{algline:compute_incremental_weight}
  \EndFor
 \end{algorithmic}
 }}
\end{algorithm}


\begin{definition}
 [$\partIncrementFunc$]
 For
 a TSWA $\TSWAWithInside$ over the data domain $\DDom$ and complete semiring $\semiring$, 
 $\signalState \in \datavaluations$, and
 $\absoluteTime \in \Rp$,
 \LongVersion{the partial increment function} \ShortVersion{$\partIncrementFunc(\signalState,\absoluteTime)$}
\LongVersion{ \[
 \partIncrementFunc(\signalState,\absoluteTime)\colon \powerset{\Loc\times\ZonesWithAbs\times\DValSeq\times\semiringBase} \to \powerset{\Loc\times\ZonesWithAbs\times\DValSeq\times\semiringBase}  
 \]}
 is as follows, 
 where
 $\TSAWithInside$ 
and $(\WSTTSState_{\signalState,\absoluteTime}, \WSTTSInitState[\signalState,\absoluteTime,],\WSTTSAccState[\signalState,\absoluteTime,],\WSTTSTransition_{\signalState,\absoluteTime}, \WSTTSWeight_{\signalState,\absoluteTime})$ is the WSTTS of the TSWA $\TSWA$ and the constant signal $\signalState^{\absoluteTime}$.
{\small \begin{align*}
  \partIncrementFunc(a,\absoluteTime)(\weights) = 
  \{(&\loc',\zone',\dvalSeq['],\semiringElem') \in \Loc\times\ZonesWithAbs\times\DValSeq\times\semiringBase \mid
  \forall \nu'\in\zone'.\, \nu'(\absClock)< \absoluteTime,\\
 \semiringElem' &=
  \bigoplus_{(\loc,\zone,\dvalSeq,\semiringElem)\in w}
  \semiringElem \semiringTimes \dist(\{(\loc,\zone,\dvalSeq)\},\{(\loc',\zone',\dvalSeq['])\},\WSTTSState_{\signalState,\absoluteTime},\WSTTSTransition_{\signalState,\absoluteTime},\WSTTSWeight_{\signalState,\absoluteTime})
  \}
 \end{align*}}
\end{definition}

\cref{alg:qtpm} shows our online algorithm for quantitative timed pattern matching.
We construct the {matching automaton} $\matching{\TSA}$ from the TSA $\TSA$ (\cref{algline:constr_matching_autom}), and
we try matching by reading each constant sub-signal $\signalState_i^{\relativeTime_i}$ of the signal $\signalWithInside$ much like the illustration in \cref{fig:incremental_algorithm}.
For each $i$, first,
we consume a prefix $\signalState_i^{\relativeTime'_i}$ of $\signalState_i^{\relativeTime_i} = \signalState_i^{\relativeTime'_i} \signalState_i^{\relativeTime''_i}$ and 
update\LongVersion{ the intermediate weight} $\incrementalWeight$ (\cref{algline:compute_part_weight}). 
Then, we update the result $M$ for each\LongVersion{ weight} $(\loc,\zone,\varepsilon,\semiringElem) \in \incrementalWeight$\LongVersion{ labelled with an accepting location}\ShortVersion{ if $\loc \in \AccLoc$} (\cref{algline:add_weigh}).
Finally, we consume the remaining part $\signalState_i^{\relativeTime''_i}$ and 
update\LongVersion{ the intermediate weight} $\incrementalWeight$ (\cref{algline:compute_incremental_weight}).

\paragraph{Complexity discussion}

In general, the time and space complexities of \cref{alg:qtpm} are polynomial to the length $n$ of the signal $\signalWithInside$ due to the bound of the size of the reachability part of the WSTTS.
On the other hand, if the TSWA has a time-bound and the sampling frequency of the signal is also bounded (such as in \cref{fig:overshoot_pattern,fig:ringing_pattern}), time and space complexities are linear and constant to the length $n$ of the signal, respectively.

\section{Experiments}
\label{sec:experiments}
We implemented our online algorithm for quantitative timed pattern matching in C++ and conducted experiments to answer the following research questions. \LongVersion{We suppose that the input piecewise-constant signals are interpolations of the actual signals by sampling.}
\begin{description}
 \item[RQ1] Is the practical performance of \cref{alg:qtpm} realistic?
 \item[RQ2] Is \cref{alg:qtpm} online capable, \ie{} does it perform in linear time and  constant space, with respect to the number of the entries in the signal?
 \item[RQ3] Can \cref{alg:qtpm} handle denser logs, \ie{} what is the performance with respect to the sampling frequency of the signal?
\end{description}
Our implementation is in \url{https://github.com/MasWag/qtpm}.
We conducted the experiments on an Amazon EC2 c4.large instance (2 vCPUs and 3.75 GiB RAM)\LongVersion{ running Ubuntu 18.04 LTS (64 bit)}.\LongVersion{ We compiled the implementation by GCC-4.9.3.} For the measurement of the execution time and memory usage,  we used GNU time and took an average of 20 executions. 
We could not compare with~\cite{DBLP:conf/formats/BakhirkinFMU17} because their implementation is not publicly available.

As the complete semiring $\semiring$,
we used the sup-inf semiring $\supInfSemiring$ and the tropical semiring $\tropicalSemiring$ in \cref{example:semirings}.
We used the cost functions $\costFunc_r$ in \cref{example:cost_functions} for the sup-inf semiring, and
$\costFunc_t$ in \cref{example:cost_functions} for the tropical semiring.
\LongVersion{\paragraph{Benchmarks}}
We used the automotive benchmark problems shown in \cref{fig:overshoot_pattern,fig:ringing_pattern,fig:overshoot_unbounded_pattern}.
A summary of quantitative timed pattern matching is on the right of each figure.
The specified behaviors in the TSWAs are taken from ST-Lib~\cite{kapinski2016st} and known to be useful for automotive control applications.
See \cref{appendix:comparizon_performance_benchmarks} for a performance comparison among the benchmarks.

%

 \begin{figure}[tbp]
 \begin{minipage}{0.5\linewidth}
  \centering
  \scalebox{0.60}{
  \begin{tikzpicture}[shorten >=1pt,node distance=4.5cm,on grid,auto]
 \node[state,initial,align=center] (s_0) {$v_{\mathit{ref}}<35$\\ $|v - v_{\mathit{ref}}|<10$};
 \node[state,node distance=3.5cm,align=center] (s_1) [right of=s_0] {$v_{\mathit{ref}}>35$\\ $|v - v_{\mathit{ref}}|>10$};
   \node[state,accepting,node distance=3.0cm] (s_2) [right of=s_1] {$\top$};
   
 \path[->] 
  (s_0) edge [above] node {$c < 10$} (s_1)
  (s_1) edge [above] node {$c < 150$} (s_2);
  \end{tikzpicture}}
\end{minipage}
 \begin{minipage}{0.5\linewidth}
  \centering
 \includegraphics[scale=0.330]{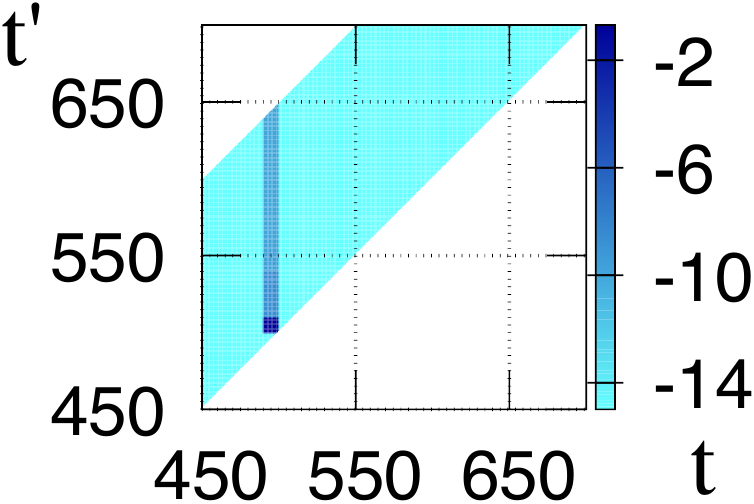}
 \includegraphics[scale=0.330]{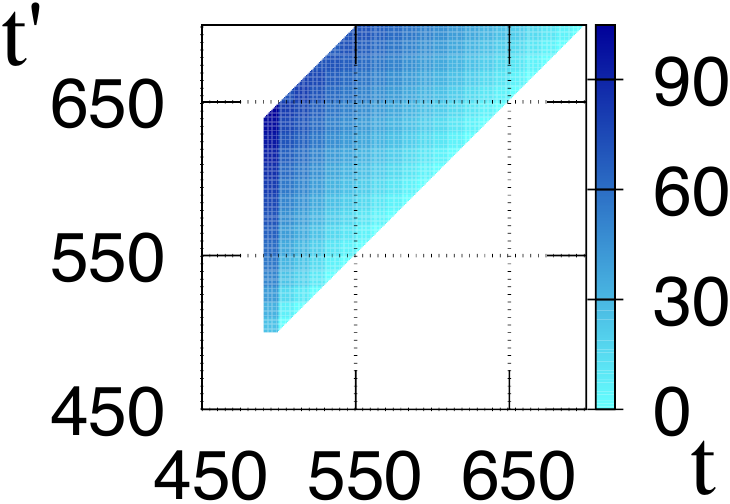} 
 \end{minipage}
  \caption{\textsc{Overshoot}: The set of input signals is generated by the cruise control model~\cite{cruise-control-simulink}.
  The TSA is for the settling when the reference value of the
  velocity is changed from $v_{\mathit{ref}} < 35$ to $v_{\mathit{ref}} > 35$.
  The left and right maps are for the sup-inf and tropical semirings, respectively.}
  \label{fig:overshoot_pattern}
 \begin{minipage}{0.5\linewidth}
  \centering
  \scalebox{0.65}{
  \begin{tikzpicture}[shorten >=1pt,node distance=4.5cm,on grid,auto]
 \node[state,initial] (rise) {$\mathsf{rise}$};
 \node[state] (after_rise) at (1.5,0.5) {$\top$};
 \node[state,node distance=3.2cm] (fall) [right of=rise] {$\mathsf{fall}$};n
 \node[state] (after_fall) at (1.5,-0.5) {$\top$};
 \node[state,accepting,node distance=2.0cm] (fin) [right of=fall] {$\top$};
   
 \path[->] 
  (rise) edge [above left,pos=0.2] node[above=0.2cm] {
   \begin{tabular}{c}
    $c_1 < 20$\\
    $c_2 < 80$
   \end{tabular}} (after_rise)
  (after_rise) edge [above right,pos=0.8] node[above=0.2cm] {
   \begin{tabular}{c}
    $c_1 < 20$\\
    $c_2 < 80$
   \end{tabular}} (fall)
  (fall) edge [below right,pos=-0.3] node[below=0.3cm] {
   \begin{tabular}{c}
    $c_1 < 20,c_2 < 80$\\
    $/c_1:=0$
   \end{tabular}} (after_fall)
  (after_fall) edge [below left,pos=0.8] node[below=0.2cm] {
   \begin{tabular}{c}
    $c_1 < 20$\\
    $c_2 < 80$
   \end{tabular}} (rise)
  (fall) edge [above] node {
   \begin{tabular}{c}
    $c_1 < 20$\\
    $c_2 < 80$
   \end{tabular}} (fin);
  \end{tikzpicture}}
\end{minipage}
 \begin{minipage}{0.5\linewidth}
  \centering
 \includegraphics[scale=0.330]{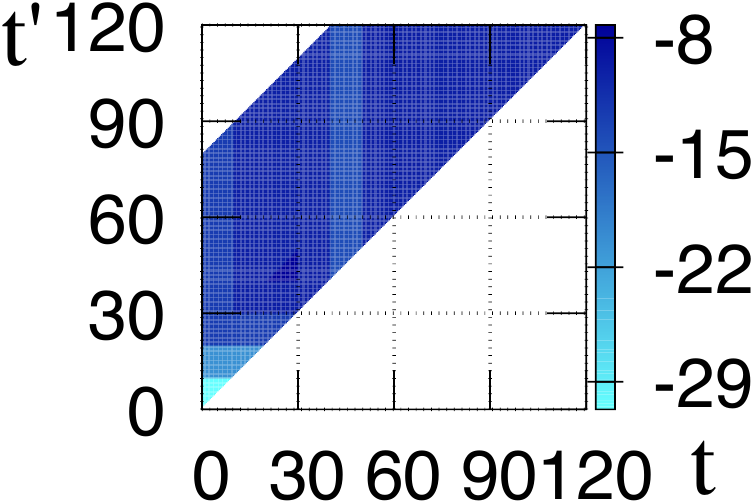}
 \includegraphics[scale=0.330]{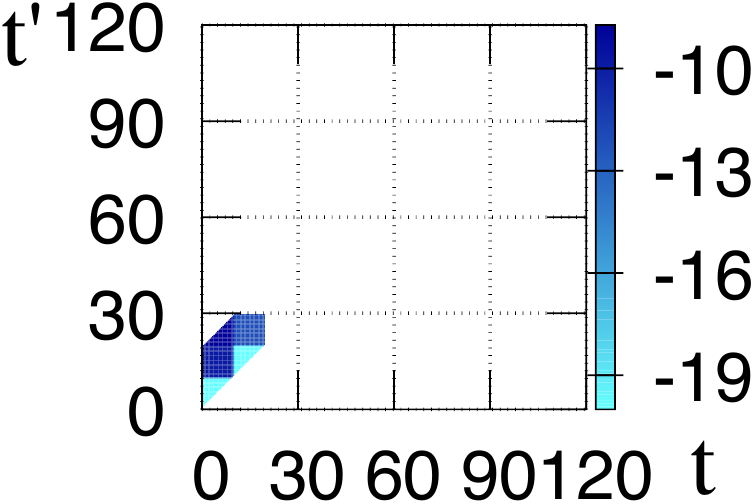}
 \end{minipage}
  \caption{\textsc{Ringing}: The set of input signals is generated by the same model~\cite{cruise-control-simulink} as that in \textsc{Overshoot}.
  The TSA is for the frequent rise and fall of the signal in 80 s.
  The constraints $\mathsf{rise}$ and $\mathsf{fall}$ are $\mathsf{rise} =v(t) - v(t-10)>10$ and $\mathsf{fall} =v(t) - v(t-10)<-10$.
  The left and right maps are for the sup-inf and tropical semirings, respectively.}
  \label{fig:ringing_pattern}
 \begin{minipage}{0.5\linewidth}
  \centering
  \scalebox{0.65}{
  \begin{tikzpicture}[shorten >=1pt,node distance=4.5cm,on grid,auto]
 \node[state,initial] (s_0) {
   \begin{tabular}{c}
    $v_{\mathit{ref}}<35$\\
    $|v - v_{\mathit{ref}}|<10$\\
   \end{tabular}};
 \node[state,node distance=3.5cm] (s_1) [right of=s_0] {
   \begin{tabular}{c}
    $v_{\mathit{ref}}>35$\\
    $|v - v_{\mathit{ref}}|>10$\\
   \end{tabular}};
   \node[state,accepting,node distance=2.3cm] (s_2) [right of=s_1] {$\top$};
   
 \path[->] 
  (s_0) edge [above] node {$c < 10$} (s_1)
  (s_1) edge [above,color=red] node {\color{red} $\top$} (s_2);
  \end{tikzpicture}}
\end{minipage}
  \begin{minipage}{0.5\linewidth}
  \centering
  \includegraphics[scale=0.330]{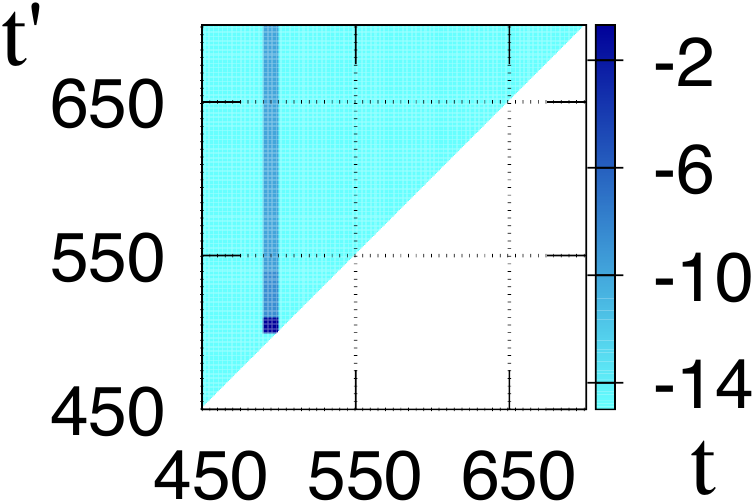}
  \includegraphics[scale=0.330]{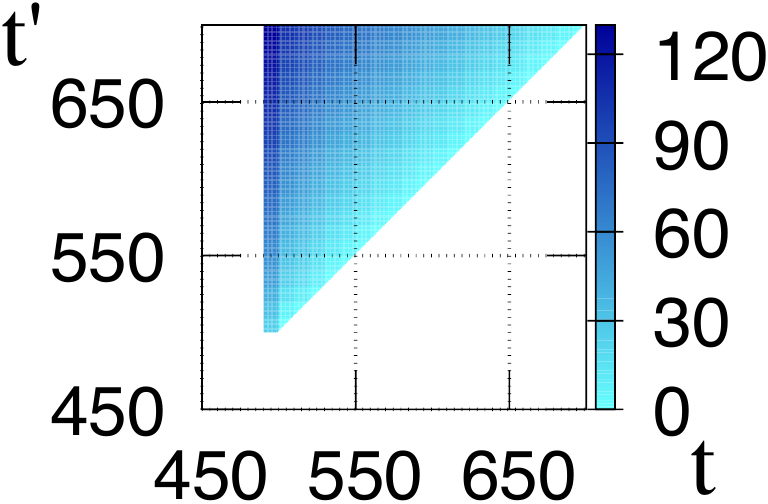}
 \end{minipage}
  \caption{\textsc{Overshoot (Unbounded)}: The set of input signals is generated by the same model~\cite{cruise-control-simulink} as that in \textsc{Overshoot}.
  The TSA is almost the same as that in \textsc{Overshoot}, but the time-bound ($c < 150$) is removed.
 The left and right maps are for the sup-inf and tropical semirings, respectively.}
  \label{fig:overshoot_unbounded_pattern}
\begin{minipage}{0.49\textwidth} 
 \centering
 \scalebox{0.40}{\begin{tikzpicture}[gnuplot]
\tikzset{every node/.append style={font={\fontsize{14.0pt}{16.8pt}\selectfont}}}
\path (0.000,0.000) rectangle (12.500,8.750);
\gpcolor{color=gp lt color border}
\gpsetlinetype{gp lt border}
\gpsetdashtype{gp dt solid}
\gpsetlinewidth{1.00}
\draw[gp path] (1.850,1.379)--(2.030,1.379);
\draw[gp path] (11.725,1.379)--(11.545,1.379);
\node[gp node right,font={\Large}] at (1.592,1.379) {$0$};
\draw[gp path] (1.850,2.370)--(2.030,2.370);
\draw[gp path] (11.725,2.370)--(11.545,2.370);
\node[gp node right,font={\Large}] at (1.592,2.370) {$20$};
\draw[gp path] (1.850,3.362)--(2.030,3.362);
\draw[gp path] (11.725,3.362)--(11.545,3.362);
\node[gp node right,font={\Large}] at (1.592,3.362) {$40$};
\draw[gp path] (1.850,4.353)--(2.030,4.353);
\draw[gp path] (11.725,4.353)--(11.545,4.353);
\node[gp node right,font={\Large}] at (1.592,4.353) {$60$};
\draw[gp path] (1.850,5.344)--(2.030,5.344);
\draw[gp path] (11.725,5.344)--(11.545,5.344);
\node[gp node right,font={\Large}] at (1.592,5.344) {$80$};
\draw[gp path] (1.850,6.335)--(2.030,6.335);
\draw[gp path] (11.725,6.335)--(11.545,6.335);
\node[gp node right,font={\Large}] at (1.592,6.335) {$100$};
\draw[gp path] (1.850,7.327)--(2.030,7.327);
\draw[gp path] (11.725,7.327)--(11.545,7.327);
\node[gp node right,font={\Large}] at (1.592,7.327) {$120$};
\draw[gp path] (1.850,8.318)--(2.030,8.318);
\draw[gp path] (11.725,8.318)--(11.545,8.318);
\node[gp node right,font={\Large}] at (1.592,8.318) {$140$};
\draw[gp path] (1.850,1.379)--(1.850,1.559);
\draw[gp path] (1.850,8.318)--(1.850,8.138);
\node[gp node center,font={\Large}] at (1.850,0.948) {$0$};
\draw[gp path] (3.496,1.379)--(3.496,1.559);
\draw[gp path] (3.496,8.318)--(3.496,8.138);
\node[gp node center,font={\Large}] at (3.496,0.948) {$1$};
\draw[gp path] (5.142,1.379)--(5.142,1.559);
\draw[gp path] (5.142,8.318)--(5.142,8.138);
\node[gp node center,font={\Large}] at (5.142,0.948) {$2$};
\draw[gp path] (6.788,1.379)--(6.788,1.559);
\draw[gp path] (6.788,8.318)--(6.788,8.138);
\node[gp node center,font={\Large}] at (6.788,0.948) {$3$};
\draw[gp path] (8.433,1.379)--(8.433,1.559);
\draw[gp path] (8.433,8.318)--(8.433,8.138);
\node[gp node center,font={\Large}] at (8.433,0.948) {$4$};
\draw[gp path] (10.079,1.379)--(10.079,1.559);
\draw[gp path] (10.079,8.318)--(10.079,8.138);
\node[gp node center,font={\Large}] at (10.079,0.948) {$5$};
\draw[gp path] (11.725,1.379)--(11.725,1.559);
\draw[gp path] (11.725,8.318)--(11.725,8.138);
\node[gp node center,font={\Large}] at (11.725,0.948) {$6$};
\draw[gp path] (1.850,8.318)--(1.850,1.379)--(11.725,1.379)--(11.725,8.318)--cycle;
\node[gp node center,rotate=-270,font={\fontsize{14.0pt}{16.8pt}\selectfont}] at (0.387,4.848) {Execution time [s]};
\node[gp node center,font={\fontsize{14.0pt}{16.8pt}\selectfont}] at (6.787,0.302) {Number of entries of the signal [$\times 10,000$]};
\node[gp node right,font={\fontsize{14.0pt}{16.8pt}\selectfont}] at (9.739,7.922) {\textsc{Overshoot}, sup-inf};
\gpcolor{rgb color={0.580,0.000,0.827}}
\gpsetlinewidth{5.00}
\draw[gp path] (9.997,7.922)--(11.209,7.922);
\draw[gp path] (2.838,1.458)--(3.825,1.537)--(4.813,1.619)--(5.800,1.698)--(6.788,1.779)%
  --(7.775,1.860)--(8.763,1.944)--(9.750,2.024)--(10.738,2.103)--(11.725,2.186);
\gpsetpointsize{8.00}
\gppoint{gp mark 2}{(2.838,1.458)}
\gppoint{gp mark 2}{(3.825,1.537)}
\gppoint{gp mark 2}{(4.813,1.619)}
\gppoint{gp mark 2}{(5.800,1.698)}
\gppoint{gp mark 2}{(6.788,1.779)}
\gppoint{gp mark 2}{(7.775,1.860)}
\gppoint{gp mark 2}{(8.763,1.944)}
\gppoint{gp mark 2}{(9.750,2.024)}
\gppoint{gp mark 2}{(10.738,2.103)}
\gppoint{gp mark 2}{(11.725,2.186)}
\gppoint{gp mark 2}{(10.603,7.922)}
\gpcolor{color=gp lt color border}
\node[gp node right,font={\fontsize{14.0pt}{16.8pt}\selectfont}] at (9.739,7.491) {\textsc{Ringing} sup-inf};
\gpcolor{rgb color={0.000,0.620,0.451}}
\draw[gp path] (9.997,7.491)--(11.209,7.491);
\draw[gp path] (2.838,2.026)--(3.825,2.676)--(4.813,3.314)--(5.800,3.959)--(6.788,4.610)%
  --(7.775,5.252)--(8.763,5.907)--(9.750,6.540)--(10.738,7.202)--(11.725,7.872);
\gppoint{gp mark 6}{(2.838,2.026)}
\gppoint{gp mark 6}{(3.825,2.676)}
\gppoint{gp mark 6}{(4.813,3.314)}
\gppoint{gp mark 6}{(5.800,3.959)}
\gppoint{gp mark 6}{(6.788,4.610)}
\gppoint{gp mark 6}{(7.775,5.252)}
\gppoint{gp mark 6}{(8.763,5.907)}
\gppoint{gp mark 6}{(9.750,6.540)}
\gppoint{gp mark 6}{(10.738,7.202)}
\gppoint{gp mark 6}{(11.725,7.872)}
\gppoint{gp mark 6}{(10.603,7.491)}
\gpcolor{color=gp lt color border}
\node[gp node right,font={\fontsize{14.0pt}{16.8pt}\selectfont}] at (9.739,7.060) {\textsc{Overshoot}, tropical};
\gpcolor{rgb color={0.337,0.706,0.914}}
\draw[gp path] (9.997,7.060)--(11.209,7.060);
\draw[gp path] (2.838,1.474)--(3.825,1.571)--(4.813,1.668)--(5.800,1.764)--(6.788,1.862)%
  --(7.775,1.958)--(8.763,2.057)--(9.750,2.154)--(10.738,2.250)--(11.725,2.344);
\gppoint{gp mark 8}{(2.838,1.474)}
\gppoint{gp mark 8}{(3.825,1.571)}
\gppoint{gp mark 8}{(4.813,1.668)}
\gppoint{gp mark 8}{(5.800,1.764)}
\gppoint{gp mark 8}{(6.788,1.862)}
\gppoint{gp mark 8}{(7.775,1.958)}
\gppoint{gp mark 8}{(8.763,2.057)}
\gppoint{gp mark 8}{(9.750,2.154)}
\gppoint{gp mark 8}{(10.738,2.250)}
\gppoint{gp mark 8}{(11.725,2.344)}
\gppoint{gp mark 8}{(10.603,7.060)}
\gpcolor{color=gp lt color border}
\node[gp node right,font={\fontsize{14.0pt}{16.8pt}\selectfont}] at (9.739,6.629) {\textsc{Ringing} tropical};
\gpcolor{rgb color={0.902,0.624,0.000}}
\draw[gp path] (9.997,6.629)--(11.209,6.629);
\draw[gp path] (2.838,1.710)--(3.825,2.032)--(4.813,2.360)--(5.800,2.682)--(6.788,3.009)%
  --(7.775,3.329)--(8.763,3.665)--(9.750,3.992)--(10.738,4.313)--(11.725,4.634);
\gppoint{gp mark 4}{(2.838,1.710)}
\gppoint{gp mark 4}{(3.825,2.032)}
\gppoint{gp mark 4}{(4.813,2.360)}
\gppoint{gp mark 4}{(5.800,2.682)}
\gppoint{gp mark 4}{(6.788,3.009)}
\gppoint{gp mark 4}{(7.775,3.329)}
\gppoint{gp mark 4}{(8.763,3.665)}
\gppoint{gp mark 4}{(9.750,3.992)}
\gppoint{gp mark 4}{(10.738,4.313)}
\gppoint{gp mark 4}{(11.725,4.634)}
\gppoint{gp mark 4}{(10.603,6.629)}
\gpcolor{color=gp lt color border}
\gpsetlinewidth{1.00}
\draw[gp path] (1.850,8.318)--(1.850,1.379)--(11.725,1.379)--(11.725,8.318)--cycle;
\gpdefrectangularnode{gp plot 1}{\pgfpoint{1.850cm}{1.379cm}}{\pgfpoint{11.725cm}{8.318cm}}
\end{tikzpicture}
}
\end{minipage}
 \hfill
\begin{minipage}{0.49\textwidth} 
 \centering
 \scalebox{0.40}{\begin{tikzpicture}[gnuplot]
\tikzset{every node/.append style={font={\fontsize{14.0pt}{16.8pt}\selectfont}}}
\path (0.000,0.000) rectangle (12.500,8.750);
\gpcolor{color=gp lt color border}
\gpsetlinetype{gp lt border}
\gpsetdashtype{gp dt solid}
\gpsetlinewidth{1.00}
\draw[gp path] (1.850,1.379)--(2.030,1.379);
\draw[gp path] (11.725,1.379)--(11.545,1.379);
\node[gp node right,font={\Large}] at (1.592,1.379) {$6.9$};
\draw[gp path] (1.850,2.150)--(2.030,2.150);
\draw[gp path] (11.725,2.150)--(11.545,2.150);
\node[gp node right,font={\Large}] at (1.592,2.150) {$7$};
\draw[gp path] (1.850,2.921)--(2.030,2.921);
\draw[gp path] (11.725,2.921)--(11.545,2.921);
\node[gp node right,font={\Large}] at (1.592,2.921) {$7.1$};
\draw[gp path] (1.850,3.692)--(2.030,3.692);
\draw[gp path] (11.725,3.692)--(11.545,3.692);
\node[gp node right,font={\Large}] at (1.592,3.692) {$7.2$};
\draw[gp path] (1.850,4.463)--(2.030,4.463);
\draw[gp path] (11.725,4.463)--(11.545,4.463);
\node[gp node right,font={\Large}] at (1.592,4.463) {$7.3$};
\draw[gp path] (1.850,5.234)--(2.030,5.234);
\draw[gp path] (11.725,5.234)--(11.545,5.234);
\node[gp node right,font={\Large}] at (1.592,5.234) {$7.4$};
\draw[gp path] (1.850,6.005)--(2.030,6.005);
\draw[gp path] (11.725,6.005)--(11.545,6.005);
\node[gp node right,font={\Large}] at (1.592,6.005) {$7.5$};
\draw[gp path] (1.850,6.776)--(2.030,6.776);
\draw[gp path] (11.725,6.776)--(11.545,6.776);
\node[gp node right,font={\Large}] at (1.592,6.776) {$7.6$};
\draw[gp path] (1.850,7.547)--(2.030,7.547);
\draw[gp path] (11.725,7.547)--(11.545,7.547);
\node[gp node right,font={\Large}] at (1.592,7.547) {$7.7$};
\draw[gp path] (1.850,8.318)--(2.030,8.318);
\draw[gp path] (11.725,8.318)--(11.545,8.318);
\node[gp node right,font={\Large}] at (1.592,8.318) {$7.8$};
\draw[gp path] (1.850,1.379)--(1.850,1.559);
\draw[gp path] (1.850,8.318)--(1.850,8.138);
\node[gp node center,font={\Large}] at (1.850,0.948) {$0$};
\draw[gp path] (3.496,1.379)--(3.496,1.559);
\draw[gp path] (3.496,8.318)--(3.496,8.138);
\node[gp node center,font={\Large}] at (3.496,0.948) {$1$};
\draw[gp path] (5.142,1.379)--(5.142,1.559);
\draw[gp path] (5.142,8.318)--(5.142,8.138);
\node[gp node center,font={\Large}] at (5.142,0.948) {$2$};
\draw[gp path] (6.788,1.379)--(6.788,1.559);
\draw[gp path] (6.788,8.318)--(6.788,8.138);
\node[gp node center,font={\Large}] at (6.788,0.948) {$3$};
\draw[gp path] (8.433,1.379)--(8.433,1.559);
\draw[gp path] (8.433,8.318)--(8.433,8.138);
\node[gp node center,font={\Large}] at (8.433,0.948) {$4$};
\draw[gp path] (10.079,1.379)--(10.079,1.559);
\draw[gp path] (10.079,8.318)--(10.079,8.138);
\node[gp node center,font={\Large}] at (10.079,0.948) {$5$};
\draw[gp path] (11.725,1.379)--(11.725,1.559);
\draw[gp path] (11.725,8.318)--(11.725,8.138);
\node[gp node center,font={\Large}] at (11.725,0.948) {$6$};
\draw[gp path] (1.850,8.318)--(1.850,1.379)--(11.725,1.379)--(11.725,8.318)--cycle;
\node[gp node center,rotate=-270,font={\fontsize{14.0pt}{16.8pt}\selectfont}] at (0.387,4.848) {Memory usage [MiB]};
\node[gp node center,font={\fontsize{14.0pt}{16.8pt}\selectfont}] at (6.787,0.302) {Number of entries of the signal [$\times 10,000$]};
\node[gp node right,font={\fontsize{14.0pt}{16.8pt}\selectfont}] at (9.739,7.922) {\textsc{Overshoot}, sup-inf};
\gpcolor{rgb color={0.580,0.000,0.827}}
\gpsetlinewidth{5.00}
\draw[gp path] (9.997,7.922)--(11.209,7.922);
\draw[gp path] (2.838,2.362)--(3.825,1.966)--(4.813,2.374)--(5.800,2.094)--(6.788,2.132)%
  --(7.775,1.993)--(8.763,2.374)--(9.750,2.358)--(10.738,2.090)--(11.725,2.412);
\gpsetpointsize{8.00}
\gppoint{gp mark 2}{(2.838,2.362)}
\gppoint{gp mark 2}{(3.825,1.966)}
\gppoint{gp mark 2}{(4.813,2.374)}
\gppoint{gp mark 2}{(5.800,2.094)}
\gppoint{gp mark 2}{(6.788,2.132)}
\gppoint{gp mark 2}{(7.775,1.993)}
\gppoint{gp mark 2}{(8.763,2.374)}
\gppoint{gp mark 2}{(9.750,2.358)}
\gppoint{gp mark 2}{(10.738,2.090)}
\gppoint{gp mark 2}{(11.725,2.412)}
\gppoint{gp mark 2}{(10.603,7.922)}
\gpcolor{color=gp lt color border}
\node[gp node right,font={\fontsize{14.0pt}{16.8pt}\selectfont}] at (9.739,7.491) {\textsc{Ringing}, sup-inf};
\gpcolor{rgb color={0.000,0.620,0.451}}
\draw[gp path] (9.997,7.491)--(11.209,7.491);
\draw[gp path] (2.838,8.032)--(3.825,7.842)--(4.813,8.113)--(5.800,8.246)--(6.788,7.753)%
  --(7.775,8.065)--(8.763,8.122)--(9.750,7.911)--(10.738,8.161)--(11.725,8.207);
\gppoint{gp mark 6}{(2.838,8.032)}
\gppoint{gp mark 6}{(3.825,7.842)}
\gppoint{gp mark 6}{(4.813,8.113)}
\gppoint{gp mark 6}{(5.800,8.246)}
\gppoint{gp mark 6}{(6.788,7.753)}
\gppoint{gp mark 6}{(7.775,8.065)}
\gppoint{gp mark 6}{(8.763,8.122)}
\gppoint{gp mark 6}{(9.750,7.911)}
\gppoint{gp mark 6}{(10.738,8.161)}
\gppoint{gp mark 6}{(11.725,8.207)}
\gppoint{gp mark 6}{(10.603,7.491)}
\gpcolor{color=gp lt color border}
\node[gp node right,font={\fontsize{14.0pt}{16.8pt}\selectfont}] at (9.739,7.060) {\textsc{Overshoot}, tropical};
\gpcolor{rgb color={0.337,0.706,0.914}}
\draw[gp path] (9.997,7.060)--(11.209,7.060);
\draw[gp path] (2.838,3.090)--(3.825,2.941)--(4.813,2.978)--(5.800,3.260)--(6.788,3.180)%
  --(7.775,3.052)--(8.763,3.054)--(9.750,3.007)--(10.738,3.019)--(11.725,2.888);
\gppoint{gp mark 8}{(2.838,3.090)}
\gppoint{gp mark 8}{(3.825,2.941)}
\gppoint{gp mark 8}{(4.813,2.978)}
\gppoint{gp mark 8}{(5.800,3.260)}
\gppoint{gp mark 8}{(6.788,3.180)}
\gppoint{gp mark 8}{(7.775,3.052)}
\gppoint{gp mark 8}{(8.763,3.054)}
\gppoint{gp mark 8}{(9.750,3.007)}
\gppoint{gp mark 8}{(10.738,3.019)}
\gppoint{gp mark 8}{(11.725,2.888)}
\gppoint{gp mark 8}{(10.603,7.060)}
\gpcolor{color=gp lt color border}
\node[gp node right,font={\fontsize{14.0pt}{16.8pt}\selectfont}] at (9.739,6.629) {\textsc{Ringing}, tropical};
\gpcolor{rgb color={0.902,0.624,0.000}}
\draw[gp path] (9.997,6.629)--(11.209,6.629);
\draw[gp path] (2.838,5.996)--(3.825,6.005)--(4.813,6.237)--(5.800,6.186)--(6.788,6.180)%
  --(7.775,6.226)--(8.763,6.246)--(9.750,6.436)--(10.738,6.366)--(11.725,6.439);
\gppoint{gp mark 4}{(2.838,5.996)}
\gppoint{gp mark 4}{(3.825,6.005)}
\gppoint{gp mark 4}{(4.813,6.237)}
\gppoint{gp mark 4}{(5.800,6.186)}
\gppoint{gp mark 4}{(6.788,6.180)}
\gppoint{gp mark 4}{(7.775,6.226)}
\gppoint{gp mark 4}{(8.763,6.246)}
\gppoint{gp mark 4}{(9.750,6.436)}
\gppoint{gp mark 4}{(10.738,6.366)}
\gppoint{gp mark 4}{(11.725,6.439)}
\gppoint{gp mark 4}{(10.603,6.629)}
\gpcolor{color=gp lt color border}
\gpsetlinewidth{1.00}
\draw[gp path] (1.850,8.318)--(1.850,1.379)--(11.725,1.379)--(11.725,8.318)--cycle;
\gpdefrectangularnode{gp plot 1}{\pgfpoint{1.850cm}{1.379cm}}{\pgfpoint{11.725cm}{8.318cm}}
\end{tikzpicture}
}
\end{minipage}
 \caption{Change in execution time (left) and memory usage (right) for \textsc{Overshoot} and \textsc{Ringing} with the number of the entries of the signals}
  \label{fig:result_long_signal}

\begin{minipage}{0.49\textwidth} 
 \centering
 \scalebox{0.40}{\begin{tikzpicture}[gnuplot]
\tikzset{every node/.append style={font={\fontsize{14.0pt}{16.8pt}\selectfont}}}
\path (0.000,0.000) rectangle (12.500,8.750);
\gpcolor{color=gp lt color border}
\gpsetlinetype{gp lt border}
\gpsetdashtype{gp dt solid}
\gpsetlinewidth{1.00}
\draw[gp path] (1.850,1.379)--(2.030,1.379);
\draw[gp path] (11.725,1.379)--(11.545,1.379);
\node[gp node right,font={\Large}] at (1.592,1.379) {$0$};
\draw[gp path] (1.850,2.370)--(2.030,2.370);
\draw[gp path] (11.725,2.370)--(11.545,2.370);
\node[gp node right,font={\Large}] at (1.592,2.370) {$20$};
\draw[gp path] (1.850,3.362)--(2.030,3.362);
\draw[gp path] (11.725,3.362)--(11.545,3.362);
\node[gp node right,font={\Large}] at (1.592,3.362) {$40$};
\draw[gp path] (1.850,4.353)--(2.030,4.353);
\draw[gp path] (11.725,4.353)--(11.545,4.353);
\node[gp node right,font={\Large}] at (1.592,4.353) {$60$};
\draw[gp path] (1.850,5.344)--(2.030,5.344);
\draw[gp path] (11.725,5.344)--(11.545,5.344);
\node[gp node right,font={\Large}] at (1.592,5.344) {$80$};
\draw[gp path] (1.850,6.335)--(2.030,6.335);
\draw[gp path] (11.725,6.335)--(11.545,6.335);
\node[gp node right,font={\Large}] at (1.592,6.335) {$100$};
\draw[gp path] (1.850,7.327)--(2.030,7.327);
\draw[gp path] (11.725,7.327)--(11.545,7.327);
\node[gp node right,font={\Large}] at (1.592,7.327) {$120$};
\draw[gp path] (1.850,8.318)--(2.030,8.318);
\draw[gp path] (11.725,8.318)--(11.545,8.318);
\node[gp node right,font={\Large}] at (1.592,8.318) {$140$};
\draw[gp path] (1.850,1.379)--(1.850,1.559);
\draw[gp path] (1.850,8.318)--(1.850,8.138);
\node[gp node center,font={\Large}] at (1.850,0.948) {$1$};
\draw[gp path] (2.947,1.379)--(2.947,1.559);
\draw[gp path] (2.947,8.318)--(2.947,8.138);
\node[gp node center,font={\Large}] at (2.947,0.948) {$2$};
\draw[gp path] (4.044,1.379)--(4.044,1.559);
\draw[gp path] (4.044,8.318)--(4.044,8.138);
\node[gp node center,font={\Large}] at (4.044,0.948) {$3$};
\draw[gp path] (5.142,1.379)--(5.142,1.559);
\draw[gp path] (5.142,8.318)--(5.142,8.138);
\node[gp node center,font={\Large}] at (5.142,0.948) {$4$};
\draw[gp path] (6.239,1.379)--(6.239,1.559);
\draw[gp path] (6.239,8.318)--(6.239,8.138);
\node[gp node center,font={\Large}] at (6.239,0.948) {$5$};
\draw[gp path] (7.336,1.379)--(7.336,1.559);
\draw[gp path] (7.336,8.318)--(7.336,8.138);
\node[gp node center,font={\Large}] at (7.336,0.948) {$6$};
\draw[gp path] (8.433,1.379)--(8.433,1.559);
\draw[gp path] (8.433,8.318)--(8.433,8.138);
\node[gp node center,font={\Large}] at (8.433,0.948) {$7$};
\draw[gp path] (9.531,1.379)--(9.531,1.559);
\draw[gp path] (9.531,8.318)--(9.531,8.138);
\node[gp node center,font={\Large}] at (9.531,0.948) {$8$};
\draw[gp path] (10.628,1.379)--(10.628,1.559);
\draw[gp path] (10.628,8.318)--(10.628,8.138);
\node[gp node center,font={\Large}] at (10.628,0.948) {$9$};
\draw[gp path] (11.725,1.379)--(11.725,1.559);
\draw[gp path] (11.725,8.318)--(11.725,8.138);
\node[gp node center,font={\Large}] at (11.725,0.948) {$10$};
\draw[gp path] (1.850,8.318)--(1.850,1.379)--(11.725,1.379)--(11.725,8.318)--cycle;
\node[gp node center,rotate=-270,font={\fontsize{14.0pt}{16.8pt}\selectfont}] at (0.387,4.848) {Execution time [s]};
\node[gp node center,font={\fontsize{14.0pt}{16.8pt}\selectfont}] at (6.787,0.302) {Number of entries of the signal [$\times 100$]};
\node[gp node right,font={\fontsize{13.0pt}{15.6pt}\selectfont}] at (9.853,7.938) {\textsc{Overshoot (Unbounded)}, sup-inf};
\gpcolor{rgb color={0.580,0.000,0.827}}
\gpsetlinewidth{5.00}
\draw[gp path] (10.092,7.938)--(11.228,7.938);
\draw[gp path] (1.850,1.386)--(2.947,1.425)--(4.044,1.526)--(5.142,1.715)--(6.239,2.046)%
  --(7.336,2.530)--(8.433,3.206)--(9.531,4.136)--(10.628,5.306)--(11.725,6.938);
\gpsetpointsize{8.00}
\gppoint{gp mark 3}{(1.850,1.386)}
\gppoint{gp mark 3}{(2.947,1.425)}
\gppoint{gp mark 3}{(4.044,1.526)}
\gppoint{gp mark 3}{(5.142,1.715)}
\gppoint{gp mark 3}{(6.239,2.046)}
\gppoint{gp mark 3}{(7.336,2.530)}
\gppoint{gp mark 3}{(8.433,3.206)}
\gppoint{gp mark 3}{(9.531,4.136)}
\gppoint{gp mark 3}{(10.628,5.306)}
\gppoint{gp mark 3}{(11.725,6.938)}
\gppoint{gp mark 3}{(10.660,7.938)}
\gpcolor{color=gp lt color border}
\node[gp node right,font={\fontsize{13.0pt}{15.6pt}\selectfont}] at (9.853,7.538) {\textsc{Overshoot (Unbounded)}, tropical};
\gpcolor{rgb color={0.000,0.620,0.451}}
\draw[gp path] (10.092,7.538)--(11.228,7.538);
\draw[gp path] (1.850,1.388)--(2.947,1.437)--(4.044,1.557)--(5.142,1.786)--(6.239,2.160)%
  --(7.336,2.734)--(8.433,3.527)--(9.531,4.565)--(10.628,5.877)--(11.725,7.469);
\gppoint{gp mark 4}{(1.850,1.388)}
\gppoint{gp mark 4}{(2.947,1.437)}
\gppoint{gp mark 4}{(4.044,1.557)}
\gppoint{gp mark 4}{(5.142,1.786)}
\gppoint{gp mark 4}{(6.239,2.160)}
\gppoint{gp mark 4}{(7.336,2.734)}
\gppoint{gp mark 4}{(8.433,3.527)}
\gppoint{gp mark 4}{(9.531,4.565)}
\gppoint{gp mark 4}{(10.628,5.877)}
\gppoint{gp mark 4}{(11.725,7.469)}
\gppoint{gp mark 4}{(10.660,7.538)}
\gpcolor{color=gp lt color border}
\gpsetlinewidth{1.00}
\draw[gp path] (1.850,8.318)--(1.850,1.379)--(11.725,1.379)--(11.725,8.318)--cycle;
\gpdefrectangularnode{gp plot 1}{\pgfpoint{1.850cm}{1.379cm}}{\pgfpoint{11.725cm}{8.318cm}}
\end{tikzpicture}
}
\end{minipage}
\begin{minipage}{0.49\textwidth} 
 \centering
 \scalebox{0.40}{\begin{tikzpicture}[gnuplot]
\tikzset{every node/.append style={font={\fontsize{14.0pt}{16.8pt}\selectfont}}}
\path (0.000,0.000) rectangle (12.500,8.750);
\gpcolor{color=gp lt color border}
\gpsetlinetype{gp lt border}
\gpsetdashtype{gp dt solid}
\gpsetlinewidth{1.00}
\draw[gp path] (1.850,1.379)--(2.030,1.379);
\draw[gp path] (11.725,1.379)--(11.545,1.379);
\node[gp node right,font={\Large}] at (1.592,1.379) {$0$};
\draw[gp path] (1.850,2.767)--(2.030,2.767);
\draw[gp path] (11.725,2.767)--(11.545,2.767);
\node[gp node right,font={\Large}] at (1.592,2.767) {$50$};
\draw[gp path] (1.850,4.155)--(2.030,4.155);
\draw[gp path] (11.725,4.155)--(11.545,4.155);
\node[gp node right,font={\Large}] at (1.592,4.155) {$100$};
\draw[gp path] (1.850,5.542)--(2.030,5.542);
\draw[gp path] (11.725,5.542)--(11.545,5.542);
\node[gp node right,font={\Large}] at (1.592,5.542) {$150$};
\draw[gp path] (1.850,6.930)--(2.030,6.930);
\draw[gp path] (11.725,6.930)--(11.545,6.930);
\node[gp node right,font={\Large}] at (1.592,6.930) {$200$};
\draw[gp path] (1.850,8.318)--(2.030,8.318);
\draw[gp path] (11.725,8.318)--(11.545,8.318);
\node[gp node right,font={\Large}] at (1.592,8.318) {$250$};
\draw[gp path] (1.850,1.379)--(1.850,1.559);
\draw[gp path] (1.850,8.318)--(1.850,8.138);
\node[gp node center,font={\Large}] at (1.850,0.948) {$1$};
\draw[gp path] (2.947,1.379)--(2.947,1.559);
\draw[gp path] (2.947,8.318)--(2.947,8.138);
\node[gp node center,font={\Large}] at (2.947,0.948) {$2$};
\draw[gp path] (4.044,1.379)--(4.044,1.559);
\draw[gp path] (4.044,8.318)--(4.044,8.138);
\node[gp node center,font={\Large}] at (4.044,0.948) {$3$};
\draw[gp path] (5.142,1.379)--(5.142,1.559);
\draw[gp path] (5.142,8.318)--(5.142,8.138);
\node[gp node center,font={\Large}] at (5.142,0.948) {$4$};
\draw[gp path] (6.239,1.379)--(6.239,1.559);
\draw[gp path] (6.239,8.318)--(6.239,8.138);
\node[gp node center,font={\Large}] at (6.239,0.948) {$5$};
\draw[gp path] (7.336,1.379)--(7.336,1.559);
\draw[gp path] (7.336,8.318)--(7.336,8.138);
\node[gp node center,font={\Large}] at (7.336,0.948) {$6$};
\draw[gp path] (8.433,1.379)--(8.433,1.559);
\draw[gp path] (8.433,8.318)--(8.433,8.138);
\node[gp node center,font={\Large}] at (8.433,0.948) {$7$};
\draw[gp path] (9.531,1.379)--(9.531,1.559);
\draw[gp path] (9.531,8.318)--(9.531,8.138);
\node[gp node center,font={\Large}] at (9.531,0.948) {$8$};
\draw[gp path] (10.628,1.379)--(10.628,1.559);
\draw[gp path] (10.628,8.318)--(10.628,8.138);
\node[gp node center,font={\Large}] at (10.628,0.948) {$9$};
\draw[gp path] (11.725,1.379)--(11.725,1.559);
\draw[gp path] (11.725,8.318)--(11.725,8.138);
\node[gp node center,font={\Large}] at (11.725,0.948) {$10$};
\draw[gp path] (1.850,8.318)--(1.850,1.379)--(11.725,1.379)--(11.725,8.318)--cycle;
\node[gp node center,rotate=-270,font={\fontsize{14.0pt}{16.8pt}\selectfont}] at (0.387,4.848) {Memory Usage [MiB]};
\node[gp node center,font={\fontsize{14.0pt}{16.8pt}\selectfont}] at (6.787,0.302) {Number of entries of the signal [$\times 100$]};
\node[gp node right,font={\fontsize{13.0pt}{15.6pt}\selectfont}] at (9.853,7.938) {\textsc{Overshoot (Unbounded)}, sup-inf};
\gpcolor{rgb color={0.580,0.000,0.827}}
\gpsetlinewidth{5.00}
\draw[gp path] (10.092,7.938)--(11.228,7.938);
\draw[gp path] (1.850,1.624)--(2.947,1.776)--(4.044,2.021)--(5.142,2.367)--(6.239,2.810)%
  --(7.336,3.348)--(8.433,3.987)--(9.531,4.720)--(10.628,5.551)--(11.725,6.476);
\gpsetpointsize{8.00}
\gppoint{gp mark 3}{(1.850,1.624)}
\gppoint{gp mark 3}{(2.947,1.776)}
\gppoint{gp mark 3}{(4.044,2.021)}
\gppoint{gp mark 3}{(5.142,2.367)}
\gppoint{gp mark 3}{(6.239,2.810)}
\gppoint{gp mark 3}{(7.336,3.348)}
\gppoint{gp mark 3}{(8.433,3.987)}
\gppoint{gp mark 3}{(9.531,4.720)}
\gppoint{gp mark 3}{(10.628,5.551)}
\gppoint{gp mark 3}{(11.725,6.476)}
\gppoint{gp mark 3}{(10.660,7.938)}
\gpcolor{color=gp lt color border}
\node[gp node right,font={\fontsize{13.0pt}{15.6pt}\selectfont}] at (9.853,7.538) {\textsc{Overshoot (Unbounded)}, tropical};
\gpcolor{rgb color={0.000,0.620,0.451}}
\draw[gp path] (10.092,7.538)--(11.228,7.538);
\draw[gp path] (1.850,1.645)--(2.947,1.835)--(4.044,2.145)--(5.142,2.575)--(6.239,3.128)%
  --(7.336,3.797)--(8.433,4.591)--(9.531,5.501)--(10.628,6.531)--(11.725,7.680);
\gppoint{gp mark 4}{(1.850,1.645)}
\gppoint{gp mark 4}{(2.947,1.835)}
\gppoint{gp mark 4}{(4.044,2.145)}
\gppoint{gp mark 4}{(5.142,2.575)}
\gppoint{gp mark 4}{(6.239,3.128)}
\gppoint{gp mark 4}{(7.336,3.797)}
\gppoint{gp mark 4}{(8.433,4.591)}
\gppoint{gp mark 4}{(9.531,5.501)}
\gppoint{gp mark 4}{(10.628,6.531)}
\gppoint{gp mark 4}{(11.725,7.680)}
\gppoint{gp mark 4}{(10.660,7.538)}
\gpcolor{color=gp lt color border}
\gpsetlinewidth{1.00}
\draw[gp path] (1.850,8.318)--(1.850,1.379)--(11.725,1.379)--(11.725,8.318)--cycle;
\gpdefrectangularnode{gp plot 1}{\pgfpoint{1.850cm}{1.379cm}}{\pgfpoint{11.725cm}{8.318cm}}
\end{tikzpicture}
}
\end{minipage}
 \caption{Change in execution time (left) and memory usage (right) for \textsc{Overshoot (Unbounded)} with the number of the entries of the signals}
 \label{fig:result_long_signal_unbounded}
\end{figure}

\paragraph{RQ1: Practical Performance}

\cref{fig:result_long_signal,fig:result_long_signal_unbounded} show the execution time and memory usage of our quantitative timed pattern matching for the TSWAs $\mathcal{W}$ 
and signals $\sigma$.
Here, we fixed the sampling frequency to be 0.1 Hz and changed the duration $|\sigma|$ of the signal from 60,000 s to 600,000 s in \textsc{Overshoot} and \textsc{Ringing}, and from 1,000 s to 10,000 s in \textsc{Overshoot (Unbounded)}.

In \cref{fig:result_long_signal}, we observe that \cref{alg:qtpm} handles the log with 60,000 entries in less than 20 s with less than 7.1 MiB of memory usage for \textsc{Overshoot},
and in about 1 or 2 minutes with less than 7.8 MiB of memory usage for \textsc{Ringing}.
In \cref{fig:result_long_signal_unbounded}, we observe that \cref{alg:qtpm} handles the log with 10,000 entries in less than 120 s with less than 250 MiB of memory usage for \textsc{Overshoot (Unbounded)}.
Although the quantitative timed pattern matching problem is complex, we conclude that its practical performance is realistic.

\paragraph{RQ2: Change in Speed and Memory Usage with Signal Size}
\cref{fig:result_long_signal,fig:result_long_signal_unbounded} show the execution time and memory usage of our quantitative timed pattern matching.
See RQ1 for the detail of our experimental setting.

In \cref{fig:result_long_signal}, for the TSAs with time-bound, we observe that the execution time is linear with respect to the duration $\duration{\signal}$ of the input signals and the memory usage is more or less 
constant with respect to \LongVersion{the duration }$\duration{\signal}$\LongVersion{ of the input signals}. 
This performance is essential for a monitor to keep monitoring a running system.

In \cref{fig:result_long_signal_unbounded}, for the TSA without any time-bound, we observe that the execution time is cubic and the memory usage is quadratic with respect to the number of the entries in $\duration{\signal}$.
The memory usage increases quadratically with the number of the entries because 
the intermediate weight $\incrementalWeight_j$ has an entry for each initial interval $[\tau_i,\tau_{i+1})$ of the trimming and for each interval $[\tau_k,\tau_{k+1})$ where the transition occurred.
The execution time increases cubically with respect to the number of the entries because the shortest distance is computed for each entry of $\incrementalWeight_j$.
However, we note that our quantitative timed pattern matching still works when the number of the entries is relatively small.

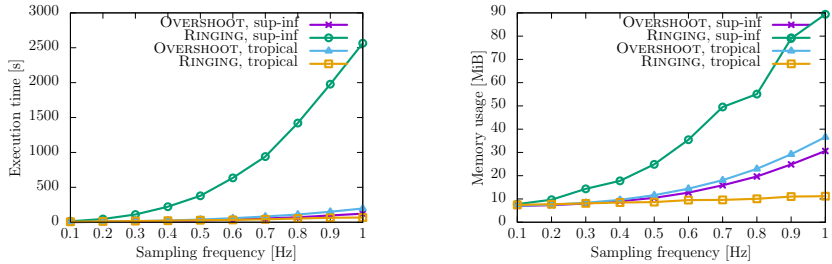
\begin{figure}[tbp]
 \begin{minipage}{0.49\textwidth} 
 \centering
 \scalebox{0.40}{\begin{tikzpicture}[gnuplot]
\tikzset{every node/.append style={font={\fontsize{14.0pt}{16.8pt}\selectfont}}}
\path (0.000,0.000) rectangle (12.500,8.750);
\gpcolor{color=gp lt color border}
\gpsetlinetype{gp lt border}
\gpsetdashtype{gp dt solid}
\gpsetlinewidth{1.00}
\draw[gp path] (2.108,1.379)--(2.288,1.379);
\draw[gp path] (11.725,1.379)--(11.545,1.379);
\node[gp node right,font={\Large}] at (1.850,1.379) {$0$};
\draw[gp path] (2.108,2.536)--(2.288,2.536);
\draw[gp path] (11.725,2.536)--(11.545,2.536);
\node[gp node right,font={\Large}] at (1.850,2.536) {$500$};
\draw[gp path] (2.108,3.692)--(2.288,3.692);
\draw[gp path] (11.725,3.692)--(11.545,3.692);
\node[gp node right,font={\Large}] at (1.850,3.692) {$1000$};
\draw[gp path] (2.108,4.849)--(2.288,4.849);
\draw[gp path] (11.725,4.849)--(11.545,4.849);
\node[gp node right,font={\Large}] at (1.850,4.849) {$1500$};
\draw[gp path] (2.108,6.005)--(2.288,6.005);
\draw[gp path] (11.725,6.005)--(11.545,6.005);
\node[gp node right,font={\Large}] at (1.850,6.005) {$2000$};
\draw[gp path] (2.108,7.162)--(2.288,7.162);
\draw[gp path] (11.725,7.162)--(11.545,7.162);
\node[gp node right,font={\Large}] at (1.850,7.162) {$2500$};
\draw[gp path] (2.108,8.318)--(2.288,8.318);
\draw[gp path] (11.725,8.318)--(11.545,8.318);
\node[gp node right,font={\Large}] at (1.850,8.318) {$3000$};
\draw[gp path] (2.108,1.379)--(2.108,1.559);
\draw[gp path] (2.108,8.318)--(2.108,8.138);
\node[gp node center,font={\Large}] at (2.108,0.948) {$0.1$};
\draw[gp path] (3.177,1.379)--(3.177,1.559);
\draw[gp path] (3.177,8.318)--(3.177,8.138);
\node[gp node center,font={\Large}] at (3.177,0.948) {$0.2$};
\draw[gp path] (4.245,1.379)--(4.245,1.559);
\draw[gp path] (4.245,8.318)--(4.245,8.138);
\node[gp node center,font={\Large}] at (4.245,0.948) {$0.3$};
\draw[gp path] (5.314,1.379)--(5.314,1.559);
\draw[gp path] (5.314,8.318)--(5.314,8.138);
\node[gp node center,font={\Large}] at (5.314,0.948) {$0.4$};
\draw[gp path] (6.382,1.379)--(6.382,1.559);
\draw[gp path] (6.382,8.318)--(6.382,8.138);
\node[gp node center,font={\Large}] at (6.382,0.948) {$0.5$};
\draw[gp path] (7.451,1.379)--(7.451,1.559);
\draw[gp path] (7.451,8.318)--(7.451,8.138);
\node[gp node center,font={\Large}] at (7.451,0.948) {$0.6$};
\draw[gp path] (8.519,1.379)--(8.519,1.559);
\draw[gp path] (8.519,8.318)--(8.519,8.138);
\node[gp node center,font={\Large}] at (8.519,0.948) {$0.7$};
\draw[gp path] (9.588,1.379)--(9.588,1.559);
\draw[gp path] (9.588,8.318)--(9.588,8.138);
\node[gp node center,font={\Large}] at (9.588,0.948) {$0.8$};
\draw[gp path] (10.656,1.379)--(10.656,1.559);
\draw[gp path] (10.656,8.318)--(10.656,8.138);
\node[gp node center,font={\Large}] at (10.656,0.948) {$0.9$};
\draw[gp path] (11.725,1.379)--(11.725,1.559);
\draw[gp path] (11.725,8.318)--(11.725,8.138);
\node[gp node center,font={\Large}] at (11.725,0.948) {$1$};
\draw[gp path] (2.108,8.318)--(2.108,1.379)--(11.725,1.379)--(11.725,8.318)--cycle;
\node[gp node center,rotate=-270,font={\fontsize{14.0pt}{16.8pt}\selectfont}] at (0.387,4.848) {Execution time [s]};
\node[gp node center,font={\fontsize{14.0pt}{16.8pt}\selectfont}] at (6.916,0.302) {Sampling frequency [Hz]};
\node[gp node right,font={\fontsize{14.0pt}{16.8pt}\selectfont}] at (9.739,7.922) {\textsc{Overshoot}, sup-inf};
\gpcolor{rgb color={0.580,0.000,0.827}}
\gpsetlinewidth{5.00}
\draw[gp path] (9.997,7.922)--(11.209,7.922);
\draw[gp path] (2.108,1.383)--(3.177,1.390)--(4.245,1.402)--(5.314,1.417)--(6.382,1.438)%
  --(7.451,1.467)--(8.519,1.503)--(9.588,1.543)--(10.656,1.603)--(11.725,1.662);
\gpsetpointsize{8.00}
\gppoint{gp mark 2}{(2.108,1.383)}
\gppoint{gp mark 2}{(3.177,1.390)}
\gppoint{gp mark 2}{(4.245,1.402)}
\gppoint{gp mark 2}{(5.314,1.417)}
\gppoint{gp mark 2}{(6.382,1.438)}
\gppoint{gp mark 2}{(7.451,1.467)}
\gppoint{gp mark 2}{(8.519,1.503)}
\gppoint{gp mark 2}{(9.588,1.543)}
\gppoint{gp mark 2}{(10.656,1.603)}
\gppoint{gp mark 2}{(11.725,1.662)}
\gppoint{gp mark 2}{(10.603,7.922)}
\gpcolor{color=gp lt color border}
\node[gp node right,font={\fontsize{14.0pt}{16.8pt}\selectfont}] at (9.739,7.491) {\textsc{Ringing}, sup-inf};
\gpcolor{rgb color={0.000,0.620,0.451}}
\draw[gp path] (9.997,7.491)--(11.209,7.491);
\draw[gp path] (2.108,1.409)--(3.177,1.485)--(4.245,1.634)--(5.314,1.893)--(6.382,2.254)%
  --(7.451,2.848)--(8.519,3.555)--(9.588,4.666)--(10.656,5.954)--(11.725,7.312);
\gppoint{gp mark 6}{(2.108,1.409)}
\gppoint{gp mark 6}{(3.177,1.485)}
\gppoint{gp mark 6}{(4.245,1.634)}
\gppoint{gp mark 6}{(5.314,1.893)}
\gppoint{gp mark 6}{(6.382,2.254)}
\gppoint{gp mark 6}{(7.451,2.848)}
\gppoint{gp mark 6}{(8.519,3.555)}
\gppoint{gp mark 6}{(9.588,4.666)}
\gppoint{gp mark 6}{(10.656,5.954)}
\gppoint{gp mark 6}{(11.725,7.312)}
\gppoint{gp mark 6}{(10.603,7.491)}
\gpcolor{color=gp lt color border}
\node[gp node right,font={\fontsize{14.0pt}{16.8pt}\selectfont}] at (9.739,7.060) {\textsc{Overshoot}, tropical};
\gpcolor{rgb color={0.337,0.706,0.914}}
\draw[gp path] (9.997,7.060)--(11.209,7.060);
\draw[gp path] (2.108,1.383)--(3.177,1.395)--(4.245,1.413)--(5.314,1.435)--(6.382,1.467)%
  --(7.451,1.512)--(8.519,1.568)--(9.588,1.634)--(10.656,1.726)--(11.725,1.834);
\gppoint{gp mark 8}{(2.108,1.383)}
\gppoint{gp mark 8}{(3.177,1.395)}
\gppoint{gp mark 8}{(4.245,1.413)}
\gppoint{gp mark 8}{(5.314,1.435)}
\gppoint{gp mark 8}{(6.382,1.467)}
\gppoint{gp mark 8}{(7.451,1.512)}
\gppoint{gp mark 8}{(8.519,1.568)}
\gppoint{gp mark 8}{(9.588,1.634)}
\gppoint{gp mark 8}{(10.656,1.726)}
\gppoint{gp mark 8}{(11.725,1.834)}
\gppoint{gp mark 8}{(10.603,7.060)}
\gpcolor{color=gp lt color border}
\node[gp node right,font={\fontsize{14.0pt}{16.8pt}\selectfont}] at (9.739,6.629) {\textsc{Ringing}, tropical};
\gpcolor{rgb color={0.902,0.624,0.000}}
\draw[gp path] (9.997,6.629)--(11.209,6.629);
\draw[gp path] (2.108,1.394)--(3.177,1.406)--(4.245,1.417)--(5.314,1.429)--(6.382,1.446)%
  --(7.451,1.455)--(8.519,1.475)--(9.588,1.500)--(10.656,1.526)--(11.725,1.534);
\gppoint{gp mark 4}{(2.108,1.394)}
\gppoint{gp mark 4}{(3.177,1.406)}
\gppoint{gp mark 4}{(4.245,1.417)}
\gppoint{gp mark 4}{(5.314,1.429)}
\gppoint{gp mark 4}{(6.382,1.446)}
\gppoint{gp mark 4}{(7.451,1.455)}
\gppoint{gp mark 4}{(8.519,1.475)}
\gppoint{gp mark 4}{(9.588,1.500)}
\gppoint{gp mark 4}{(10.656,1.526)}
\gppoint{gp mark 4}{(11.725,1.534)}
\gppoint{gp mark 4}{(10.603,6.629)}
\gpcolor{color=gp lt color border}
\gpsetlinewidth{1.00}
\draw[gp path] (2.108,8.318)--(2.108,1.379)--(11.725,1.379)--(11.725,8.318)--cycle;
\gpdefrectangularnode{gp plot 1}{\pgfpoint{2.108cm}{1.379cm}}{\pgfpoint{11.725cm}{8.318cm}}
\end{tikzpicture}
}
\end{minipage}
 \begin{minipage}{0.49\textwidth} 
 \centering
 \scalebox{0.40}{\begin{tikzpicture}[gnuplot]
\tikzset{every node/.append style={font={\fontsize{14.0pt}{16.8pt}\selectfont}}}
\path (0.000,0.000) rectangle (12.500,8.750);
\gpcolor{color=gp lt color border}
\gpsetlinetype{gp lt border}
\gpsetdashtype{gp dt solid}
\gpsetlinewidth{1.00}
\draw[gp path] (1.592,1.379)--(1.772,1.379);
\draw[gp path] (11.725,1.379)--(11.545,1.379);
\node[gp node right,font={\Large}] at (1.334,1.379) {$0$};
\draw[gp path] (1.592,2.150)--(1.772,2.150);
\draw[gp path] (11.725,2.150)--(11.545,2.150);
\node[gp node right,font={\Large}] at (1.334,2.150) {$10$};
\draw[gp path] (1.592,2.921)--(1.772,2.921);
\draw[gp path] (11.725,2.921)--(11.545,2.921);
\node[gp node right,font={\Large}] at (1.334,2.921) {$20$};
\draw[gp path] (1.592,3.692)--(1.772,3.692);
\draw[gp path] (11.725,3.692)--(11.545,3.692);
\node[gp node right,font={\Large}] at (1.334,3.692) {$30$};
\draw[gp path] (1.592,4.463)--(1.772,4.463);
\draw[gp path] (11.725,4.463)--(11.545,4.463);
\node[gp node right,font={\Large}] at (1.334,4.463) {$40$};
\draw[gp path] (1.592,5.234)--(1.772,5.234);
\draw[gp path] (11.725,5.234)--(11.545,5.234);
\node[gp node right,font={\Large}] at (1.334,5.234) {$50$};
\draw[gp path] (1.592,6.005)--(1.772,6.005);
\draw[gp path] (11.725,6.005)--(11.545,6.005);
\node[gp node right,font={\Large}] at (1.334,6.005) {$60$};
\draw[gp path] (1.592,6.776)--(1.772,6.776);
\draw[gp path] (11.725,6.776)--(11.545,6.776);
\node[gp node right,font={\Large}] at (1.334,6.776) {$70$};
\draw[gp path] (1.592,7.547)--(1.772,7.547);
\draw[gp path] (11.725,7.547)--(11.545,7.547);
\node[gp node right,font={\Large}] at (1.334,7.547) {$80$};
\draw[gp path] (1.592,8.318)--(1.772,8.318);
\draw[gp path] (11.725,8.318)--(11.545,8.318);
\node[gp node right,font={\Large}] at (1.334,8.318) {$90$};
\draw[gp path] (1.592,1.379)--(1.592,1.559);
\draw[gp path] (1.592,8.318)--(1.592,8.138);
\node[gp node center,font={\Large}] at (1.592,0.948) {$0.1$};
\draw[gp path] (2.718,1.379)--(2.718,1.559);
\draw[gp path] (2.718,8.318)--(2.718,8.138);
\node[gp node center,font={\Large}] at (2.718,0.948) {$0.2$};
\draw[gp path] (3.844,1.379)--(3.844,1.559);
\draw[gp path] (3.844,8.318)--(3.844,8.138);
\node[gp node center,font={\Large}] at (3.844,0.948) {$0.3$};
\draw[gp path] (4.970,1.379)--(4.970,1.559);
\draw[gp path] (4.970,8.318)--(4.970,8.138);
\node[gp node center,font={\Large}] at (4.970,0.948) {$0.4$};
\draw[gp path] (6.096,1.379)--(6.096,1.559);
\draw[gp path] (6.096,8.318)--(6.096,8.138);
\node[gp node center,font={\Large}] at (6.096,0.948) {$0.5$};
\draw[gp path] (7.221,1.379)--(7.221,1.559);
\draw[gp path] (7.221,8.318)--(7.221,8.138);
\node[gp node center,font={\Large}] at (7.221,0.948) {$0.6$};
\draw[gp path] (8.347,1.379)--(8.347,1.559);
\draw[gp path] (8.347,8.318)--(8.347,8.138);
\node[gp node center,font={\Large}] at (8.347,0.948) {$0.7$};
\draw[gp path] (9.473,1.379)--(9.473,1.559);
\draw[gp path] (9.473,8.318)--(9.473,8.138);
\node[gp node center,font={\Large}] at (9.473,0.948) {$0.8$};
\draw[gp path] (10.599,1.379)--(10.599,1.559);
\draw[gp path] (10.599,8.318)--(10.599,8.138);
\node[gp node center,font={\Large}] at (10.599,0.948) {$0.9$};
\draw[gp path] (11.725,1.379)--(11.725,1.559);
\draw[gp path] (11.725,8.318)--(11.725,8.138);
\node[gp node center,font={\Large}] at (11.725,0.948) {$1$};
\draw[gp path] (1.592,8.318)--(1.592,1.379)--(11.725,1.379)--(11.725,8.318)--cycle;
\node[gp node center,rotate=-270,font={\fontsize{14.0pt}{16.8pt}\selectfont}] at (0.387,4.848) {Memory usage [MiB]};
\node[gp node center,font={\fontsize{14.0pt}{16.8pt}\selectfont}] at (6.658,0.302) {Sampling frequency [Hz]};
\node[gp node right,font={\fontsize{14.0pt}{16.8pt}\selectfont}] at (9.739,7.922) {\textsc{Overshoot}, sup-inf};
\gpcolor{rgb color={0.580,0.000,0.827}}
\gpsetlinewidth{5.00}
\draw[gp path] (9.997,7.922)--(11.209,7.922);
\draw[gp path] (1.592,1.919)--(2.718,1.945)--(3.844,1.999)--(4.970,2.077)--(6.096,2.186)%
  --(7.221,2.359)--(8.347,2.603)--(9.473,2.895)--(10.599,3.294)--(11.725,3.740);
\gpsetpointsize{8.00}
\gppoint{gp mark 2}{(1.592,1.919)}
\gppoint{gp mark 2}{(2.718,1.945)}
\gppoint{gp mark 2}{(3.844,1.999)}
\gppoint{gp mark 2}{(4.970,2.077)}
\gppoint{gp mark 2}{(6.096,2.186)}
\gppoint{gp mark 2}{(7.221,2.359)}
\gppoint{gp mark 2}{(8.347,2.603)}
\gppoint{gp mark 2}{(9.473,2.895)}
\gppoint{gp mark 2}{(10.599,3.294)}
\gppoint{gp mark 2}{(11.725,3.740)}
\gppoint{gp mark 2}{(10.603,7.922)}
\gpcolor{color=gp lt color border}
\node[gp node right,font={\fontsize{14.0pt}{16.8pt}\selectfont}] at (9.739,7.491) {\textsc{Ringing}, sup-inf};
\gpcolor{rgb color={0.000,0.620,0.451}}
\draw[gp path] (9.997,7.491)--(11.209,7.491);
\draw[gp path] (1.592,1.977)--(2.718,2.125)--(3.844,2.485)--(4.970,2.750)--(6.096,3.297)%
  --(7.221,4.113)--(8.347,5.198)--(9.473,5.626)--(10.599,7.468)--(11.725,8.272);
\gppoint{gp mark 6}{(1.592,1.977)}
\gppoint{gp mark 6}{(2.718,2.125)}
\gppoint{gp mark 6}{(3.844,2.485)}
\gppoint{gp mark 6}{(4.970,2.750)}
\gppoint{gp mark 6}{(6.096,3.297)}
\gppoint{gp mark 6}{(7.221,4.113)}
\gppoint{gp mark 6}{(8.347,5.198)}
\gppoint{gp mark 6}{(9.473,5.626)}
\gppoint{gp mark 6}{(10.599,7.468)}
\gppoint{gp mark 6}{(11.725,8.272)}
\gppoint{gp mark 6}{(10.603,7.491)}
\gpcolor{color=gp lt color border}
\node[gp node right,font={\fontsize{14.0pt}{16.8pt}\selectfont}] at (9.739,7.060) {\textsc{Overshoot}, tropical};
\gpcolor{rgb color={0.337,0.706,0.914}}
\draw[gp path] (9.997,7.060)--(11.209,7.060);
\draw[gp path] (1.592,1.927)--(2.718,1.959)--(3.844,2.023)--(4.970,2.120)--(6.096,2.275)%
  --(7.221,2.491)--(8.347,2.772)--(9.473,3.147)--(10.599,3.631)--(11.725,4.200);
\gppoint{gp mark 8}{(1.592,1.927)}
\gppoint{gp mark 8}{(2.718,1.959)}
\gppoint{gp mark 8}{(3.844,2.023)}
\gppoint{gp mark 8}{(4.970,2.120)}
\gppoint{gp mark 8}{(6.096,2.275)}
\gppoint{gp mark 8}{(7.221,2.491)}
\gppoint{gp mark 8}{(8.347,2.772)}
\gppoint{gp mark 8}{(9.473,3.147)}
\gppoint{gp mark 8}{(10.599,3.631)}
\gppoint{gp mark 8}{(11.725,4.200)}
\gppoint{gp mark 8}{(10.603,7.060)}
\gpcolor{color=gp lt color border}
\node[gp node right,font={\fontsize{14.0pt}{16.8pt}\selectfont}] at (9.739,6.629) {\textsc{Ringing}, tropical};
\gpcolor{rgb color={0.902,0.624,0.000}}
\draw[gp path] (9.997,6.629)--(11.209,6.629);
\draw[gp path] (1.592,1.955)--(2.718,1.975)--(3.844,2.004)--(4.970,2.032)--(6.096,2.046)%
  --(7.221,2.112)--(8.347,2.119)--(9.473,2.155)--(10.599,2.229)--(11.725,2.241);
\gppoint{gp mark 4}{(1.592,1.955)}
\gppoint{gp mark 4}{(2.718,1.975)}
\gppoint{gp mark 4}{(3.844,2.004)}
\gppoint{gp mark 4}{(4.970,2.032)}
\gppoint{gp mark 4}{(6.096,2.046)}
\gppoint{gp mark 4}{(7.221,2.112)}
\gppoint{gp mark 4}{(8.347,2.119)}
\gppoint{gp mark 4}{(9.473,2.155)}
\gppoint{gp mark 4}{(10.599,2.229)}
\gppoint{gp mark 4}{(11.725,2.241)}
\gppoint{gp mark 4}{(10.603,6.629)}
\gpcolor{color=gp lt color border}
\gpsetlinewidth{1.00}
\draw[gp path] (1.592,8.318)--(1.592,1.379)--(11.725,1.379)--(11.725,8.318)--cycle;
\gpdefrectangularnode{gp plot 1}{\pgfpoint{1.592cm}{1.379cm}}{\pgfpoint{11.725cm}{8.318cm}}
\end{tikzpicture}
}
\end{minipage}
 \caption{Change in execution time (left) and memory usage (right) for \textsc{Overshoot} and \textsc{Ringing} with the sampling frequency}
\label{fig:result_high_freq}
\end{figure}
\paragraph{RQ3: Change in Speed and Memory Usage with Sampling Frequency}
\cref{fig:result_high_freq} shows the execution time and memory usage\LongVersion{ of our quantitative timed pattern matching} for each TSWA $\TSWA$ and signal $\signal$ of \textsc{Overshoot} and \textsc{Ringing}.
Here, we fixed the number of the entries to be 6,000 and changed the sampling frequency from 0.1 Hz to 1.0 Hz.

In \cref{fig:result_high_freq}, we observe that the execution time is cubic, and the memory usage is more or less quadratic with respect to the sampling frequency of the signals.
This is because the number of the entries in a certain duration is linear to the sampling frequency, 
which increases the number of the reachability states of the WSTTSs quadratically. 
Despite the steep curve of the execution time, we also observe that the execution time is smaller than the duration of the signal.
Therefore, our algorithm is online capable at least for these sampling frequencies.

\section{Conclusions and future work}
\label{sec:conclusions_and_future_work}

Using an automata-based approach, we proposed an online algorithm for quantitative timed pattern matching.\LongVersion{ The key idea of this approach is the reduction to the shortest distance of a weighted graph using zones.}

Comparison of the expressiveness of TSWAs with other formalisms \eg{} \emph{signal temporal logic}~\cite{DBLP:conf/formats/MalerN04} or \emph{signal regular expressions}~\cite{DBLP:conf/formats/BakhirkinFMU17} is future work.
\LongVersion{Another future work is the comparison with the quantitative semantics based on the distance between traces presented in~\cite{DBLP:journals/tcad/JaksicBGN18}.}
\newpage
\ifdefined\VersionLong
	\bibliographystyle{alpha} 
	\newcommand{\IJFCS}{International Journal of Foundations of Computer Science}
	\newcommand{\JLAP}{Journal of Logic and Algebraic Programming}
	\newcommand{\LNCS}{Lecture Notes in Computer Science}
	\newcommand{\STTT}{International Journal on Software Tools for Technology Transfer}
	\newcommand{\ToPNoC}{Transactions on Petri Nets and Other Models of Concurrency}
\else
	\bibliographystyle{splncs04} 
	\newcommand{\IJFCS}{International Journal of Foundations of Computer Science}
	\newcommand{\JLAP}{Journal of Logic and Algebraic Programming}
	\newcommand{\LNCS}{Lecture Notes in Computer Science}
	\newcommand{\STTT}{International Journal on Software Tools for Technology Transfer}
	\newcommand{\ToPNoC}{Transactions on Petri Nets and Other Models of Concurrency}
\fi
\bibliography{dblp_refs}

\LongVersion{\newpage
\appendix
\section{Omitted Proof}

\begin{definition}
 [path value]
 For a WTTS $\WTTSWithInside$, a sequence $\WTTSstate_0,\WTTSstate_1,\dots,\WTTSstate_n$ of $\WTTSState$ is a \emph{path} of $\WTTS$
 if we have $\WTTSstate_0\WTTSTransitionRel\WTTSstate_1\WTTSTransitionRel\dots\WTTSTransitionRel\WTTSstate_n$.
 For a WTTS $\WTTSWithInside$ and a path $\WTTSPath = \WTTSstate_0,\WTTSstate_1,\dots,\WTTSstate_n$ of $\WTTS$, the \emph{path value} is 
 $\pathValue(\WTTSPath) = \bigotimes_{i=1}^n \WTTSWeight(\WTTSstate_{i-1},\WTTSstate_{i})$.
\end{definition}

For any WTTS $\WTTSWithInside$, we have
 $\traceValue(\WTTS) = \bigoplus_{\WTTSPath\in\ARun(\WTTS)} \pathValue(\WTTSPath)$, where $\ARun(\WTTS)$ is the set of paths of $\WTTSstate_0,\WTTSstate_1,\dots,\WTTSstate_n$ of $\WTTS$ satisfying $\WTTSstate_0\in\InitWTTSState$ and $\WTTSstate_n\in\AccWTTSState$.

\subsection{Finiteness of the reachable part of WSTTSs}
\label{subsec:finiteness_proof}

For a WSTTS $\WSTTSWithInside$, we denote the reachable set by $\reach = \WSTTSInitState \cup \{\WSTTSstate \in \WSTTSState\mid \exists \WSTTSstate_0 \in \WSTTSInitState,\WSTTSstate_1,\WSTTSstate_2,\dots,\WSTTSstate_m \in \WSTTSState.\,\WSTTSstate_0,\WSTTSstate_1,\dots,\WSTTSstate_m,\WSTTSstate \text{ is a path of $\WSTTS$} \}$. 

\begin{lemma}
 \label{lemma:zone_bounded}
 Let $\WSTTSWithInside$ be a WSTTS of a signal $\signalWithInside$ and a TSWA $\TSWA$.
 For any $(\loc,\zone,\dvalSeq) \in \reach$, $\cval\in\zone$, and $\clock\in\Clock$, we have $0 \le \cval(\clock) \leq \duration{\signal}$.
\end{lemma}
\begin{proof}
 Since for any $(\loc,\zone,\dvalSeq)\in\WSTTSState$ and $\cval\in\zone$, we have $\cval(\absClock) \leq \duration{\signal}$, it suffice to prove $\cval(\clock) \le \cval(\absClock)$ for any $\clock\in\Clock$.
  Let $(\loc,\zone,\dvalSeq) \in \reach$.
 If $(\loc,\zone,\dvalSeq) \in \WSTTSInitState$, we have $\zone = \zerovalue[\ClockWithAbs]$ and for any $\cval\in\zone$ and for any $\clock\in\Clock$, we have $\cval(\clock) = \cval(\absClock) = 0$.

 Assume $(\loc,\zone,\dvalSeq) \not\in \WSTTSInitState$ and let $(\loc',\zone',\dvalSeq') \in \reach$ satisfying $\bigl( (\loc',\zone',\dvalSeq'), (\loc,\zone,\dvalSeq)\bigr) \in \WSTTSTransition$.
 If $\dvalSeq = \varepsilon$, there exists $(\loc,\guard,\resets,\loc')\in\Transition$ satisfying $\zone= \{ \cval'[\resets:=0] \mid \cval'\in\zone', \cval'\models\guard \}$.
 Since $\absClock\not\in\rho$, $\forall \cval'\in\zone',\clock\in\Clock.\, \cval'(\clock) \leq \cval'(\absClock)$ implies
 $\forall \cval\in\zone,\clock\in\Clock.\, \cval(\clock) \leq \cval(\absClock)$.

 If $\dvalSeq \neq \varepsilon$, for any $\cval \in\zone$, there are $\cval'\in\zone'$ and $\relativeTime\in\Rp$ satisfying $\cval = \cval' + \relativeTime$.
 Therefore, $\forall \cval'\in\zone',\clock\in\Clock.\, \cval'(\clock) \leq \cval'(\absClock)$ implies
 $\forall \cval\in\zone,\clock\in\Clock.\, \cval(\clock) \leq \cval(\absClock)$.
 \qed
\end{proof}

 For a nonempty zone $\zone\in\ZonesWithAbs$ and $\clock,\clock'\in \Clock \amalg \{\absClock,0\}$, we define ${\prec}_{\zone,\clock,\clock'} \in\{{<},{\leq}\}$ and $d_{\zone,\clock,\clock'} \in \R \amalg \{\infty\}$ be the smallest elements satisfying the following, where we define ${<}$ is smaller than ${\leq}$ and we denote $\cval(0) = 0$.
 \begin{displaymath}
  \zone = \Bigl\{\cval \, \Bigm|\, \bigwedge_{\clock,\clock' \in \Clock \amalg \{\absClock,0\}} (\cval(\clock) - \cval(\clock')) \prec_{\zone,\clock,\clock'} d_{\zone,\clock,\clock'}\Bigr\}
 \end{displaymath}

\begin{lemma}
\label{lemma:zone_discrete}
 Let $\WSTTSWithInside$ be a WSTTS of a signal $\signalWithInside$ and a TSWA $\TSWA$.
 For any $(\loc,\zone,\dvalSeq) \in \reach$ and $\clock,\clock'\in\Clock\amalg\{\absClock,0\}$, 
 we have $d_{\zone,\clock,\clock'} = \infty$ or
 there is $k_i \in \Z$ for each $i\in\{0,1,\dots,n\}$ satisfying $d_{\zone,\clock,\clock'} = k_0 + \sum_{i=1}^{n} (k_i (\sum_{j=1}^{i} \tau_j))$.
\end{lemma}
\begin{proof}
 If $(\loc,\zone,\dvalSeq) \in \WSTTSInitState$, we have $\zone = \zerovalue[\ClockWithAbs]$ and we have $d_{\zone,\clock,\clock'} = 0$ for each $\clock,\clock'\in\Clock\amalg\{\absClock,0\}$.

 Assume $(\loc,\zone,\dvalSeq) \not\in \WSTTSInitState$ and let $(\loc',\zone',\dvalSeq') \in \reach$ satisfying $\bigl( (\loc',\zone',\dvalSeq'), (\loc,\zone,\dvalSeq)\bigr) \in \WSTTSTransition$.
 If $\dvalSeq = \varepsilon$, there exists $(\loc,\guard,\resets,\loc')\in\Transition$ satisfying $\zone= \{ \cval'[\resets:=0] \mid \cval'\in\zone', \cval'\models\guard \}$.
 For each $\clock\in\rho$, we have $d_{\zone,\clock,0} = d_{\zone,0,\clock} = 0$.
 For each $\clock,\clock' \in \Clock \amalg \{\absClock,0\}$, $d_{\zone,\clock,\clock'}$ is the shortest distance in the graph interpretation of $\zone'$, 
 where for each $\clock\in\rho$, $d_{\zone',\clock,0}$ and $d_{\zone,0,\clock}$ are replaced with $0$. (See e.g.,~\cite{DBLP:conf/ac/BengtssonY03} for the graph interpretation of a zone.)
 Therefore, for each $\clock,\clock'\in\Clock\amalg\{\absClock,0\}$, there are $k_{\clock'',\clock'''} \in\Z$ satisfying $d_{\zone,\clock,\clock'} = \sum_{\clock'',\clock''' \in \Clock\amalg\{\absClock,0\}} k_{\clock'',\clock'''} d_{\zone',\clock'',\clock'''}$.
 By induction hypothesis, 
 For any $\clock,\clock'\in\Clock\amalg\{\absClock,0\}$, 
 we have $d_{\zone,\clock,\clock'} = \infty$ or
 there is $k_i \in \Z$ for each $i\in\{0,1,\dots,n\}$ satisfying $d_{\zone,\clock,\clock'} = k_0 + \sum_{i=1}^{n} (k_i (\sum_{j=1}^{i} \tau_j))$.

 If $\dvalSeq \neq \varepsilon$, $d_{\zone,\clock,\clock'}$ are computed by the following procedure.
 \begin{enumerate}
  \item For each $\clock\in\Clock$, we replace $(d_{\zone',\clock,0},{\prec}_{\zone',\clock,0})$ with $(\infty,{<})$.
  \item We replace $(d_{\zone',\absClock,0},{\prec}_{\zone',\clock,0})$ and $(d_{\zone',0,\absClock},{\prec}_{\zone',0,\clock})$ with 
        $\bigl(\sum_{j=0}^i\tau_j,{<}\bigr)$ and $\bigl(\sum_{j=0}^{i-1}\tau_j,{<}\bigr)$, or
        $\bigl(\sum_{j=0}^i\tau_j,{\le}\bigr)$ and $\bigl(\sum_{j=0}^{i}\tau_j,{\le}\bigr)$, respectively.
  \item We take the shortest distance in the graph interpretation of $\zone'$, where some weights are replaced by the above.
 \end{enumerate}
 Therefore, for each $\clock,\clock'\in\Clock\amalg\{\absClock,0\}$ and for each $i\in\{0,1,\dots,n\}$, there are $k_{\clock'',\clock'''} \in\Z$ and $k_i \in \Z$ satisfying the following.
 \[
  d_{\zone,\clock,\clock'} = \sum_{\clock'',\clock''' \in \Clock\amalg\{\absClock,0\}} k_{\clock'',\clock'''} d_{\zone',\clock'',\clock'''} + \sum_{i=1}^{n} (k_i (\sum_{j=1}^{i} \tau_j))
 \]
 By induction hypothesis, 
 For any $\clock,\clock'\in\Clock\amalg\{\absClock,0\}$, 
 we have $d_{\zone,\clock,\clock'} = \infty$ or
 there is $k_i \in \Z$ for each $i\in\{0,1,\dots,n\}$ satisfying $d_{\zone,\clock,\clock'} = k_0 + \sum_{i=1}^{n} (k_i (\sum_{j=1}^{i} \tau_j))$.
\qed
\end{proof}

\begin{lemma}
\label{lemma:dvalSeq_substring}
 For any WSTTS $\WSTTSWithInside$ and
 for any $(\loc,\zone,\dvalSeq) \in \reach$, either $\dvalSeq = \varepsilon$ holds or there is $\absoluteTime\in\Rnn$ such that for any $\cval\in\zone$, $\dvalSeq = \values(\signal([t,\cval(\absClock))))$ holds.
\end{lemma}
\begin{proof}
 Let $(\loc,\zone,\dvalSeq) \in \reach$.
 If $(\loc,\zone,\dvalSeq) \in \WSTTSInitState$, we have $\dvalSeq = \varepsilon$.

 If $(\loc,\zone,\dvalSeq) \not\in \WSTTSInitState$, there is $(\loc',\zone',\dvalSeq') \in \reach$ such that $\bigl( (\loc',\zone',\dvalSeq'), (\loc,\zone,\dvalSeq)\bigr) \in \WSTTSTransition$.
 If $\dvalSeq \neq \varepsilon$, we have $\dvalSeq = \dvalSeq' \absConcat \values(\signal([\cval(\absClock),\cval'(\absClock))))$.
 By induction hypothesis, there is $\cval'\in\zone'$ such that for any $\cval\in\zone$, we have $\dvalSeq = \values(\signal([\cval'(\absClock),\cval(\absClock))))$ or 
 there is $\absoluteTime\in\Rnn$ such that for any $\cval\in\zone$, $\dvalSeq = \values(\signal([t,\cval(\absClock))))$ holds.
 \qed
\end{proof}

\begin{theorem}
 [finiteness]
 \label{theorem:finiteness}
 For any WSTTS $\WSTTS$, there are only finitely many states reachable from $\WSTTSInitState$.
\end{theorem}

\begin{proof}
 The locations $\Loc$ is a finite set.
 By \cref{lemma:zone_bounded} and \cref{lemma:zone_discrete}, the number of zones appearing in $\reach$ is finitely many.
 By \cref{lemma:dvalSeq_substring}, the subsequences $\dvalSeq$ appearing in $\reach$ is finitely many.
 Therefore, $\reach$ is a finite set.
 \qed
\end{proof}

\subsection{Proof of \cref{corollary:trace_value_correctness}}

First, we define symbolic path value. 

\begin{definition}
 [symbolic path value]
 For a WSTTS $\WSTTSWithInside$, a sequence $\WSTTSstate_0,\WSTTSstate_1,\dots,\WSTTSstate_n$ of $\WSTTSState$ is a \emph{path} of $\WSTTS$
 if for any $i\in\{1,\dots,n\}$, we have $(\WSTTSstate_{i-1},\WSTTSstate_{i}) \in \WSTTSTransition$.
 For a WSTTS $\WSTTSWithInside$ and a path $\WSTTSPath=\WSTTSstate_0, \WSTTSstate_1, \dots \WSTTSstate_n$ of $\WSTTS$, the \emph{symbolic path value} is
 \begin{math}
  \symbolicPathValue(\WSTTSPath)= \bigotimes_{i=1}^n \WSTTSWeight(\WSTTSstate_{i-1},\WSTTSstate_{i})
 \end{math}.
\end{definition}

Similarly to the trace value $\traceValue(\WTTS)$, for any WSTTS $\WSTTSWithInside$, we have 
$\symbolicTraceValue(\WSTTS) = \bigoplus_{\WSTTSPath\in\ARun(\WSTTS)} \symbolicPathValue(\WSTTSPath)$,
where $\ARun(\WSTTS)$ is the set of paths of $\WSTTSstate_0,\WSTTSstate_1,\dots,\WSTTSstate_n$ of $\WSTTS$ satisfying $\WSTTSstate_0\in\WSTTSInitState$ and $\WSTTSstate_n\in\WSTTSAccState$.

For a semiring $\semiringWithInside$, we denote the canonical order by ${\preceq} \subseteq \semiringBase\times \semiringBase$, where $\semiringElem\preceq \semiringElem' \iff \semiringElem \semiringPlus \semiringElem'=\semiringElem'$. 
When $\semiring$ is idempotent, $\semiringElem=\semiringElem'$ if and only of $\semiringElem \preceq \semiringElem'$ and $\semiringElem' \preceq \semiringElem$ because:
if $\semiringElem = \semiringElem'$, we have
$\semiringElem \semiringPlus \semiringElem'=\semiringElem' \semiringPlus \semiringElem' = \semiringElem'$ and
$\semiringElem' \semiringPlus \semiringElem=\semiringElem \semiringPlus \semiringElem = \semiringElem$; and
if $\semiringElem \preceq \semiringElem'$ and $\semiringElem' \preceq \semiringElem$, we have $\semiringElem = \semiringElem \semiringPlus \semiringElem'= \semiringElem' \semiringPlus \semiringElem = \semiringElem'$.

For simplicity, we assume that for any signal $\signal = a_{1}^{\tau_{1}} a_{2}^{\tau_{2}} \dots a_{n}^{\tau_{n}} \in\signals$ and for any $i\in\{1,2,\dots,n-1\}$, we have $a_i \neq a_{i+1}$.

\begin{lemma}
 \label{lemma:WSTTS_widening}
 Let $\signalWithInside\in\signals$ be a signal, let $\TSWA$ be a TSWA, and let $\WSTTSWithInside$ be a WSTTS of $\signal$ and $\TSWA$.
 For any $(\loc,\zone_1,\dvalSeq), (\loc',\zone_1',\dvalSeq[']) \in \WSTTSState$ if $\bigl((\loc,\zone_1,\dvalSeq), (\loc',\zone_1',\dvalSeq['])\bigr) \in \WSTTSTransition$ holds, 
 for any $(\loc,\zone_2,\dvalSeq)\in \WSTTSState$ satisfying $\zone_1 \subseteq \zone_2$, there exists $(\loc',\zone_2',\dvalSeq[']) \in\WSTTSState$ satisfying $\zone_1' \subseteq \zone_2'$ and $\bigl((\loc,\zone_2,\dvalSeq), (\loc',\zone_2',\dvalSeq['])\bigr) \in \WSTTSTransition$.
\end{lemma}
\begin{proof}
 If $\dvalSeq['] = \varepsilon$, there exists $(\loc,\guard,\resets,\loc')\in\Transition$ satisfying $\zone_1'= \{ \cval[\resets:=0] \mid \cval\in\zone_1, \cval\models\guard \}$.
 Since $\zone_1 \subseteq \zone_2$, $\zone_2'= \{ \cval[\resets:=0] \mid \cval\in\zone_2, \cval\models\guard \}$ is nonempty and we have $\bigl((\loc,\zone_2,\dvalSeq), (\loc',\zone_2',\dvalSeq['])\bigr) \in \WSTTSTransition$.
 We also have $\zone'_1 \subseteq \zone'_2$.

 If $\dvalSeq['] \neq \varepsilon$, let $M$ be either $M_{i,=} = \{ \cval \mid \cval(\absClock) = \sum_{j=0}^{i} \tau_j\}$ or $M_{i} = \{ \cval \mid \sum_{j=0}^{i-1} \tau_j < \cval(\absClock) < \sum_{j=0}^{i} \tau_j\}$ satisfying $\zone_1' = \{\cval + \relativeTime \mid \cval + \zone_1, \relativeTime\in\Rp\} \cap M$, where $i\in\{1,2,\dots,n\}$.
Let $\zone_2' = \{\cval + \relativeTime \mid \cval \in \zone_2, \relativeTime\in\Rp\} \cap M$.
 Then, we have $\bigl((\loc,\zone_2,\dvalSeq), (\loc',\zone_2',\dvalSeq['])\bigr) \in \WSTTSTransition$ and since $\zone_1 \subseteq \zone_2$, we have $\zone_1' \subseteq \zone_2'$.
 \qed
\end{proof}

\begin{lemma}
 \label{lemma:WTTS_WSTTS_simulation_path_step}
 Let $\signalWithInside\in\signals$ be a signal and let $\TSWAWithInside$ be a TSWA, where $\TSAWithInside$.
 Let $\WTTS$ and $\WSTTS$ be the WTTS and WSTTS of $\signal$ and $\TSWA$, respectively.
 For any $(\loc,\cval,\absoluteTime,\dvalSeq) \WTTSTransitionRel (\loc',\cval',\absoluteTime',\dvalSeq['])$, there is a zone $\zone'\in\ZonesWithAbs$ satisfying the following.
 \begin{itemize}
  \item $\bigl( (\loc,\{\cval_{\zone}\},\dvalSeq), (\loc',\zone',\dvalSeq[']) \bigr) \in \WSTTSTransition$, where $\cval_{\zone}\in\clockvaluations[\ClockWithAbs]$ is for any $\clock\in\Clock$, $\cval_{\zone}(\clock) = \cval(\clock)$ and $\cval_{\zone}(\absClock) = \absoluteTime$.
  \item There exists $\cval'_{\zone} \in\zone'$ such that for any $\clock\in\Clock$ $\cval'_{\zone}(\clock) = \cval'(\clock)$ and $\cval'_{\zone}(\absClock) = \absoluteTime'$.
  \item $\WTTSWeight\bigl(((\loc,\cval,\absoluteTime,\dvalSeq), (\loc',\cval',\absoluteTime',\dvalSeq[']))\bigr) = \WSTTSWeight\bigl(((\loc,\{\cval_{\zone}\},\dvalSeq), (\loc',\zone',\dvalSeq[']))\bigr)$
 \end{itemize}
\end{lemma}
\begin{proof}
 If $\dvalSeq['] = \varepsilon$, there is $(\loc,\guard,\resets,\loc')\in\Transition$ satisfying
 $\cval\models\guard$, 
 $\cval'=\cval[\resets:=0]$, 
 $\absoluteTime'=\absoluteTime$, 
 $\dvalSeq \neq \varepsilon$, and
 $\dvalSeq['] = \varepsilon$.
 Let $\zone'$ be $\zone'=\{\cval_{\zone}[\resets:=0]\}$.
 Since $\cval_{\zone} \models \guard$,
 $\dvalSeq \neq \varepsilon$, and
 $\dvalSeq['] = \varepsilon$, 
 we have $\bigl((\loc,\{\cval_{\zone}\}, \dvalSeq), (\loc',\zone', \dvalSeq['])\bigr)$.
 Since $\cval'=\cval[\resets:=0]$,  for any $\clock\in\Clock$, we have $(\cval_{\zone}[\resets:=0])(\clock) = \cval'(\clock)$.
 Since $\absoluteTime = \absoluteTime'$ and $\absClock \not\in \resets$, we have $(\cval_{\zone}[\resets:=0])(\absClock) = \cval_{\zone}(\absClock) = \absoluteTime = \absoluteTime'$.
 We also have $\WTTSWeight\bigl(((\loc,\cval,\absoluteTime,\dvalSeq), (\loc',\cval',\absoluteTime',\dvalSeq[']))\bigr) = \costFunc(\Label(\loc),\dvalSeq) = \WSTTSWeight\bigl(((\loc,\{\cval_{\zone}\},\dvalSeq), (\loc',\zone',\dvalSeq[']))\bigr)$.

 If $\dvalSeq['] \neq \varepsilon$, $\loc = \loc'$ and there is $\relativeTime\in\Rp$ satisfying $\cval'=\cval + \relativeTime$, $\absoluteTime' = \absoluteTime + \relativeTime$, and $\dvalSeq['] = \dvalSeq \absConcat \signal( [\absoluteTime,\absoluteTime + \relativeTime) )$.
 Let $M$ be either $M_{i,=} = \{ \cval \mid \cval(\absClock) = \sum_{j=0}^{i} \tau_j\}$ or $M_{i} = \{ \cval \mid \sum_{j=0}^{i-1} \tau_j < \cval(\absClock) < \sum_{j=0}^{i} \tau_j\}$ satisfying $\absoluteTime' \in M$, where $i\in\{1,2,\dots,n\}$.
Let $\zone' = \{\cval + \relativeTime \mid \cval \in \zone, \relativeTime\in\Rp\} \cap M$.
We have $\bigl( (\loc,\{\cval_{\zone}\},\dvalSeq), (\loc',\zone',\dvalSeq[']) \bigr) \in \WSTTSTransition$.
 Since $\cval'=\cval + \relativeTime$ and $\absoluteTime' \in M$, there exists $\cval'_{\zone} \in\zone'$ such that for any $\clock\in\Clock$ $\cval'_{\zone}(\clock) = \cval'(\clock)$ and $\cval'_{\zone}(\absClock) = \absoluteTime'$.
 We also have $\WTTSWeight\bigl(((\loc,\cval,\absoluteTime,\dvalSeq), (\loc',\cval',\absoluteTime',\dvalSeq[']))\bigr) = \semiringTimesUnit = \WSTTSWeight\bigl(((\loc,\{\cval_{\zone}\},\dvalSeq), (\loc',\zone',\dvalSeq[']))\bigr)$.
 \qed
\end{proof}

\begin{lemma}
 \label{lemma:WTTS_WSTTS_simulation_path}
 Let $\signal\in\signals$ be a signal and let $\TSWAWithInside$ be a TSWA, where $\TSAWithInside$.
 Let $\WTTS$ and $\WSTTS$ be the WTTS and WSTTS of $\signal$ and $\TSWA$, respectively.
 For any path $\WTTSPath = (\loc_0,\cval_0,\absoluteTime_0,\dvalSeqi{0}),(\loc_1,\cval_1,\absoluteTime_1,\dvalSeqi{1}),\dots,(\loc_n,\cval_n,\absoluteTime_n,\dvalSeqi{n})$ of $\WTTS$, 
 there is a path $\WSTTSPath = (\loc_0,\zone_0,\dvalSeqi{0}),(\loc_1,\zone_1,\dvalSeqi{1}),\dots,(\loc_n,\zone_n,\dvalSeqi{n})$ of $\WTTS$, such that
 $\zone_0 = \{ \cval_{\zone,0} \mid \forall \clock\in\Clock.\,\cval_{\zone,0}(\clock)=\cval_{0}(\clock), \cval_{\zone,0}(\absClock) = \absoluteTime_0\}$ and 
 for any $i \in\{1,2,\dots,n\}$, there exists $\cval_{\zone,i} \in \zone_i$ satisfying $\cval_{\zone,i}(\clock)=\cval_{i}(\clock)$ for any $\clock\in\Clock$ and $\cval_{\zone,i}(\absClock) = \absoluteTime_i$.
\end{lemma}
\begin{proof}
 We prove the lemma by induction on $n$.

 When $n=1$, by \cref{lemma:WTTS_WSTTS_simulation_path_step}, for $\zone_{0} = \{\cval_{\zone,0} \in\clockvaluations[\ClockWithAbs] \mid \forall \clock\in\Clock.\,\cval_{\zone,0}(\clock) = \cval_{0}(\clock), \cval_{\zone,0}(\absClock) = \absoluteTime_0\}$ there is a zone $\zone_{1}\in\ZonesWithAbs$ satisfying:
 \begin{itemize}
  \item $\bigl((\loc_0,\zone_0,\dvalSeqi{0}),(\loc_1,\zone_1,\dvalSeqi{1})\bigr)\in\WSTTSTransition$; and
  \item there exists $\cval_{\zone,1} \in \zone_1$ satisfying $\cval_{\zone,1}(\clock)=\cval_{1}(\clock)$ for any $\clock\in\Clock$ and $\cval_{\zone,1}(\absClock) = \absoluteTime_1$.
 \end{itemize}

 When $n > 1$, by \cref{lemma:WTTS_WSTTS_simulation_path_step}, for $\zone'_{n-1} = \{\cval_{\zone,n-1} \in\clockvaluations[\ClockWithAbs] \mid \forall \clock\in\Clock.\,\cval_{\zone,n-1}(\clock) = \cval_{n-1}(\clock), \cval_{\zone,n-1}(\absClock) = \absoluteTime_{n-1}\}$ there is a zone $\zone'_{n}\in\ZonesWithAbs$ satisfying:
 \begin{itemize}
  \item $\bigl((\loc_{n-1},\zone'_{n-1},\dvalSeqi{n-1}),(\loc_{n},\zone'_{n},\dvalSeqi{n})\bigr)\in\WSTTSTransition$; and
  \item there exists $\cval_{\zone,n} \in \zone'_{n}$ satisfying $\cval_{\zone,n}(\clock)=\cval_{n}(\clock)$ for any $\clock\in\Clock$ and $\cval_{\zone,n}(\absClock) = \absoluteTime_{n}$.
 \end{itemize}
 By induction hypothesis, there is a path $(\loc_0,\zone_0,\dvalSeqi{0}),(\loc_1,\zone_1,\dvalSeqi{1}),\dots,(\loc_{n-1},\zone_{n-1},\dvalSeqi{n-1})$ of $\WSTTS$, such that
 $\zone_0 = \{ \cval_{\zone,0} \mid \forall \clock\in\Clock.\,\cval_{\zone,0}(\clock)=\cval_{0}(\clock), \cval_{\zone,0}(\absClock) = \absoluteTime_0\}$ and 
 for any $i \in\{1,2,\dots,n-1\}$, there exists $\cval_{\zone,i} \in \zone_i$ satisfying $\cval_{\zone,i}(\clock)=\cval_{i}(\clock)$ for any $\clock\in\Clock$ and $\cval_{\zone,i}(\absClock) = \absoluteTime_i$.
 Since $\zone'_{n-1} \subseteq \zone_{n-1}$ and \cref{lemma:WSTTS_widening}, there exists $\zone_{n} \in\ZonesWithAbs$ satisfying $\zone'_{n} \subseteq \zone_{n}$ and $\bigl((\loc_{n-1},\zone_{n-1},\dvalSeqi{n-1}),(\loc_{n},\zone_{n},\dvalSeqi{n})\bigr)\in\WSTTSTransition$.
 Therefore, $(\loc_0,\zone_0,\dvalSeqi{0}),(\loc_1,\zone_1,\dvalSeqi{1}),\dots,(\loc_n,\zone_n,\dvalSeqi{n})$ is a path of $\WSTTS$, such that
 $\zone_0 = \{ \cval_{\zone,0} \mid \forall \clock\in\Clock.\,\cval_{\zone,0}(\clock)=\cval_{0}(\clock), \cval_{\zone,0}(\absClock) = \absoluteTime_0\}$ and 
 for any $i \in\{1,2,\dots,n\}$, there exists $\cval_{\zone,i} \in \zone_i$ satisfying $\cval_{\zone,i}(\clock)=\cval_{i}(\clock)$ for any $\clock\in\Clock$ and $\cval_{\zone,i}(\absClock) = \absoluteTime_i$.
 \qed
\end{proof}

\begin{lemma}
 \label{lemma:WTTS_WSTTS_simulation_path_value} 
 Let $\signal\in\signals$ be a signal and let $\TSWAWithInside$ be a TSWA, where $\TSAWithInside$.
 Let $\WTTS$ and $\WSTTS$ be the WTTS and WSTTS of $\signal$ and $\TSWA$, respectively.
 For any path $\WTTSPath$ of $\WTTS$, 
 there is a path $\WSTTSPath$ of $\WSTTS$ satisfying $\pathValue(\WTTSPath) = \symbolicPathValue(\WSTTSPath)$.
 Moreover, for any $\WTTSPath\in\ARun(\WTTS)$,
 there is $\WSTTSPath\in\ARun(\WSTTS)$ satisfying $\pathValue(\WTTSPath) = \symbolicPathValue(\WSTTSPath)$.
\end{lemma}
\begin{proof}
 By \cref{lemma:WTTS_WSTTS_simulation_path},  
 for any path $\WTTSPath = (\loc_0,\cval_0,\absoluteTime_0,\dvalSeqi{0}),(\loc_1,\cval_1,\absoluteTime_1,\dvalSeqi{1}),\dots,(\loc_n,\cval_n,\absoluteTime_n,\dvalSeqi{n})$ of $\WTTS$, 
 there is a path $\WSTTSPath = (\loc_0,\zone_0,\dvalSeqi{0}),(\loc_1,\zone_1,\dvalSeqi{1}),\dots,(\loc_n,\zone_n,\dvalSeqi{n})$ of $\WTTS$.
 For any $i\in\{1,2,\dots,n\}$, we have 
 \begin{align*}
  &\WTTSWeight\bigl( (\loc_{i-1},\cval_{i-1},\absoluteTime_{i-1},\dvalSeqi{{i-1}}),(\loc_i,\cval_i,\absoluteTime_i,\dvalSeqi{i}) \bigr)\\
  =&
  \begin{cases}
   \costFunc(\Label(\loc_{i-1},\dvalSeqi{i-1})) & \text{if $\dvalSeqi{i} = \varepsilon$}\\
   \semiringTimesUnit & \text{if $\dvalSeqi{i} \neq \varepsilon$}\\
  \end{cases}\\
  =& \WSTTSWeight\bigl( (\loc_{i-1},\zone_{i-1},\dvalSeqi{{i-1}}),(\loc_i,\zone_i,\dvalSeqi{i}) \bigr)\\
 \end{align*}
 Therefore, we have $\pathValue(\WTTSPath) = \symbolicPathValue(\WSTTSPath)$.

 When $\WTTSPath\in\ARun(\WTTS)$, we have
 $\loc_0\in\InitLoc$, 
 $\cval_0 = \zerovalue$, 
 $\absoluteTime_0 = 0$,
 $\dvalSeqi{0} = \varepsilon$,
 $\loc_n\in\AccLoc$,
 $\absoluteTime_n = \duration{\signal}$, and
 $\dvalSeqi{n} = \varepsilon$.
 By \cref{lemma:WTTS_WSTTS_simulation_path}, we have $\zone_0 = \{\zerovalue[\ClockWithAbs]\}$ and 
 there is $\cval_{\zone,n} \in \zone_n$ satisfying $\cval_{\zone,n}(\absClock) = \absoluteTime_n$.
 Therefore $\WSTTSPath\in\ARun(\WSTTS)$ also holds.
 \qed
\end{proof}

\begin{theorem}
 \label{theorem:WTTS_WSTTS_simulation_trace_value} 
 Let $\signal\in\signals$ be a signal and let $\TSWAWithInside$ be a TSWA, where $\TSAWithInside$.
 Let $\WTTS$ and $\WSTTS$ be the WTTS and WSTTS of $\signal$ and $\TSWA$, respectively.
 If the semiring $\semiring$ is idempotent,
 we have $\traceValue(\WTTS) \preceq \symbolicTraceValue(\WSTTS)$.
\end{theorem}
\begin{proof}
 By \cref{lemma:WTTS_WSTTS_simulation_path_value}, there is a mapping $f\colon\ARun(\WTTS) \to \ARun(\WSTTS)$ satisfying $\WSTTSWeight(f(\WTTSPath)) = \WTTSWeight(\WTTSPath)$.
 We have
 \begin{align*}
  & \traceValue(\WTTS) \semiringPlus  \symbolicTraceValue(\WSTTS)\\
  =& \bigl( \bigoplus_{\WTTSPath\in\ARun(\WTTS)} \pathValue(\WTTSPath) \bigr) \semiringPlus \bigl(\bigoplus_{\WSTTSPath\in\ARun(\WSTTS)} \symbolicPathValue(\WSTTSPath)\bigr)\\
  =& \bigl( \bigoplus_{\WSTTSPath\in f(\ARun(\WTTS))} \symbolicPathValue(\WSTTSPath) \bigr) \semiringPlus \bigl(\bigoplus_{\WSTTSPath\in\ARun(\WSTTS)} \symbolicPathValue(\WSTTSPath)\bigr)\\
  =& \bigl(\bigoplus_{\WSTTSPath\in\ARun(\WSTTS)} \symbolicPathValue(\WSTTSPath)\bigr) = \symbolicTraceValue(\WSTTS)
 \end{align*}
 Therefore, we have $\traceValue(\WTTS) \preceq \symbolicTraceValue(\WSTTS)$. 
 \qed
\end{proof}

\begin{lemma}
 \label{lemma:WSTTS_WTTS_simulation_path_step}
 Let $\signal\in\signals$ be a signal and let $\TSWAWithInside$ be a TSWA, where $\TSAWithInside$.
 Let $\WTTS$ and $\WSTTS$ be the WTTS and WSTTS of $\signal$ and $\TSWA$, respectively.
 For any $\bigl((\loc,\zone,\dvalSeq),  (\loc',\zone',\dvalSeq['])\bigr) \in \WSTTSTransition$, and $\cval'_{\zone} \in\zone'$,
 there is a clock valuation $\cval_{\zone}\in\zone$ satisfying the following.
 \begin{itemize}
  \item $(\loc,\project{\cval_{\zone}}{\Clock},\cval_{\zone}(\absClock),\dvalSeq) \WTTSTransitionRel (\loc',\project{\cval'_{\zone}}{\Clock},\cval'_{\zone}(\absClock),\dvalSeq['])$, where
        $\project{\cval_{\zone}}{\Clock},\project{\cval'_{\zone}}{\Clock}\in\clockvaluations$ are for any $\clock\in\Clock$,
        $\project{\cval_{\zone}}{\Clock}(\clock) = \cval_{\zone}(\clock)$ and 
        $\project{\cval'_{\zone}}{\Clock}(\clock) = \cval'_{\zone}(\clock)$.
  \item $\WTTSWeight\bigl(((\loc,\project{\cval_{\zone}}{\Clock},\cval_{\zone}(\absClock),\dvalSeq), (\loc',\project{\cval'_{\zone}}{\Clock},\cval'_{\zone}(\absClock),\dvalSeq[']))\bigr) = \WSTTSWeight\bigl(((\loc,\zone,\dvalSeq), (\loc',\zone',\dvalSeq[']))\bigr)$
 \end{itemize}
\end{lemma}
\begin{proof}
 If $\dvalSeq['] = \varepsilon$, we have 
 $\dvalSeq\neq\varepsilon$ and 
 there is $(\loc,\guard,\resets,\loc')\in\Transition$ satisfying
 $\zone'= \{ \cval[\resets:=0] \mid \cval\in\zone, \cval\models\guard \}$, which is nonempty.
 Because of $\zone'= \{ \cval[\resets:=0] \mid \cval\in\zone, \cval\models\guard \}$, there exists $\cval_{\zone}\in\zone$ such that $\cval'_{\zone} = \cval_{\zone}[\resets:=0]$.
 Since $\project{\cval_{\zone}}{\Clock} \models \guard$, we have 
 $(\loc,\project{\cval_{\zone}}{\Clock},\cval_{\zone}(\absClock),\dvalSeq) \WTTSTransitionRel (\loc',\project{\cval'_{\zone}}{\Clock},\cval'_{\zone}(\absClock),\dvalSeq['])$.
 We also have the following.
 \begin{displaymath}
  \WTTSWeight\bigl(((\loc,\project{\cval_{\zone}}{\Clock},\cval_{\zone}(\absClock),\dvalSeq), (\loc',\project{\cval'_{\zone}}{\Clock},\cval'_{\zone}(\absClock),\dvalSeq[']))\bigr) = \costFunc(\Label(\loc),\dvalSeq) = \WSTTSWeight\bigl(((\loc,\zone,\dvalSeq), (\loc',\zone',\dvalSeq[']))\bigr)  
 \end{displaymath}

 If $\dvalSeq['] \neq \varepsilon$,  
 we have $\loc=\loc'$, $\dvalSeq['] = \dvalSeq\absConcat\signal([\cval_{\zone}(\absClock),\cval'_{\zone}(\absClock)))$, and
 for any $\cval'\in\zone'$, there exists $\cval\in\zone$ and $\relativeTime\in\Rp$ satisfying $\cval' = \cval + \relativeTime$, where $\cval_{\zone}\in\zone$.
 Let $\cval_{\zone}\in\zone$ and $\relativeTime\in\Rp$ be such that $\cval_{\zone}' = \cval_{\zone} + \relativeTime$.
 Because of
\begin{itemize}
 \item $\loc = \loc'$, 
 \item $\project{\cval'_{\zone}}{\Clock}=\project{\cval_{\zone}}{\Clock} + \relativeTime$, 
 \item $\cval'_{\zone}(\absClock) = \cval_{\zone}(\absClock) + \relativeTime$, and
 \item  $\dvalSeq['] = \dvalSeq \absConcat \signal( [\cval_{\zone}(\absClock),\cval'_{\zone}(\absClock)) )$,
\end{itemize}
 we have
 \begin{math}
  (\loc,\project{\cval_{\zone}}{\Clock},\cval_{\zone}(\absClock),\dvalSeq) \WTTSTransitionRel (\loc',\project{\cval'_{\zone}}{\Clock},\cval'_{\zone}(\absClock),\dvalSeq['])  
 \end{math}.
 Thus we  also have 
 \begin{displaymath}
  \WTTSWeight\bigl(((\loc,\project{\cval_{\zone}}{\Clock},\cval_{\zone}(\absClock),\dvalSeq), (\loc',\project{\cval'_{\zone}}{\Clock},\cval'_{\zone}(\absClock),\dvalSeq[']))\bigr) = \semiringTimesUnit = \WSTTSWeight\bigl(((\loc,\zone,\dvalSeq), (\loc',\zone',\dvalSeq[']))\bigr)\enspace.
 \end{displaymath}
 \qed
\end{proof}

\begin{lemma}
 \label{lemma:WSTTS_WTTS_simulation_path}
 Let $\signal\in\signals$ be a signal and let $\TSWAWithInside$ be a TSWA, where $\TSAWithInside$.
 Let $\WTTS$ and $\WSTTS$ be the WTTS and WSTTS of $\signal$ and $\TSWA$, respectively.
 For any path $\WSTTSPath = (\loc_0,\zone_0,\dvalSeqi{0}),(\loc_1,\zone_1,\dvalSeqi{1}),\dots,(\loc_n,\zone_n,\dvalSeqi{n})$ of $\WTTS$
 and for any $\cval_{\zone,n} \in\zone_n$, 
 there is a path $\WTTSPath = (\loc_0,\cval_0,\absoluteTime_0,\dvalSeqi{0}),(\loc_1,\cval_1,\absoluteTime_1,\dvalSeqi{1}),\dots,(\loc_n,\cval_n,\absoluteTime_n,\dvalSeqi{n})$ of $\WTTS$ such that
 we have $\project{(\cval_{\zone,n})}{\Clock}=\cval_{n}$ and $\cval_{\zone,n}(\absClock) = \absoluteTime_n$, and
 for any $i \in\{0,1,\dots,n-1\}$, there exists $\cval_{\zone,i} \in \zone_i$ satisfying $\project{(\cval_{\zone,i})}{\Clock}=\cval_{i}$ and $\cval_{\zone,i}(\absClock) = \absoluteTime_i$.
\end{lemma}
\begin{proof}
 We prove the lemma by induction on $n$.
 When $n=1$, by \cref{lemma:WSTTS_WTTS_simulation_path_step}, for any $\cval_{\zone,1}\in\zone_1$, there is $\cval_{\zone,0}\in\zone_0$ satisfying 
$(\loc,\project{(\cval_{\zone,0})}{\Clock},\cval_{\zone,0}(\absClock),\dvalSeq) \WTTSTransitionRel (\loc',\project{(\cval_{\zone,1})}{\Clock},\cval_{\zone,1}(\absClock),\dvalSeq['])$.

 When $n > 1$,  by induction hypothesis, for any $\cval_{\zone,n} \in\zone_n$,  there is a path
 $(\loc_1,\cval_1,\absoluteTime_1,\dvalSeqi{1}),(\loc_2,\cval_2,\absoluteTime_2,\dvalSeqi{2}),\dots,(\loc_{n},\cval_{n},\absoluteTime_{n},\dvalSeqi{n})$ of $\WTTS$, such that
 we have $\project{(\cval_{\zone,n})}{\Clock}=\cval_{n}$ and $\cval_{\zone,n}(\absClock) = \absoluteTime_n$, and
 for any $i \in\{1,2,\dots,n-1\}$, there exists $\cval_{\zone,i} \in \zone_i$ satisfying $\project{(\cval_{\zone,i})}{\Clock}=\cval_{i}$ and $\cval_{\zone,i}(\absClock) = \absoluteTime_i$.
 By \cref{lemma:WSTTS_WTTS_simulation_path_step}, 
 there is a clock valuation $\cval_{\zone,0}\in\zone_0$ satisfying
 $(\loc,\project{(\cval_{\zone,0})}{\Clock},\cval_{\zone,0}(\absClock),\dvalSeq) \WTTSTransitionRel (\loc',\cval_1,\absoluteTime_1,\dvalSeq['])$.
 Therefore, $(\loc_0,\cval_{\zone,0},\cval_{\zone,0}(\absClock),\dvalSeqi{0}),(\loc_1,\cval_1,\absoluteTime_1,\dvalSeqi{1}),\dots,(\loc_n,\cval_n,\absoluteTime_n,\dvalSeqi{n})$ is a path of $\WTTS$, such that
 we have $\project{(\cval_{\zone,n})}{\Clock}=\cval_{n}$ and $\cval_{\zone,n}(\absClock) = \absoluteTime_n$, and
 for any $i \in\{1,2,\dots,n-1\}$, there exists $\cval_{\zone,i} \in \zone_i$ satisfying $\project{(\cval_{\zone,i})}{\Clock}=\cval_{i}$ and $\cval_{\zone,i}(\absClock) = \absoluteTime_i$.
 \qed
\end{proof}

\begin{lemma}
 \label{lemma:WSTTS_WTTS_simulation_path_value} 
 Let $\signal\in\signals$ be a signal and let $\TSWAWithInside$ be a TSWA, where $\TSAWithInside$.
 Let $\WTTS$ and $\WSTTS$ be the WTTS and WSTTS of $\signal$ and $\TSWA$, respectively.
 For any path $\WSTTSPath$ of $\WSTTS$, 
 there is a path $\WTTSPath$ of $\WTTS$ satisfying $\pathValue(\WTTSPath) = \symbolicPathValue(\WSTTSPath)$.
 Moreover, for any $\WSTTSPath\in\ARun(\WSTTS)$,
 there is $\WTTSPath\in\ARun(\WTTS)$ satisfying $\pathValue(\WTTSPath) = \symbolicPathValue(\WSTTSPath)$.
\end{lemma}
\begin{proof}
 By \cref{lemma:WSTTS_WTTS_simulation_path},  
 for any path 
 \begin{displaymath}
  \WSTTSPath = (\loc_0,\zone_0,\dvalSeqi{0}),(\loc_1,\zone_1,\dvalSeqi{1}),\dots,(\loc_n,\zone_n,\dvalSeqi{n})  
 \end{displaymath}
 of $\WSTTS$.
 there is a path
 \begin{displaymath}
  \WTTSPath = (\loc_0,\cval_0,\absoluteTime_0,\dvalSeqi{0}),(\loc_1,\cval_1,\absoluteTime_1,\dvalSeqi{1}),\dots,(\loc_n,\cval_n,\absoluteTime_n,\dvalSeqi{n}) 
 \end{displaymath}
 of $\WTTS$.
 For any $i\in\{1,2,\dots,n\}$, we have 
 \begin{align*}
  &\WTTSWeight\bigl( (\loc_{i-1},\cval_{i-1},\absoluteTime_{i-1},\dvalSeqi{{i-1}}),(\loc_i,\cval_i,\absoluteTime_i,\dvalSeqi{i}) \bigr)\\
  =&
  \begin{cases}
   \costFunc(\Label(\loc_{i-1},\dvalSeqi{i-1})) & \text{if $\dvalSeqi{i} = \varepsilon$}\\
   \semiringTimesUnit & \text{if $\dvalSeqi{i} \neq \varepsilon$}\\
  \end{cases}\\
  =& \WSTTSWeight\bigl( (\loc_{i-1},\zone_{i-1},\dvalSeqi{{i-1}}),(\loc_i,\zone_i,\dvalSeqi{i}) \bigr)\\
 \end{align*}
 Therefore, we have $\pathValue(\WTTSPath) = \symbolicPathValue(\WSTTSPath)$.

 When $\WSTTSPath\in\ARun(\WSTTS)$, we have
 $\loc_0\in\InitLoc$, 
 $\zone_0 = \zerovalue[\ClockWithAbs]$, 
 $\dvalSeqi{0} = \varepsilon$,
 $\loc_n\in\AccLoc$,
 $\exists \cval_{\zone,n}\in\zone_n.\, \cval_{\zone,n} = \duration{\signal}$, and
 $\dvalSeqi{n} = \varepsilon$.
 By \cref{lemma:WSTTS_WTTS_simulation_path}, for $\cval_{\zone,n}\in\zone_n$ satisfying $\cval_{\zone,n} = \duration{\signal}$,
 there is a path $\WTTSPath = (\loc_0,\cval_0,\absoluteTime_0,\dvalSeqi{0}),(\loc_1,\cval_1,\absoluteTime_1,\dvalSeqi{1}),\dots,(\loc_n,\cval_n,\absoluteTime_n,\dvalSeqi{n})$ of $\WTTS$ such that
 we have $\project{(\cval_{\zone,n})}{\Clock}=\cval_{n}$, $\cval_{\zone,n}(\absClock) = \absoluteTime_n$
 $\cval_{0} = \zerovalue$, $\absoluteTime_0 = 0$, and
 for any $i \in\{1,2,\dots,n-1\}$, there exists $\cval_{\zone,i} \in \zone_i$ satisfying $\project{(\cval_{\zone,i})}{\Clock}=\cval_{i}$ and $\cval_{\zone,i}(\absClock) = \absoluteTime_i$.
 Therefore $\WTTSPath\in\ARun(\WTTS)$ also holds.
 \qed
\end{proof}

\begin{theorem}
 \label{theorem:WSTTS_WTTS_simulation_trace_value}
 Let $\signal\in\signals$ be a signal and let $\TSWAWithInside$ be a TSWA, where $\TSAWithInside$.
 Let $\WTTS$ and $\WSTTS$ be the WTTS and WSTTS of $\signal$ and $\TSWA$, respectively.
 If the semiring $\semiring$ is idempotent,
 we have $\symbolicTraceValue(\WSTTS) \preceq \traceValue(\WTTS)$.
\end{theorem}
\begin{proof}
 By \cref{lemma:WSTTS_WTTS_simulation_path_value}, there is a mapping $f\colon\ARun(\WSTTS) \to \ARun(\WTTS)$ satisfying $\WTTSWeight(f(\WSTTSPath)) = \WSTTSWeight(\WSTTSPath)$.
 We have
 \begin{align*}
  & \symbolicTraceValue(\WSTTS) \semiringPlus \traceValue(\WTTS) \\
  =& \bigl(\bigoplus_{\WSTTSPath\in\ARun(\WSTTS)} \symbolicPathValue(\WSTTSPath)\bigr) \semiringPlus \bigl( \bigoplus_{\WTTSPath\in\ARun(\WTTS)} \pathValue(\WTTSPath) \bigr)\\
  =& \bigl(\bigoplus_{\WTTSPath\in f(\ARun(\WSTTS))} \pathValue(\WTTSPath)\bigr) \semiringPlus \bigl( \bigoplus_{\WTTSPath\in \ARun(\WTTS)} \pathValue(\WTTSPath) \bigr)\\
  =& \bigl(\bigoplus_{\WTTSPath\in\ARun(\WTTS)} \pathValue(\WTTSPath)\bigr) = \traceValue(\WTTS)
 \end{align*}
 Therefore, we have $\symbolicTraceValue(\WSTTS) \preceq \traceValue(\WTTS)$.
 \qed
\end{proof}

\begin{proof}
 [\cref{corollary:trace_value_correctness}]
 By \cref{theorem:WTTS_WSTTS_simulation_trace_value} and \cref{theorem:WSTTS_WTTS_simulation_trace_value}, we have $\traceValue(\WTTS) \preceq \symbolicTraceValue(\WSTTS)$. and $\symbolicTraceValue(\WSTTS) \preceq \traceValue(\WTTS)$. Therefore, we have the following.
 \begin{displaymath}
  \symbolicTraceValue(\WSTTS)
  = \traceValue(\WTTS) \semiringPlus \symbolicTraceValue(\WSTTS)
  = \symbolicTraceValue(\WSTTS) \semiringPlus \traceValue(\WTTS)
  = \traceValue(\WTTS)
 \end{displaymath}
 \qed
\end{proof}

\subsection{Proof of \cref{theorem:incremental_correctness}}

For locations $\loc,\loc'$ of $\TSAWithInside$, 
$\zone,\zone'\in\ZonesWithAbs$,
$\dvalSeq,\dvalSeq[']\in\DValSeq$, and
the WSTTS $\WSTTS$ of a signal $\signal$ and $\TSWAWithInside$,
We denote the set of paths from $(\loc,\zone,\dvalSeq)$ to $(\loc',\zone',\dvalSeq['])$ of $\WSTTS$ as follows.
\begin{align*}
&\Path(\WSTTS,\loc,\zone,\dvalSeq,\loc',\zone',\dvalSeq['])\\ =& 
 \{
 \WSTTSPath \mid
  \text{
 $(\loc,\zone,\dvalSeq),\WSTTSstate_1,\WSTTSstate_2,\dots,\WSTTSstate_n,(\loc',\zone',\dvalSeq['])$ is a path of $\WSTTS$
 }
 \}
\end{align*}

By symbolic path value $\symbolicPathValue(\WSTTSPath)$, we can rewrite the increment function $\incrementFunc(a,\absoluteTime)$ as follows.

\begin{align*}
  \incrementFunc(a,\absoluteTime)(\weights) = 
  \{(\loc',\zone',\dvalSeq['],\semiringElem')\mid &
  \loc'\in\Loc, 
  \zone'\in\ZonesWithAbs,
  \forall \nu'\in\zone'.\, \nu'(\absClock)= \absoluteTime,\\
  \dvalSeq[']\in\DValSeq,
 \semiringElem' &=
  \bigoplus_{
  \substack{(\loc,\zone,\dvalSeq,\semiringElem)\in w\\
  \WSTTSPath\in\Path(\WSTTS_{a,\absoluteTime},\loc,\zone,\dvalSeq,\loc',\zone',\dvalSeq['])
  }}
  \semiringElem \semiringTimes \symbolicPathValue(\WSTTSPath)
  \}  
\end{align*}

\begin{lemma}
 Let $\WSTTS$ be the WSTTS of a signal $\signalWithInside$ and a TSWA $\TSWAWithInside$, where $\TSAWithInside$.
 For any $i\in\{1,2,\dots,n\}$ and
 for any path
 \begin{displaymath}
  \WSTTSPath = (\loc_0,\{\zerovalue[\ClockWithAbs]\},\dvalSeqi{0}),(\loc_1,\zone_1,\dvalSeqi{1}),\dots,(\loc_{m},\zone_{m},\dvalSeqi{m}) 
 \end{displaymath}
 of $\WSTTS$,
 if for any $\cval_m \in \zone_m$, $\cval_m(\absClock) = \sum_{k=1}^{i} \tau_k$ holds,
 for any $j\in\{1,2,\dots,i-1\}$, 
 there exist $\loc'\in\Loc$, $\zone' \in \ZonesWithAbs$, and $\dvalSeq{'}\in\DValSeq$ such that 
 for any $\cval' \in \zone'$, we have $\cval'(\absClock) = \sum_{k=1}^{j} \tau_k$,
 and there are paths $\WSTTSPath' \in \Path(\WSTTS, \loc_0,\{\zerovalue[\ClockWithAbs]\},\dvalSeqi{0}, \loc',\zone',\dvalSeq['])$ and
 $\WSTTSPath'' \in \Path(\WSTTS, \loc',\zone',\dvalSeq['], \loc_m,\zone_m,\dvalSeqi{m})$ satisfying
 $\symbolicPathValue(\WSTTSPath) =  \symbolicPathValue(\WSTTSPath') \semiringTimes \symbolicPathValue(\WSTTSPath'')$.
\end{lemma}
\begin{proof}
 If there exists $h \in \{1,2,\dots,m\}$ such that for any $\cval_{h} \in \zone_h$, we have $\cval_{h}(\absClock) = \sum_{k=1}^{j} \tau_k$,
 we have the following.
 \begin{align*}
  \WSTTSPath' =& (\loc_0,\{\zerovalue[\ClockWithAbs]\},\dvalSeqi{0}),(\loc_1,\zone_1,\dvalSeqi{1}),\dots,(\loc_{h},\zone_{h},\dvalSeqi{h})\\
  \in& \Path(\WSTTS, \loc_0,\{\zerovalue[\ClockWithAbs]\},\dvalSeqi{0}, \loc',\zone',\dvalSeq['])\\
 \WSTTSPath'' =& (\loc_{h},\zone_{h},\dvalSeqi{h}),(\loc_{h+1},\zone_{h+1},\dvalSeqi{h+1}),\dots,(\loc_{m},\zone_{m},\dvalSeqi{m})\\
 \in& \Path(\WSTTS, \loc',\zone',\dvalSeq['], \loc_m,\zone_m,\dvalSeqi{m})\\
  \symbolicPathValue(\WSTTSPath) =&  \symbolicPathValue(\WSTTSPath') \semiringTimes \symbolicPathValue(\WSTTSPath'')
\end{align*}

 Assume for any $h \in \{1,2,\dots,m\}$, there is $\cval_{h} \in \zone_h$ satisfying $\cval_{h}(\absClock) \neq \sum_{k=1}^{j} \tau_k$.
 Let $h \in \{1,2,\dots,m-1\}$ such that
 for any $\cval_{h} \in \zone_h$, we have $\cval_{h}(\absClock) < \sum_{k=1}^{j} \tau_k$ and
 for any $\cval_{h+1} \in \zone_{h+1}$, we have $\cval_{h+1}(\absClock) > \sum_{k=1}^{j} \tau_k$.
 Since $\bigl( (\loc_{h},\zone_{h},\dvalSeqi{h}), (\loc_{h+1},\zone_{h+1},\dvalSeqi{h+1}) \bigr) \in \WSTTSTransition$, 
 there is $g\in \{1,2,\dots,n\}$, satisfying
 $\loc_{h}=\loc_{h+1}$, 
 $\dvalSeqi{h+1} = \dvalSeqi{h}\absConcat\values(\signal([\cval_{h}(\absClock),\cval_{h+1}(\absClock))))$, and
 $\zone_{h+1} = \{\cval + \relativeTime \mid \cval \in \zone, \relativeTime\in\Rp\} \cap M_{g}$,
 where $\cval_{h}\in\zone_{h},\cval_{h+1}\in\zone_{h+1}$ and 
 $M_{g} = \{ \cval \mid \sum_{k=0}^{g-1} \tau_k < \cval(\absClock) < \sum_{k=0}^{g} \tau_k\}$.
 Let $\zone' = \{\cval + \relativeTime \mid \cval \in \zone, \relativeTime\in\Rp\} \cap \{ \cval \mid \sum_{k=0}^{j} \tau_k = \cval(\absClock)\}$ and
 $\dvalSeq['] = \dvalSeqi{h}\absConcat\values(\signal([\cval_{h}(\absClock),\sum_{k=0}^{j} \tau_k)))$.
 Since 
 $\bigl( (\loc_{h},\zone_{h},\dvalSeqi{h}), (\loc_{h},\zone',\dvalSeq[']) \bigr) \in \WSTTSTransition$ and
 $\bigl( (\loc_{h+1},\zone',\dvalSeq[']), (\loc_{h+1},\zone_{h+1},\dvalSeqi{h+1}) \bigr) \in \WSTTSTransition$ hold, 
 we have the following.
 \begin{align*}
  \WSTTSPath' =& (\loc_0,\{\zerovalue[\ClockWithAbs]\},\dvalSeqi{0}),(\loc_1,\zone_1,\dvalSeqi{1}),\dots,(\loc_{h},\zone_{h},\dvalSeqi{h}),(\loc_{h},\zone',\dvalSeq['])\\
  \in& \Path(\WSTTS, \loc_0,\{\zerovalue[\ClockWithAbs]\},\dvalSeqi{0}, \loc_{h'},\zone',\dvalSeq['])*\\
  \WSTTSPath'' =& (\loc_{h+1},\zone',\dvalSeq[']),(\loc_{h+1},\zone_{h+1},\dvalSeqi{h+1}),\dots,(\loc_{m},\zone_{m},\dvalSeqi{m})\\
  \in& \Path(\WSTTS, \loc_{h+1},\zone',\dvalSeq['], \loc_m,\zone_m,\dvalSeqi{m})\\
  \symbolicPathValue(\WSTTSPath) =&  \symbolicPathValue(\WSTTSPath') \semiringTimes \symbolicPathValue(\WSTTSPath'')
 \end{align*}
 \qed
\end{proof}

\begin{lemma}
 \label{lemma:incremental_and_non_incremental}
 For a TSWA $\TSWA$, a signal $\signalWithInside$, and $i\in\{1,2,\dots,n\}$,  we have the following.
 \begin{align*}
  \incrementalWeight_i = \bigl\{
  (\loc_i,\zone_i,\dvalSeqi{i},\semiringElem_i) \,\bigm|\, &
  \loc_i \in \Loc,
  \zone_i \in \ZonesWithAbs,
  \dvalSeqi{i} \in \DValSeq,\\
  &
  \semiringElem_i  =
  \bigoplus_{
  \substack{\initLoc\in\InitLoc\\
  \WSTTSPath\in\Path(\WSTTS,\initLoc,\zerovalue{\ClockWithAbs},\varepsilon,\loc_i,\zone_i,\dvalSeqi{i})
  }}
  \symbolicPathValue(\WSTTSPath)\\
  &
  \forall \nu'\in\zone'.\, \nu'(\absClock)= \sum_{j=1}^{i} \tau_{j}
  \bigr\}
 \end{align*} 
\end{lemma}
\begin{proof}
 We prove by induction on $i$.
 If $i = 1$, we have the following.
 \begin{align*}
  \incrementalWeight_1 =& \incrementFunc\bigl(a_{1}, \tau_{1} \bigr)(\{(\initLoc,\zerovalue[\ClockWithAbs],\varepsilon,\semiringTimesUnit)\mid \initLoc \in\InitLoc\})\\
 =& \{(\loc',\zone',\dvalSeq['],\semiringElem')\mid 
  \loc'\in\Loc, 
  \zone'\in\ZonesWithAbs,
  \dvalSeq[']\in\DValSeq,\\
  &\qquad
  \semiringElem' =
  \bigoplus_{
  \substack{
  \initLoc \in\InitLoc\\
  \WSTTSPath\in\Path(\WSTTS_{a_1,\tau_1},\initLoc,\zerovalue[\ClockWithAbs],\varepsilon,\loc',\zone',\dvalSeq['])
  }}
  \semiringTimesUnit \semiringTimes \symbolicPathValue(\WSTTSPath),\\
  &\qquad
  \forall \nu'\in\zone'.\, \nu'(\absClock)= \tau_1
  \}  \\
 =& \{(\loc',\zone',\dvalSeq['],\semiringElem')\mid 
  \loc'\in\Loc, 
  \zone'\in\ZonesWithAbs,
  \dvalSeq[']\in\DValSeq,\\
  &\qquad
  \semiringElem' =
  \bigoplus_{
  \substack{
  \initLoc \in\InitLoc\\
  \WSTTSPath\in\Path(\WSTTS,\initLoc,\zerovalue[\ClockWithAbs],\varepsilon,\loc',\zone',\dvalSeq['])
  }}
  \semiringTimesUnit \semiringTimes \symbolicPathValue(\WSTTSPath),\\
  &\qquad
  \forall \nu'\in\zone'.\, \nu'(\absClock)= \tau_1
  \}  
 \end{align*}

 If $i > 1$, by induction hypothesis, we have the following.
 \begin{align*}
  \incrementalWeight_{i} =& \incrementFunc(a_{i}, \tau_{i})(\incrementalWeight_{i-1})\\
 =& \bigl\{(\loc',\zone',\dvalSeq['],\semiringElem')\mid 
  \loc'\in\Loc, 
  \zone'\in\ZonesWithAbs,
  \dvalSeq[']\in\DValSeq,\\
  &\qquad
  \semiringElem' =
  \bigoplus_{
  \substack{
  (\loc,\zone,\dvalSeq,\semiringElem) \in \incrementalWeight_{i-1}\\
  \WSTTSPath\in\Path(\WSTTS_{a_{i},\tau_{i}},\loc,\zone,\dvalSeq,\loc',\zone',\dvalSeq['])
  }}
  \semiringElem \semiringTimes \symbolicPathValue(\WSTTSPath),\\
  &\qquad
  \forall \nu'\in\zone'.\, \nu'(\absClock)= \sum_{j=1}^{i} \tau_j
  \bigr\}  \\
 =& \bigl\{(\loc',\zone',\dvalSeq['],\semiringElem')\mid
  \loc'\in\Loc, 
  \zone'\in\ZonesWithAbs,
  \dvalSeq[']\in\DValSeq,\\
  &\qquad
  \semiringElem' =
  \bigoplus_{
  \substack{
  \initLoc \in\InitLoc,
  \loc\in\Loc, 
  \zone\in\ZonesWithAbs,
  \dvalSeq\in\DValSeq,\\
  \WSTTSPath\in\Path(\WSTTS_{a_i,\tau_i},\loc,\zone,\dvalSeq,\loc',\zone',\dvalSeq['])\\
  \WSTTSPath'\in\Path(\WSTTS,\initLoc,\zerovalue[\ClockWithAbs],\varepsilon,\loc,\zone,\dvalSeq)\\
  \forall \nu\in\zone.\, \nu(\absClock)= \sum_{j=1}^{i-1} \tau_j
  }}
  \symbolicPathValue(\WSTTSPath') \semiringTimes \symbolicPathValue(\WSTTSPath),\\
  &\qquad
  \forall \nu'\in\zone'.\, \nu'(\absClock)= \sum_{j=1}^{i} \tau_j
  \bigr\}\\
 =& \{(\loc',\zone',\dvalSeq['],\semiringElem')\mid 
  \loc'\in\Loc, 
  \zone'\in\ZonesWithAbs,
  \dvalSeq[']\in\DValSeq,\\
  &\qquad
  \semiringElem' =
  \bigoplus_{
  \substack{
  \initLoc \in\InitLoc\\
  \WSTTSPath\in\Path(\WSTTS,\initLoc,\zerovalue[\ClockWithAbs],\varepsilon,\loc',\zone',\dvalSeq['])
  }}
  \symbolicPathValue(\WSTTSPath),\\
  &\qquad
  \forall \nu'\in\zone'.\, \nu'(\absClock)= \sum_{j=1}^{i} \tau_j
  \}  
 \end{align*}
 \qed
\end{proof}

\begin{proof}
 [\cref{theorem:incremental_correctness}]
 By \cref{lemma:incremental_and_non_incremental},
 we have the following.
 \begin{align*}
&\bigoplus_{\substack{
  (\loc,\zone,\dvalSeq) \in \WSTTSAccState\\
  (\loc,\zone,\dvalSeq,\semiringElem)\in\incrementalWeight_n
  }} \semiringElem  \\
=& \bigoplus_{\substack{
  (\loc,\zone,\dvalSeq) \in \WSTTSInitState\\
  (\loc',\zone',\dvalSeq') \in \WSTTSAccState\\
  \WSTTSPath \in \Path(\WSTTS, \loc,\zone,\dvalSeq, \loc',\zone',\dvalSeq')
  }} \symbolicPathValue(\WSTTSPath)  \\
=& \bigoplus_{
  \WSTTSPath \in \ARun(\WSTTS)
  } \symbolicPathValue(\WSTTSPath)  \\
=&\symbolicTraceValue(\WSTTS)
 \end{align*}
 \qed
\end{proof}

\section{Performance comparison between the benchmarks}
\label{appendix:comparizon_performance_benchmarks}

In \cref{fig:result_long_signal}, 
we observe that the execution time and memory usage of \textsc{Ringing} are higher than those of \textsc{Overshoot}. This is because the TSA of \textsc{Ringing} of is more complex than that of \textsc{Overshoot}: it has more states and clock variables, and it contains a loop.
We also observe that for \textsc{Ringing}, the execution time for the tropical semiring is shorter. This is because staying at the locations with $\top$ minimizes the weight for tropical semiring, and we need less exploration.

\section{Detailed experimental results}

 \begin{table*}[tb]
 \centering
 \caption{Execution time and memory usage under long signals for \textsc{Overshoot} and \textsc{Ringing} for sup-inf semiring}
  \label{table:result_long_signal}
  \scalebox{0.75}{
   \begin{filecontents*}{./dat/qtpm-all-BrCCPulse.tsv}
60000	1.591	7196.2	13.0625	7949.2
120000	3.193	7143.6	26.173	7924
180000	4.8475	7197.8	39.032	7960
240000	6.4405	7160.6	52.059	7977.6
300000	8.0655	7165.6	65.19	7912.2
360000	9.708	7147.2	78.133	7953.6
420000	11.398	7197.8	91.3655	7961.2
480000	13.007	7195.6	104.121	7933.2
540000	14.603	7160	117.492	7966.4
600000	16.282	7202.8	130.998	7972.4
\end{filecontents*}
\pgfplotstabletypeset[
   sci,
   sci zerofill,
   multicolumn names, 
   columns={[index]0,[index]1,[index]2,[index]3,[index]4},
   display columns/0/.style={
   column name=\begin{tabular}{c}$|\sigma|$\end{tabular}, 
   fixed,fixed zerofill,precision=0,
   },
   display columns/1/.style={
   column name=\begin{tabular}{c}Execution Time [s]\\\textsc{(Overshoot)}\end{tabular},
   precision=2},
   display columns/2/.style={
   fixed,precision=2,
   column name=\begin{tabular}{c}Memory Usage [KiB]\\\textsc{(Overshoot)}\end{tabular}},
   display columns/3/.style={
   column name=\begin{tabular}{c}Execution Time [s]\\\textsc{(Ringing)}\end{tabular},
   precision=2},
   display columns/4/.style={
   fixed,precision=2,
   column name=\begin{tabular}{c}Memory Usage [KiB]\\\textsc{(Ringing)}\end{tabular}},
   fixed,fixed zerofill,precision=1,
   every head row/.style={
   before row={\toprule}, 
   after row={\midrule} 
   },
   string replace={0}{},
   empty cells with={$< 0.01$},
   every last row/.style={after row=\bottomrule}, 
   ]{./dat/qtpm-all-BrCCPulse.tsv}}
 \end{table*}
 \begin{table*}[tb]
 \centering
 \caption{Execution time and memory usage under long signals for \textsc{Overshoot (Unbounded)} for sup-inf semiring}
  \label{table:result_long_signal_unbounded}
  \scalebox{0.75}{
   \begin{filecontents*}{./dat/qtpm-all-BrCCPulseUnbounded.tsv}
1000	0.14	9050.6
2000	0.923	14633.8
3000	2.967	23667.8
4000	6.786	36452
5000	13.4645	52806.8
6000	23.2245	72639.4
7000	36.8625	96220.8
8000	55.63	123263
9000	79.233	153900
10000	112.149	188054
\end{filecontents*}
\pgfplotstabletypeset[
   sci,
   sci zerofill,
   multicolumn names, 
   columns={[index]0,[index]1,[index]2},
   display columns/0/.style={
   column name=\begin{tabular}{c}$|\sigma|$\end{tabular}, 
   fixed,fixed zerofill,precision=0,
   },
   display columns/1/.style={
   column name=\begin{tabular}{c}Execution Time [s]\end{tabular},
   precision=2},
   display columns/2/.style={
   fixed,precision=2,
   column name=\begin{tabular}{c}Memory Usage [KiB]\end{tabular}},
   fixed,fixed zerofill,precision=1,
   every head row/.style={
   before row={\toprule}, 
   after row={\midrule} 
   },
   string replace={0}{},
   empty cells with={$< 0.01$},
   every last row/.style={after row=\bottomrule}, 
   ]{./dat/qtpm-all-BrCCPulseUnbounded.tsv}}
 \end{table*}
 \begin{table*}[tb]
 \centering
 \caption{Execution time and memory usage under high frequency for \textsc{Overshoot} and \textsc{Ringing} for sup-inf semiring}
  \label{table:result_high_freq}
 \scalebox{0.75}{
  \pgfplotstableset{create on use/freq/.style={create
  col/expr={6000/\thisrow{0}}}}
  \begin{filecontents*}{./dat/qtpm-all-BrCCPulseDensity.tsv}
60000	1.603	7169.4	13.0105	7943
30000	4.85	7518	45.6315	9906.2
20000	9.9335	8234.8	110.442	14688.8
15000	16.389	9277	222.278	18209
12000	25.3215	10723	378.38	25469.2
10000	37.903	13019.6	634.946	36310.2
8.571429e+03	53.517	16250.6	940.809	50715.4
7500	70.982	20134.8	1421.05	56401.8
6.666667e+03	96.6375	25429	1977.94	80871
6000	122.297	31358.6	2565.19	91547.2
\end{filecontents*}
\pgfplotstabletypeset[
   sci,
   sci zerofill,
   multicolumn names, 
   columns={freq,[index]1,[index]2,[index]3,[index]4},
   display columns/0/.style={
   column name=\begin{tabular}{c}Sampling Freq. [Hz]\end{tabular}, 
   fixed,fixed zerofill,precision=1,
   },
   display columns/1/.style={
   column name=\begin{tabular}{c}Execution Time [s]\\\textsc{(Overshoot)}\end{tabular},
   precision=2},
   display columns/2/.style={
   fixed,precision=2,
   column name=\begin{tabular}{c}Memory Usage [KiB]\\\textsc{(Overshoot)}\end{tabular}},
   display columns/3/.style={
   column name=\begin{tabular}{c}Execution Time [s]\\\textsc{(Ringing)}\end{tabular},
   precision=2},
   display columns/4/.style={
   fixed,precision=2,
   column name=\begin{tabular}{c}Memory Usage [KiB]\\\textsc{(Ringing)}\end{tabular}},
   every head row/.style={
   before row={\toprule}, 
   after row={\midrule} 
   },
   string replace={0}{},
   empty cells with={$< 0.01$},
   every last row/.style={after row=\bottomrule}, 
   ]{./dat/qtpm-all-BrCCPulseDensity.tsv}}
 \end{table*}

 \begin{table*}[tb]
 \centering
 \caption{Execution time and memory usage under long signals for \textsc{Overshoot} and \textsc{Ringing} for tropical semiring}
  \label{table:result_long_signal_tropical}
  \scalebox{0.75}{
   \begin{filecontents*}{./dat/qtpm-all-BrCCPulse-tropical.tsv}
60000	1.915	7292.8	6.6745	7678.8
120000	3.864	7273	13.1665	7680
180000	5.826	7278	19.7905	7710.8
240000	7.759	7315.4	26.299	7704
300000	9.7495	7304.8	32.8805	7703.2
360000	11.689	7287.8	39.3455	7709.4
420000	13.6705	7288	46.114	7712
480000	15.644	7281.8	52.725	7737.2
540000	17.5705	7283.4	59.1955	7728
600000	19.4785	7266	65.6735	7737.6
\end{filecontents*}
\pgfplotstabletypeset[
   sci,
   sci zerofill,
   multicolumn names, 
   columns={[index]0,[index]1,[index]2,[index]3,[index]4},
   display columns/0/.style={
   column name=\begin{tabular}{c}$|\sigma|$\end{tabular}, 
   fixed,fixed zerofill,precision=0,
   },
   display columns/1/.style={
   column name=\begin{tabular}{c}Execution Time [s]\\\textsc{(Overshoot)}\end{tabular},
   precision=2},
   display columns/2/.style={
   fixed,precision=2,
   column name=\begin{tabular}{c}Memory Usage [KiB]\\\textsc{(Overshoot)}\end{tabular}},
   display columns/3/.style={
   column name=\begin{tabular}{c}Execution Time [s]\\\textsc{(Ringing)}\end{tabular},
   precision=2},
   display columns/4/.style={
   fixed,precision=2,
   column name=\begin{tabular}{c}Memory Usage [KiB]\\\textsc{(Ringing)}\end{tabular}},
   fixed,fixed zerofill,precision=1,
   every head row/.style={
   before row={\toprule}, 
   after row={\midrule} 
   },
   string replace={0}{},
   empty cells with={$< 0.01$},
   every last row/.style={after row=\bottomrule}, 
   ]{./dat/qtpm-all-BrCCPulse-tropical.tsv}}
 \end{table*}
 \begin{table*}[tb]
 \centering
 \caption{Execution time and memory usage under long signals for \textsc{Overshoot (Unbounded)} for tropical semiring}
  \label{table:result_long_signal_unbounded_tropical}
  \scalebox{0.75}{
   \begin{filecontents*}{./dat/qtpm-all-BrCCPulseUnbounded-tropical.tsv}
1000	0.1805	9799.4
2000	1.161	16840.4
3000	3.59	28256.2
4000	8.217	44139.4
5000	15.7545	64509.4
6000	27.3435	89218.4
7000	43.3285	118498
8000	64.287	152075
9000	90.744	190065
10000	122.865	232474
\end{filecontents*}
\pgfplotstabletypeset[
   sci,
   sci zerofill,
   multicolumn names, 
   columns={[index]0,[index]1,[index]2},
   display columns/0/.style={
   column name=\begin{tabular}{c}$|\sigma|$\end{tabular}, 
   fixed,fixed zerofill,precision=0,
   },
   display columns/1/.style={
   column name=\begin{tabular}{c}Execution Time [s]\end{tabular},
   precision=2},
   display columns/2/.style={
   fixed,precision=2,
   column name=\begin{tabular}{c}Memory Usage [KiB]\end{tabular}},
   fixed,fixed zerofill,precision=1,
   every head row/.style={
   before row={\toprule}, 
   after row={\midrule} 
   },
   string replace={0}{},
   empty cells with={$< 0.01$},
   every last row/.style={after row=\bottomrule}, 
   ]{./dat/qtpm-all-BrCCPulseUnbounded-tropical.tsv}}
 \end{table*}
 \begin{table*}[tb]
 \centering
 \caption{Execution time and memory usage under high frequency for \textsc{Overshoot} and \textsc{Ringing} for tropical semiring}
  \label{table:result_high_freq_tropical}
 \scalebox{0.75}{
  \pgfplotstableset{create on use/freq/.style={create
  col/expr={6000/\thisrow{0}}}}
  \begin{filecontents*}{./dat/qtpm-all-BrCCPulseDensity-tropical.tsv}
60000	1.9135	7274.2	6.6595	7646.4
30000	6.936	7704.8	11.6855	7910.2
20000	14.576	8547.2	16.5665	8302.6
15000	24.3575	9836.8	21.742	8666.2
12000	37.953	11905.8	28.7535	8855.8
10000	57.4885	14764.2	32.9325	9730.8
8.571429e+03	81.734	18501.6	41.4135	9831.8
7500	110.034	23484	52.33	10307
6.666667e+03	149.967	29905.4	63.6255	11291.2
6000	196.831	37465.2	66.996	11443.8
\end{filecontents*}
\pgfplotstabletypeset[
   sci,
   sci zerofill,
   multicolumn names, 
   columns={freq,[index]1,[index]2,[index]3,[index]4},
   display columns/0/.style={
   column name=\begin{tabular}{c}Sampling Freq. [Hz]\end{tabular}, 
   fixed,fixed zerofill,precision=1,
   },
   display columns/1/.style={
   column name=\begin{tabular}{c}Execution Time [s]\\\textsc{(Overshoot)}\end{tabular},
   precision=2},
   display columns/2/.style={
   fixed,precision=2,
   column name=\begin{tabular}{c}Memory Usage [KiB]\\\textsc{(Overshoot)}\end{tabular}},
   display columns/3/.style={
   column name=\begin{tabular}{c}Execution Time [s]\\\textsc{(Ringing)}\end{tabular},
   precision=2},
   display columns/4/.style={
   fixed,precision=2,
   column name=\begin{tabular}{c}Memory Usage [KiB]\\\textsc{(Ringing)}\end{tabular}},
   every head row/.style={
   before row={\toprule}, 
   after row={\midrule} 
   },
   string replace={0}{},
   empty cells with={$< 0.01$},
   every last row/.style={after row=\bottomrule}, 
   ]{./dat/qtpm-all-BrCCPulseDensity-tropical.tsv}}
 \end{table*}

The detailed experimental results are summarized in \crefrange{table:result_long_signal}{table:result_high_freq_tropical}.}
\end{document}